\documentclass[acmsmall,screen,nonacm]{acmart}
\usepackage{graphicx} %
\usepackage{pgfplots}

\usepackage{xcolor}
\definecolor{Cayenne}{HTML}{941100}
\definecolor{Mocha}{HTML}{945200}
\definecolor{Asparagus}{HTML}{929000}
\definecolor{Fern}{HTML}{4F8F00}
\definecolor{Clover}{HTML}{008F00}
\definecolor{Moss}{HTML}{009051}
\definecolor{Teal}{HTML}{009193}
\definecolor{Ocean}{HTML}{005493}
\definecolor{Midnight}{HTML}{011993}
\definecolor{Eggplant}{HTML}{011993}
\definecolor{Plum}{HTML}{942193}
\definecolor{Maroon}{HTML}{942193}

\usepackage{wrapfig}
\usepackage{mathpartir}
\usepackage{caption}
\usepackage{subcaption}
\usepackage{amsmath}
\usepackage{adjustbox}
\usepackage{syntax}
\usepackage{stmaryrd}
\usepackage{relsize}
\usepackage{algpseudocode}
\usepackage{makecell}
\usepackage[underline=false]{pgf-umlsd}
\usepackage{paralist}
\usepackage{transparent}
\usepackage{booktabs}

\microtypesetup{letterspace=-100}

\usepackage{titlesec}
\titleformat{\subparagraph}[runin]{\normalfont\normalsize\itshape}{\thesubparagraph}{1em}{}
\titlespacing*{\subparagraph}{\parindent}{0pt}{3.5pt}

\usepackage{tikz}
\usetikzlibrary{automata, arrows.meta, positioning, patterns}

\newcommand{\reftxt}[2]{\hyperref[#2]{#1~\ref*{#2}}}

\newcommand{\exreftext}{Example}
\newcommand{\exref}[1]{\reftxt{\exreftext}{#1}}

\usepackage{listings}

\lstdefinelanguage{P}{
  morekeywords={
    machine, event, fun, var, type, enum, spec, module, test, 
    start, state, entry, exit, on, goto, do, defer, ignore, raise, send, broadcast,
    new, assert, assume, if, else, while, foreach, in, sizeof, return, true, false, null,
    int, seq
  },
  sensitive=true,
  morecomment=[l]{//},
  morecomment=[s]{/*}{*/},
  morestring=[b]",
  morestring=[b]',
  alsoletter={:}
}
\newcommand{\lstbasicstyle}{\ttfamily\ttfamily\small}
\newcommand{\lstsmallerstyle}{\ttfamily\ttfamily\tiny}
\newcommand{\lstkeywordstyle}{\color{Cayenne}\bfseries}
\newcommand{\lstcommentstyle}{\sl}
\newcommand{\lstnumberstyle}{\scriptsize\em}

\lstdefinestyle{default}{%
  backgroundcolor=\color{white},%
  basicstyle=\lstbasicstyle,%
  commentstyle=\lstcommentstyle,%
  keywordstyle=\lstkeywordstyle,%
  columns=fullflexible,%
  keepspaces=true,%
  mathescape=true%
}
\lstdefinestyle{number}{%
  numbers=left,%
  numberstyle=\lstnumberstyle,%
  xleftmargin=2em%
}
\makeatletter
\newcount\lst@savelstnumber
\newcommand{\lstnonum}{%
  \global\lst@savelstnumber\c@lstnumber%
  \gdef\thelstnumber{}%
}
\newcommand{\lststopn}{%
  \lstnonum%
}
\newcommand{\lststartn}{%
  \global\c@lstnumber\lst@savelstnumber%
  \gdef\thelstnumber{\arabic{lstnumber}}%
}
\newcommand{\lstsetn}[1]{%
  \global\c@lstnumber#1%
  \gdef\thelstnumber{\arabic{lstnumber}}%
}
\newcommand{\lstbeginn}{\lstsetn{0}}
\newcommand{\lstadvancen}{%
  \advance\c@lstnumber 1\relax%
}
\makeatother

\lstdefinestyle{mystyle}
{
    basicstyle=\small,
    frame=single,
    breaklines,
    columns=fullflexible,
    breakindent=1.2em,
    breakatwhitespace,
    escapeinside={(*}{*)},
    numbers=left,
    language=C,
    linewidth=0.44\linewidth 
}
\lstdefinestyle{numberstyle}{ 
     basicstyle=\small,
    columns=fullflexible,
    language=C,
    escapeinside={(*}{*)},
    numberstyle=\tiny\color{gray},
    numbers=left
}
\lstdefinestyle{simplestyle}{ 
     basicstyle=\small,
    columns=fullflexible,
    language=C,
    escapeinside={(*}{*)}
}

\definecolor{lightredbg}{rgb}{1.0,0.93,0.93}
\definecolor{lightgreenbg}{rgb}{0.92,1.0,0.92}
\definecolor{lightyellowbg}{rgb}{1.0,1.0,0.88}
\lstnewenvironment{lstdiffminus}[1][]
 {\lstset{aboveskip=-\medskipamount,name=lstdiff,firstnumber=last,backgroundcolor=\color{lightredbg},#1}}
 {}
\lstnewenvironment{lstdiffplus}[1][]
 {\lstset{aboveskip=-\medskipamount,name=lstdiff,firstnumber=last,backgroundcolor=\color{lightgreenbg},#1}}
 {}
\lstnewenvironment{lstcont}[1][]
 {\lstset{aboveskip=-\medskipamount,firstnumber=last,#1}}
 {}

\newboolean{showcomments}
\setboolean{showcomments}{true}
\newboolean{showjedi}
\setboolean{showjedi}{false}
\newboolean{shownote}
\setboolean{shownote}{false}
\newcommand{\EDITCOMMENT}[2][red]{\ifthenelse{\boolean{showcomments}}{{\color{#1}$\bigl[$\bgroup\em #2\egroup$\bigr]$}}{}}
\newcommand{\NOTE}[1]{\ifthenelse{\boolean{shownote}}{\EDITCOMMENT[blue]{#1}}{}}

\newcommand{\JEDI}[1]{\ifthenelse{\boolean{showjedi}}{\EDITCOMMENT[cyan]{#1}}{}}

\title{DissProve: Automated Verification of Asynchronous Distributed Protocols with Affine Communication}
\author{Christian Fontenot}
\email{christian.fontenot@colorado.edu}
\affiliation{
    \institution{University of Colorado Boulder}
    \city{Boulder}
    \state{Colorado}
    \country{USA}
}

\author{Gowtham Kaki}
\affiliation{
   \institution{University of Colorado Boulder}
    \city{Boulder}
    \state{Colorado}
    \country{USA}
}

\author{Bor-Yuh Evan Chang}
\authornote{Bor-Yuh Evan Chang holds concurrent appointments at the University of Colorado Boulder and as an Amazon Scholar.  This paper describes work performed at CU Boulder and is not associated with Amazon.}
\affiliation{
   \institution{University of Colorado Boulder, \& Amazon}
    \country{USA}
}
\date{2026}
\begin{document}
\newcommand{\denote}[1]{\ensuremath{\llbracket#1\rrbracket}}
\newcommand{\tool}{\textsc{DissProve}}
\newcommand{\actlang}{\textsc{ActL}}
\newcommand{\Hactor}{DiSM}
\newcommand{\dCSNH}{dCSNH}
\newcommand{\todo}[1]{\textcolor{magenta}{[#1]}}
\newcommand{\note}[1]{\textcolor{blue}{[#1]}}
\newcommand{\papertext}[1]{#1}

\newcommand{\hole}{\textcolor{red}{???}}
\newcommand{\concStep}[2]{#1 \rightarrow #2}
\newcommand{\manyConcrete}[2]{#1 \rightarrow^* #2}
\newcommand{\localStep}[3]{(#1,#2) \Downarrow #3}

\newcommand{\evalExpr}[3]{(#1,#2) \Downarrow #3}
\newcommand{\defset}[1]{\{#1\}}
\newcommand{\concreteState}{\mathcal{P}}

\newcommand{\locAbStep}[3]{\{#2\}#1\{#3\}}
\newcommand{\causalStep}[2]{#2\leftsquigarrow#1}
\newcommand{\traceExt}[2]{#1\Rightarrow #2}
\newcommand{\traceLU}[3]{#1\vdash_{#2}#3}
\newcommand{\manyCausal}[2]{#2\leftsquigarrow^*#1}

\newcommand{\varfrom}[3]{#1\hookrightarrow (#2,#3)}
\newcommand{\prevmsg}[2]{#1:\Vdash#2}
\newcommand{\curmsg}[2]{#1:\vdash#2}
\newcommand{\hsends}[3]{#1\rightharpoonup_{#3} #2}
\newcommand{\hassigns}[3]{#1\rightharpoondown_{#3} #2}
\newcommand{\vbind}[4]{(#1,#2)\twoheadrightarrow_{#3} #4}

\newcommand{\infrule}[1]{\textsc{#1}}
\newcommand{\code}[1]{\texttt{#1}}

\newcommand{\actorsys}{Actor History Transition System}
\newcommand{\system}{\textbf{sys}}
\newcommand{\initmsg}{$init$}
\newcommand{\symbolictrace}{\hat{T}}
\newcommand{\symbolicstore}{\hat{\sigma}}
\newcommand{\symbolicstate}{\mathcal{A}}
\newcommand{\symbolicactor}{s}
\newcommand{\actorlogic}{\text{PAHL}}

\newcommand{\recmap}{\hat{\psi}}
\newcommand{\trace}{T ::= \circ | m::T}

\newcommand{\bnfdef}{::=}
\newcommand{\bnfalt}{\mid}
\newcommand{\bnfopt}[1]{\ulcorner #1 \lrcorner}

\newcommand{\defeq}{\stackrel{\text{\tiny def}}{=}}
\newcommand{\seq}[1]{\smash{\overline{\vphantom{|}#1}}}
\newcommand{\dom}{\operatorname{dom}}
\newcommand{\lfp}{\operatorname{lfp}}

\newcommand{\ctxnil}{\circ}
\newcommand{\ctxcons}[2]{#1, #2}
\newcommand{\ctxupd}[2]{#1[#2]}
\newcommand{\mapping}[2]{#1 \mapsto #2}
\newcommand{\maplookup}[2]{#1(#2)}

\newcommand{\tuple}[1]{\langle #1 \rangle}

\newcommand{\sublbl}[2]{#1_{\text{\relsize{-1}\lsstyle#2}}}
\newcommand{\fromlbl}[1]{\sublbl{#1}{from}}
\newcommand{\tolbl}[1]{\sublbl{#1}{to}}
\newcommand{\prelbl}[1]{\sublbl{#1}{pre}}
\newcommand{\postlbl}[1]{\sublbl{#1}{post}}
\newcommand{\fieldslbl}[1]{\sublbl{#1}{mach}}
\newcommand{\varslbl}[1]{\sublbl{#1}{hand}}

\newcommand{\fmtcode}[1]{\ensuremath{\text{\lstbasicstyle #1}}}
\newcommand{\sfmtcode}[1]{\ensuremath{\text{\lstsmallerstyle #1}}}
\newcommand{\fmtmvar}[1]{\mathit{#1}}
\newcommand{\fmtkw}[1]{\text{\relsize{-0.5}$\mathtt{#1}$}}
\newcommand{\fmtspecialvar}[1]{\text{\relsize{-1}\lsstyle$\mathtt{#1}$}}

\newcommand{\fmtabs}[1]{\widehat{#1}}

\newcommand{\lemp}{\mathsf{emp}}
\newcommand{\lsep}{\mathbin{\ast}}
\newcommand{\ltrue}{\top}
\newcommand{\lorderimpr}{\mathrel{\twoheadrightarrow}}

\newcommand{\var}{x}
\newcommand{\varalt}{y}

\newcommand{\cmd}{\fmtmvar{cmd}}
\newcommand{\sendkw}{\fmtkw{send}}
\newcommand{\bcastkw}{\fmtkw{broadcast}}
\newcommand{\mkbcast}[3]{\bcastkw\:#1.#2(#3)}
\newcommand{\mksend}[3]{\sendkw\:#1.#2(#3)}
\newcommand{\gotokw}{\fmtkw{goto}}
\newcommand{\mkgoto}[1]{\gotokw\:#1}
\newcommand{\newkw}{\fmtkw{new}}
\newcommand{\mknew}[1]{\newkw\:#1}

\newcommand{\instr}{\fmtmvar{instr}}
\newcommand{\skipkw}{\fmtkw{skip}}
\newcommand{\assignkw}{\fmtkw{=}}
\newcommand{\mkassign}[2]{#1 \mathrel{\assignkw} #2}
\newcommand{\assertkw}{\fmtkw{assert}}
\newcommand{\mkassert}[1]{\assertkw\:#1}
\newcommand{\expr}{e}
\newcommand{\mkread}[2]{#1.#2}

\newcommand{\trans}{t}
\newcommand{\mktrans}[3]{#1 \mathrel{\mathord{-}\mkern-4mu[#2]\mkern-4mu\mathord{\rightarrow}} #3}

\newcommand{\loc}{\mathscr{l}}
\newcommand{\entryloc}{\fmtspecialvar{entry}}
\newcommand{\exitloc}{\fmtspecialvar{exit}}

\newcommand{\proc}{\fmtmvar{proc}}

\newcommand{\hand}{\fmtmvar{hand}}
\newcommand{\onkw}{\fmtkw{on}}
\newcommand{\dokw}{\fmtkw{do}}
\newcommand{\evnt}{\mathscr{e}}
\newcommand{\idleevnt}{\fmtspecialvar{eIdle}}
\newcommand{\initevnt}{\fmtspecialvar{eInit}}
\newcommand{\mkonhand}[4][]{\onkw\:#2\:\dokw^{#1}\,(#3)\,#4}
\newcommand{\mkonseghand}[4][\seglen]{\mkonhand[#1]{#2}{#3}{#4}}
\newcommand{\deferkw}{\fmtkw{defer}}
\newcommand{\mkdeferhand}[1]{\deferkw\:#1}
\newcommand{\ignorekw}{\fmtkw{ignore}}
\newcommand{\mkignorehand}[1]{\ignorekw\:#1}

\newcommand{\behave}{\fmtmvar{behave}}
\newcommand{\startkw}{\fmtkw{start}}
\newcommand{\statekw}{\fmtkw{state}}
\newcommand{\stt}{\mathscr{s}}
\newcommand{\initstt}{\fmtspecialvar{Init}}
\newcommand{\mkbehave}[2]{\statekw\:#1\:#2}

\newcommand{\field}{\fmtmvar{fld}}
\newcommand{\varkw}{\fmtkw{var}}
\newcommand{\mkfield}[2]{\varkw\:#1 \mathrel{\assignkw} #2}

\newcommand{\actor}{\fmtmvar{actor}}
\newcommand{\mach}{\mathscr{m}}
\newcommand{\absmach}{\mach}
\newcommand{\mainmachtext}{Main}
\newcommand{\mainmach}{\fmtspecialvar{\mainmachtext}}
\newcommand{\summarymachtext}{System}
\newcommand{\summarymach}{\fmtabs{\fmtspecialvar{\summarymachtext}}}
\newcommand{\machinekw}{\fmtkw{machine}}
\newcommand{\mkactor}[3]{\machinekw\:#1\:#2\:#3}

\newcommand{\msgkind}{\fmtmvar{messk}}
\newcommand{\eventkw}{\fmtkw{event}}
\newcommand{\mkmsgkind}[1]{\eventkw\:#1}

\newcommand{\topsys}{\pi}
\newcommand{\mktopsys}[2]{#2}

\newcommand{\val}{v}
\newcommand{\absval}{\fmtabs{\val}}
\newcommand{\addr}{a}
\newcommand{\absaddr}{\fmtabs{\addr}}
\newcommand{\symvar}{\fmtabs{\var}}
\newcommand{\symvaralt}{\fmtabs{\varalt}}

\newcommand{\storecolor}[1]{{\color{Cayenne} #1}}
\newcommand{\netcolor}[1]{{\color{Asparagus} #1}}
\newcommand{\histcolordecl}{\color{Midnight}}
\newcommand{\histcolor}[1]{{\histcolordecl #1}}
\newcommand{\purecolordecl}{\color{Moss}}
\newcommand{\purecolor}[1]{{\purecolordecl #1}}

\newcommand{\thmcolor}[1]{\textcolor{Plum}{#1}}

\newcommand{\store}{\rho}
\newcommand{\absstore}{\fmtabs{\store}}
\newcommand{\storetop}{\top}
\newcommand{\storesep}{\lsep}
\newcommand{\absstorecons}[2]{#1 \storesep #2}

\newcommand{\actorstore}{\eta}
\newcommand{\absactorstore}{\fmtabs{\actorstore}}
\newcommand{\mkactorstore}[2]{\langle #1, #2 \rangle}

\newcommand{\sysstore}{\zeta}
\newcommand{\sysstorecolored}{\storecolor{\sysstore}}
\newcommand{\abssysstore}{\fmtabs{\sysstore}}

\newcommand{\msg}{m}
\newcommand{\absmsg}{\fmtabs{\msg}}
\newcommand{\mkmsg}[4]{#1\mathord{\dashrightarrow}#2.#3(#4)}
\newcommand{\ctormsg}{\fmtspecialvar{ctor}}
\newcommand{\mkctor}[2]{#1\mathord{\dashrightarrow}#2.\ctormsg}

\newcommand{\seglen}{k}
\newcommand{\symseglen}{\fmtabs{k}}
\newcommand{\mkmsgseg}[5][\seglen]{#2\mathord{\dashrightarrow}#3.#4^{#1}(#5)}

\newcommand{\net}{\mu}
\newcommand{\netcolored}{\netcolor{\net}}
\newcommand{\absnet}{\fmtabs{\net}}
\newcommand{\nettop}{\top}
\newcommand{\netsep}{\otimes}
\newcommand{\netnil}{\ctxnil}
\newcommand{\netcons}[2]{\ctxcons{#1}{#2}}
\newcommand{\absnetnil}{\nettop}
\newcommand{\absnetnilx}{\sublbl{\absnetnil}{net}}
\newcommand{\absnetcons}[2]{#1 \netsep #2}

\newcommand{\lok}{\mathsf{ok}}

\newcommand{\hist}{\omega}
\newcommand{\histcolored}{\histcolor{\hist}}
\newcommand{\abshist}{\fmtabs{\hist}}
\newcommand{\histemp}{\varepsilon}
\newcommand{\histcons}[2]{#1; #2}
\newcommand{\abshistemp}{\lok}
\newcommand{\abshistempx}{\sublbl{\abshistemp}{hist}}
\newcommand{\abshisthyp}[2]{#1 \lorderimpr #2}

\newcommand{\actorloc}{\fmtmvar{l}}
\newcommand{\sysloc}{\fmtmvar{loc}}
\newcommand{\abssysloc}{\fmtabs{\sysloc}}
\newcommand{\mkmachsttevntloc}[4]{#1.#3.#4}

\newcommand{\loctop}{\top}
\newcommand{\loctopx}{\sublbl{\loctop}{loc}}
\newcommand{\locsep}{\mathrel{\lsep}}

\newcommand{\mkpstate}[2]{\tuple{#1\colon #2}}
\newcommand{\mkabspstate}[3][\puredom]{\mkpstate{#2}{#3 \mid \purecolor{#1}}}
\newcommand{\mkactorstate}[3]{\mkpstate{#1}{#2 \cdot #3}}

\newcommand{\sysstate}{\sigma}
\newcommand{\abssysstate}{\fmtabs{\sysstate}}
\newcommand{\mksysstate}[4]{\mkpstate{#1}{\storecolor{#2} \cdot \netcolor{#3} \cdot \histcolor{#4}}}

\newcommand{\puredom}{\varphi}
\newcommand{\mkabssysstate}[5][\puredom]{\mkpstate{#2}{\storecolor{#3} \cdot \netcolor{#4} \cdot \histcolor{#5} \mid \purecolor{#1}}}

\newcommand{\sysprop}{\Sigma}
\newcommand{\abssysprop}{\fmtabs{\sysprop}}

\newcommand{\mksummsg}[4]{\mathord{\dashrightarrow}_{#1}#2.#3(#4)}

\newcommand{\valuation}{\theta}
\newcommand{\concstate}[2]{#2 \models #1}
\newcommand{\concloc}[2]{#2 \models #1}
\newcommand{\concpure}[2][\valuation]{\purecolor{#1} \models #2}
\newcommand{\conc}[3][\valuation]{#3 \mathbin{\mid} \purecolor{#1} \models #2}
\newcommand{\initstate}[1]{\vdash #1\;\mathsf{init_{\bot}}}
\newcommand{\jinitbot}[1]{\initstate{#1}}
\newcommand{\jnotinit}[1]{\not\initstate{#1}}
\newcommand{\junreachable}[1]{\vdash #1\;\mathsf{unreachable}}

\newcommand{\causalgraphop}{\operatorname{causal}}
\newcommand{\causalgraph}[1]{\causalgraphop(#1)}
\newcommand{\causaltofrom}[2]{#1 \leftarrow #2}
\newcommand{\causaltofromdash}[2]{#1 \dashleftarrow #2}
\newcommand{\mkmachevnt}[2]{\histcolor{#1.#2}}
\newcommand{\mkmachfld}[2]{\storecolor{#1.#2}}

\newcommand{\exit}[4]{exit(#1,#2,#3,#4)}
\newcommand{\entry}[4]{entry(#1,#2,#3,#4)}
\newcommand{\pathext}[2]{#1 \twoheadrightarrow #2}
\newcommand{\sumpath}[2]{#1 \in path(#2)}
\newtheorem{axiom}{Axiom}

\newcommand{\msgtxt}[2]{\texttt{#1(#2)}}

\renewcommand{\mess}[4][0]{
  \stepcounter{seqlevel}
  \path
  (#2)+(0,-\theseqlevel*\unitfactor-0.7*\unitfactor) node (mess from) {};
  \addtocounter{seqlevel}{#1}
  \path
  (#4)+(0,-\theseqlevel*\unitfactor-0.7*\unitfactor) node (mess to) {};
  \draw[->,>=angle 60] (mess from) -- (mess to) node[midway, above]
  {#3};    
}

\newcommand{\ifmtcode}[1]{\fmtcode{\relsize{-1}#1}}
\newcommand{\iNode}{\ifmtcode{Node}}

\newcommand{\inode}[2][n]{#1_#2}
\newcommand{\iNodes}{\mathit{nodes}}
\newcommand{\iNodesIdx}[1]{\iNodes\fmtcode{[}#1\fmtcode{]}}
\newcommand{\ifmtfld}[1]{\ifmtcode{#1}}
\newcommand{\ifld}[2]{#1\fmtcode{.}\ifmtcode{#2}}
\newcommand{\ineq}{\mathrel{\purecolor{\neq}}}
\newcommand{\ieq}{\mathrel{\purecolor{=}}}
\newcommand{\ilor}{\mathrel{\purecolor{\lor}}}
\newcommand{\iland}{\mathrel{\purecolor{\land}}}
\newcommand{\iimplies}{\mathrel{\purecolor{\implies}}}
\newcommand{\noleader}{-1}
\newcommand{\inoleader}{\purecolor{\noleader}}

\tikzset{state/.style={rectangle, draw, minimum width=1.5cm, minimum height=1cm}}
\begin{abstract}

We consider the problem of automatically proving safety properties of distributed protocols.
Distributed protocols have been particularly challenging for automated verification due to their \emph{asynchronous} and \emph{parametric} nature.
Compared to synchronous systems, asynchronous communication leads to a combinatorial explosion of possible execution histories of message handlers. As distributed protocols are typically defined parametrically on the number of actors, these definitions lead to an unbounded number of possible execution histories of unbounded length.
Existing verification techniques for such distributed protocols typically require global invariants about the entire actor system, which are complex even for simple protocols.
In this paper, we present an automated verification technique based on proving unreachability backwards from error states in an actor system.
One key insight is that the unboundedness from parametricity can be further classified into \emph{affine} and non-affine protocols, where affine protocols have execution histories of unbounded length in a bounded number of communication rounds.
We show how to use novel, goal-directed notions of materialization, causality, and summarization to verify safety properties of affine protocols with an unbounded number of actors in an automated manner.
Using our verification tool \tool{}, we demonstrate the feasibility of
fully-automated safety verification of asynchronous distributed protocols,
including Two-Phase Commit and Leader Election.

\end{abstract}

\acmConference[]{}{2027}{}

\maketitle

\section{Introduction}\label{sec:introduction}

A distributed protocol describes how a collection of networked processes
coordinate to achieve a common goal --- electing a leader, committing a
transaction, or reaching consensus~\cite{paxos,twopc-grey06}. Such a protocol is
executed by an \emph{unbounded} number of processes that communicate by
exchanging messages over an \emph{asynchronous} network, where messages may be
delayed or reordered arbitrarily. This combination of unbounded concurrency and
asynchrony makes formal verification of distributed protocols notoriously hard:
the reachable state space is infinite, and the number of message interleavings
across processes grows without bound~\cite{statespace}. Despite decades of
research, fully automated formal verification of distributed protocols remains
elusive. 

In this work, we identify a large, non-trivial class of distributed protocols
for which we can automate safety property verification. We say a distributed
protocol has \emph{affine} communication if each process sends at most one
message of a given type to each receiver. This restriction bounds communication
in terms of the static number of message types, while still allowing an
unbounded message count from the number of processes. Despite this
unboundedness, a communication round must reach a quiescent point in which no
further messages of a given type are sent.  Distributed protocols are often
designed to communicate in rounds, where safety can be decomposed into
intra-round and inter-round safety~\cite{garcia2018paxos}. Following this
intuition, communication within a round is usually bounded modulo the number of
processes, and so a variety  of protocols that real distributed systems rely on,
including  Two-Phase Commit~\cite{twopc-grey06}, Lamport's Bakery
~\cite{bakery-cacm74}, and quorum-based consensus algorithms can be defined with
affine communication (for a single round).

The standard approach to safety verification is to overapproximate the set of
forward-reachable states using an inductive invariant and show that it excludes
the error states, i.e., states that violate the safety condition. In the context
of distributed protocols, inductive invariants relate the state of the
asynchronous network to the local states of a potentially unbounded number of
processes, e.g., via an array of processes of machine type $\iNode$. Automated
reasoning tools, such as Ivy~\cite{ivy, ivy-pldi18} and Civl~\cite{civl-cav20},
make an extensive use of such inductive invariants to facilitate SMT-aided
safety verification of distributed protocols. While effective, this approach is
not always practical as successful verification often requires sophisticated
invariants relating disparate components of the
system~\cite{verdi,padon2017paxos}.  For instance, the one-shot leader election
protocol that we discuss in \autoref{sec:motivation} requires six different
quantified invariants, some with nested quantification, capture the relationship
between the asynchronous network state and the local states of the
processes~\cite{lewchenko-oopsla25}. Working out the inductive invariants that
are sufficiently strong to imply the safety property is therefore hard, either manually or algorithmically.

Abstract interpretation~\cite{cousot} is a general framework for automatically
computing sound over-approximations of program behaviors as a fixed point of a
monotone function over a lattice of abstract states. This framework has been
used to automatically infer inductive invariants for sequential~\cite{nlcra,cra}
and shared-memory concurrent programs~\cite{mine-vmcai14}, removing the need for
manually supplied invariants. This yields a promising direction for automated
reasoning of distributed protocols. However, applying abstract interpretation to
distributed protocols poses distinct challenges beyond sequential or
shared-memory concurrent verification: processes communicate
\emph{asynchronously}, meaning messages may be arbitrarily delayed or reordered,
and the protocol involves an \emph{unbounded} number of processes each
potentially in a different local state. These are precisely the same factors
that make manually specifying inductive invariants for distributed protocols
hard, and they make it equally difficult to construct an abstract domain that
finitely represents the joint behavior of all processes and in-flight messages.

While asynchronous communication and an unbounded concurrency leads to a
super-exponential blowup in the reachable state space, not all of this state
space is relevant to \emph{refuting the violation} of a given safety property.
This observation has motivated \emph{backward reasoning}~\cite{paroshwqo}, which
starts from states that violate safety and attempts to prove their
\emph{unreachability}. Working backwards from error states naturally focuses
exploration on property-relevant states, which tend to admit more compact finite
representations. This approach has been applied successfully in sequential
verification by \textsc{Symbiotic\,9}~\cite{symbioticnine}, and to verify
array-based infinite-state systems by \textsc{Cubicle}~\cite{cubicle} and
\textsc{MCMT}~\cite{ghilardi}. Neither line of work, however, addresses the
combined challenge of asynchronous communication and an unbounded, heterogeneous
process population that characterizes distributed protocols. This raises a
natural question: can the same insight --- focusing on goal-relevant states by
starting from an error state --- be brought to bear on distributed protocol
verification despite these challenges?

In this paper, we answer this question in the affirmative --- for a natural and
practically relevant class of distributed protocols with \emph{affine}
communication. In particular, we design a goal-directed, backward-from-error
program logic for reasoning about parametric actor systems with asynchronous
communication. This logic is based on a symbolic abstraction of actor systems
using separation, linear, and ordered logics to represent actor-local state,
unbounded message buffers, and unbounded communication histories, respectively.
This abstraction enables compositional reasoning over systems with an unbounded
number of actors and messages. We define a backwards abstract interpretation
using this symbolic abstraction to consider the relevant slice of the system
necessary to refute error states.
Overall, we make the following technical contributions:
\begin{itemize}
    \item We build a novel symbolic abstraction of actor systems using
    separation, linear, and ordered logics and introduce \emph{actor-location
    materialization} as a means to enable goal-directed safety verification of
    parametric asynchronous actor systems (\autoref{sec:actorlogic}). 
    \item We introduce two key operations to ensure the termination of a
    backwards abstract interpretation using this symbolic abstraction:
    \begin{inparaenum}[(a)]
      \item \emph{actor-causality reduction} to reduce the non-deterministic branching (\autoref{sec:reflogic}), and
      \item \emph{message-segment summarization} to summarize cyclic message histories (\autoref{sec:parametric}).
    \end{inparaenum}
    We prove the soundness of these operations and establish convergence
    guarantees using a notion of causal dependences.  
    \item We implement this approach in a verification tool for P
    programs~\cite{plang} called \tool{}
    --- the \underline{Dis}tributed \underline{S}ystems \underline{Prove}r
    (\autoref{sec:evaluation}). We
    describe our experience using \tool{} to automatically verify the safety of
    several asynchronous distributed protocols with affine communication,
    including Two-Phase Commit~\cite{twopc-grey06} and One-Shot Leader
    Election~\cite{ivy-pldi18,lewchenko-oopsla25}. To the best of our knowledge,
    this is the first time these two protocols have been verified in an
    asynchronous setting \emph{without} the need for explicit inductive
    invariants. 
\end{itemize}

\section{Overview}\label{sec:motivation}
\newcommand{\exfmtcode}[1]{\fmtcode{\relsize{-2}#1}}
\newcommand{\exnodeidx}[2][n]{\ensuremath{\fmtspecialvar{#1}_{#2}}}
\newcommand{\exnode}[2][n]{n_{#2}}
\newcommand{\exsegparam}[2][\seglen]{\purecolor{\ensuremath{#1_{#2}}}}
\newcommand{\exfld}[2]{\storecolor{\exnode{#1}\fmtcode{.}\exfmtcode{#2}}}
\newcommand{\exsym}[2]{\symvar_{\mbox{\scriptsize$\fmtspecialvar{#1}_{#2}$}}}
\newcommand{\exhandparam}[2]{\purecolor{\exsym{#1}{#2}}}
\newcommand{\exspecialmsg}[5][]{\mkmsg{#2}{\exnode{#3}}{\exfmtcode{#4}^{\purecolor{#1}}}{#5}}
\newcommand{\exmsg}[5][]{\exspecialmsg[#1]{\exnode{#2}}{#3}{#4}{#5}}
\newcommand{\exloc}[5]{\mapping{\exnode{#1}}{\mkmachsttevntloc{\fmtspecialvar{#2}}{\fmtspecialvar{#3}}{\fmtspecialvar{#4}}{\text{\lstnumberstyle #5}}}}
\newcommand{\exabssysstate}[3]{\tuple{#1\colon \histcolor{#2} \mid \purecolor{#3}}}
\newcommand{\exabssysstatetwo}[3]{\begin{array}{@{}l@{\,}l@{}}\tuple{& #1\colon \histcolor{#2} \mid \\&\purecolor{#3}\,}\end{array}}
\newcommand{\exabssysstatethree}[3]{\begin{array}{@{}l@{\,}l@{}}\tuple{& #1\colon \\&\histcolor{#2} \mid \\&\purecolor{#3}\,}\end{array}}
\newcommand{\exabssysstatefour}[4]{\begin{array}{@{}l@{\,}l@{}}\tuple{& #1\colon \\&\histcolor{#2} \mid \\&\purecolor{#3}\\&\purecolor{#4}\,}\end{array}}

\newcommand{\sdbreak}[1]{\draw[thick] (#1) ++(-0.075,0) -- ++(0.15,-0.15) -- ++(-0.15,-0.15) -- ++(0.15,-0.15) -- ++(-0.15,-0.15);}
\newcommand{\sdinstend}[3][]{\path (#2)+(0,-\theseqlevel*\unitfactor-0.7*\unitfactor) node (#3) {#1};}
\newcommand{\sdnodemidbelow}[4][]{\path (#2) -- node[midway, below, yshift=-1mm] (#4) {#1} (#3);}
\newcommand{\sdstatewithin}[2][]{\node[anchor=south west, xshift=1mm] (#2) at (rt1) {#1};}
\newcommand{\sdtriggerleft}[3][]{%
  \node[anchor=south east, xshift=-2mm, fill=lightgray, %
        rounded corners=1pt, inner sep=2pt, font=\scriptsize]
    (#3) at (#2 |- rt1) {\colorlet{Moss}{white}#1};
}
\newcommand{\sdstateout}[2][]{\node[statestyle] at (#2) {#1};}
\newcommand{\exbaremsg}[3][]{\histcolor{\exfmtcode{#2}^{\purecolor{#1}}(#3)}}
\tikzstyle{statestyle}=[draw, rounded corners, fill=white, drop shadow={fill=black!20}]
\tikzstyle{donestyle}=[text opacity=0.6, draw opacity=0.4]
\newcommand{\donetransparent}[1]{{\transparent{0.6} #1}}
\newcommand{\myinststyle}{\tikzstyle{inststyle}=[rectangle, draw, anchor=west, minimum height=0.6cm, minimum width=1.2cm, fill=white, drop shadow={fill=black}]}
\newcommand{\myinstdonestyle}{\tikzstyle{inststyle}+=[donestyle, drop shadow={fill=black!10}]}
\newcommand{\storecausalthreadstyle}{\tikzstyle{threadstyle}+=[fill=Moss!50, top color=Moss!50, bottom color=Moss!50]}
\newcommand{\storecausaldonethreadstyle}{\tikzstyle{threadstyle}+=[fill=Moss!20, top color=Moss!20, bottom color=Moss!20]}
\newcommand{\msgcausalthreadstyle}{\tikzstyle{threadstyle}+=[fill=Midnight!50, top color=Midnight!50, bottom color=Midnight!50]}
\newcommand{\summarycausalthreadstyle}{\tikzstyle{threadstyle}+=[bottom color=Moss!50, top color=Midnight!50]}

\newcommand{\exsyssizefld}[1]{\ensuremath{\##1}}
\newcommand{\exsummarymach}{\fmtabs{\fmtspecialvar{Node}}}
\newcommand{\exsyssize}{\ensuremath{\#\fmtspecialvar{Node}}}
\newcommand{\sysinput}[3]{\ensuremath{#2[#3]\text{.\fmtspecialvar{#1}}}}
\newcommand{\nodeidinput}[2]{\sysinput{id}{#1}{#2}}
\newcommand{\exnodeidinput}[1]{\exhandparam{id}{#1}}

\newcommand{\inlinetikz}[2][]{%
\smash{\begin{tikzpicture}[baseline]
  #1
  #2
\end{tikzpicture}}%
}
\newcommand{\inlineactbar}[1][]{\inlinetikz[#1]{
  \path (0,0) node (start) {} (0,0.24) node (end) {};
  \draw[threadstyle] (start.west) -- (start.east) -- (end.east) -- (end.west) -- cycle;
}}

\begin{wrapfigure}{R}{0pt}%
\footnotesize
\begin{minipage}{22.5em}
\lstnonum\begin{lstlisting}[style=number]
machine Node {
  var id: int;
  var forMe: int;
  var leader: int;
  $\rule{0pt}{2.5ex}$on eInit do (c: tConfig) {$\lstbeginn$
    id = c.id;$\label{pt:eInit-entry}$
    forMe = 0;$\label{pt:eInit-pre-forMe-assign}$ // no votes for me yet
    leader = -1; // no leader yet$\label{pt:eInit-pre-leader-assign}\label{pt:eInit-post-forMe-assign}$
    send ChooseLeader(),eVote;$\label{pt:eInit-pre-choose-leader}\label{pt:eInit-post-leader-assign}$
  }$\label{pt:eInit-post-eVote-send}\lststopn$
  $\rule{0pt}{2.5ex}$on eVote do {$\lststartn$
    forMe = forMe + 1;$\label{pt:eVote-pre-forMe-assign}$
\end{lstlisting}
\begin{lstdiffminus}[style=number]
-   if (forMe >= sizeof(Node)/2) {$\label{pt:eVote-post-forMe-assign}\label{line:evote-bug}$
\end{lstdiffminus}
\begin{lstdiffplus}[style=number,firstnumber=7]
+   if (forMe > sizeof(Node)/2) {$\label{line:evote-fix}$
\end{lstdiffplus}
\begin{lstcont}[style=number]
      leader = id;$\label{line:set-leader}$
      // Broadcast, "I am the leader." $\lststopn$
      $\ldots$
    }
  }
}
\end{lstcont}
\end{minipage}
\caption{A simplified ``one-shot'' leader election protocol. Each node will choose a node to vote for, and send an \texttt{eVote} message to it. When a node has received enough of these, it marks itself as leader (and can broadcast its assumption of leader and start the next phase).
We show a bug-fix commit; in the buggy version (the \colorbox{lightredbg}{\fmtcode{-}} version of \reftxt{line}{line:evote-bug}), the protocol can reach an unsafe state where two nodes believe they are elected at the same time.
}
\label{fig:eVote}
\end{wrapfigure}

In this section, we illustrate the intuition behind our techniques using a
simple leader election protocol shown as a P program~\cite{plang} in
\autoref{fig:eVote}. The protocol begins with each actor receiving an initial
message \fmtcode{eInit}, which configures its unique identifier
(\reftxt{location}{pt:eInit-entry}) and other fields. The initial message
handler ends with the actor choosing who to vote for and sending that actor an
\fmtcode{eVote} message (at \ref{pt:eInit-pre-choose-leader}).  When an actor
receives an \fmtcode{eVote} message, it increments its vote counter by executing
$\ifmtcode{forMe} \mathrel{\ifmtcode{=}} \ifmtcode{forMe + 1}$ at
\reftxt{location}{pt:eVote-pre-forMe-assign}. It then checks whether it has
received a quorum of votes using the condition at \ref{line:evote-bug}. If so,
it sets itself as leader at \ref{line:set-leader}, and can start the
next phase of the protocol (not shown). We consider two versions of the code:
one with a bug in the quorum condition (the \colorbox{lightredbg}{\fmtcode{-}}
version) and a correct version (the \colorbox{lightgreenbg}{\fmtcode{+}}
version).

Leader election is a \emph{consensus protocol}, where a set of actors seek to agree on a value. We can state consensus as a safety property:

\begin{proposition}[\fmtcode{leader}-Consensus Safety]\label{ex:leader-consensus-safety}
For all pairs of nodes $\inode1$ and $\inode2$ with different \ifmtfld{id}s, either they agree who the \ifmtfld{leader} is or they do not yet know who it is (i.e., their \ifmtfld{leader} field is $\inoleader$).
\par\begin{minipage}{15em}
\[
\begin{array}{l}
  \forall \inode1\colon\iNode, \inode2\colon\iNode.\;
    \ifld{\inode1}{id} \neq \ifld{\inode2}{id}
    \implies
\\ \qquad
    \ifld{\inode1}{leader} = \ifld{\inode2}{leader}
    \lor
    \ifld{\inode1}{leader} = \noleader
\\ \qquad \qquad
    \lor
    \ifld{\inode2}{leader} = \noleader
\end{array}
\;
\]
\end{minipage}
\end{proposition}

The essence of safety verification is to \emph{refute} the reachability of
unsafe states, which are categorized by the negation of the safety property:

\begin{proposition}[\fmtcode{leader}-Consensus Violation]\label{ex:leader-consensus-violation}
\[
  \exists \inode1\colon\iNode, \inode2\colon\iNode.\;
    \ifld{\inode1}{id} \neq \ifld{\inode2}{id}
    \land
    \ifld{\inode1}{leader} \neq \ifld{\inode2}{leader}
    \land
    \ifld{\inode1}{leader} \neq \noleader
    \land
    \ifld{\inode2}{leader} \neq \noleader
\;.
\]
\end{proposition}

The conventional approach to safety verification overapproximates the set of
\emph{forwards-reachable} states, i.e., states that are reachable from an
initial state, and shows that it excludes the unsafe states. We adopt a dual
approach: we treat the negation of the safety property as \emph{goal state} and
we overapproximate the set of \emph{backwards-reachable} states, i.e., states
from which the goal state is reachable. In the current example, the goal state
is an error state consisting of two nodes with different \ifmtfld{id}s that
disagree on the \ifmtfld{leader} field (Prop.
\ref{ex:leader-consensus-violation}). Any concrete state containing this slice
is unsafe, so proving safety amounts to refuting the reachability of this slice.
We proceed by over-approximating paths back to an initial system state, refuting
each path, or witnessing a possible violation path.

In the current example, there are an unbounded number of paths backwards from
the goal/error state, where any actor, not just the two that disagree, could
have executed \emph{any} handler previously. A naive approach would thus lead to
a substantial state-space explosion immediately. We overcome this problem by
observing that out of the infinite possible paths backwards, only a few are
\emph{causally relevant} to the goal state. By restricting the analysis to only
consider the events that are relevant to producing  the current state, we
prioritize exploring the possible causal histories of the safety violation over
arbitrary interleavings, allowing us to contain the state-space explosion while
preserving completeness. In the current example, events relevant to the goal state include those that assign to the \ifmtfld{leader}
field, namely \fmtcode{eVote} and \fmtcode{eInitial}.  Examining these events
introduces new data dependencies: their preconditions expose additional fields
and messages that must themselves be explained by earlier events. Actors
relevant to the goal state are those explicitly referenced in the safety
violation (Prop.~\ref{ex:leader-consensus-violation}).  The analysis proceeds by
iteratively expanding the set of relevant events for relevant actors, following
these causal dependencies backward until all required state is witnessed. We
analyze distributed protocols as parametrized actor systems, that is, we
assume the sets of actors (of type $\fmtcode{Node}$) are of a fixed symbolic
size (\ifmtcode{{\lstkeywordstyle sizeof}(Node)}), and the initial deployment
code (a main actor $\fmtcode{Main}$) is outside the scope of our analysis. The
lifetime of each actor is assumed to begin with a distinguished initialization
event (e.g., \fmtcode{eInitial}), which serves as the end point for the
backwards causal analysis.

\subsection{A Symbolic Abstraction of Actor Systems}

To perform the backwards causal analysis, we maintain a symbolic representation
of the relevant slice of the system. In the current example, the goal/error
state consists of two symbolic actors with their fields constrained as shown in
Prop.~\ref{ex:leader-consensus-violation}.  We construct an abstract state
representing this slice as shown below\footnote{
  \textbf{Notation.} We use ($:$) and mid ($\mid$) to merely separate the
  components of the abstract state. No type-theoretic interpretation is intended
  despite syntactic similarities to dependent type notation.
}:

\[
 \exabssysstate
  { \loctopx }
  { \abshistempx }
  {
    (\exfld{1}{id} \neq \exfld{2}{id})
    \land
    (\exfld{1}{leader} \neq \exfld{2}{leader})
    \land
    (\exfld{1}{leader} \neq -1)
    \land
    (\exfld{2}{leader} \neq -1)
  }
\]

This state contains 3 components:

\begin{enumerate}
\item \textbf{Location} - We maintain a conjunction of the current program
location of every (symbolic) actor manifested by the analysis. The special
symbol $\loctopx$ denotes an unconstrained location. Locations are modeled using
intuitionistic separation logic (cf.
\autoref{sec:distributed-control-materialization} and \autoref{sec:actorlogic}).

\item \textbf{\histcolor{History}} - We use ordered logic~\cite{orderedlogic} to
record the event/message history witnessed by the backwards analysis on the path
to manifesting the current state. The unmanifest prefix of the history that is
yet to be explored is denoted using the special symbol
$\abshistempx$~\cite{historia}. Technically, $\abshistempx$ denotes the set of
\emph{realizable} message histories~\cite{historia}, i.e., the sequence of
messages that are possible to witness in a concrete protocol execution beginning
at the initial state. 

\item \textbf{\purecolor{Constraints}} - We maintain a set of quantifier-free
logical formulas capturing the constraints on \storecolor{actor-local fields}
manifest in the current state.  \end{enumerate}

The formal development (\autoref{sec:actorlogic}) distinguishes between pure and
separation-logic constraints, and includes an additional component --- a
linear-logic~\cite{DBLP:journals/tcs/Girard87} abstraction of the network state. We elide these details here to simplify presentation.

\subsection{Goal-Directed Actor Materialization}
\label{sec:distributed-control-materialization}

We illustrate our analysis using sequence diagrams such as \autoref{fig:distributed-control-materialization}. Reading such a diagram from bottom to top traces the goal-directed backwards analysis, whereas reading it top to bottom recovers the partial violation witness that the analysis constructs.

\begin{figure}
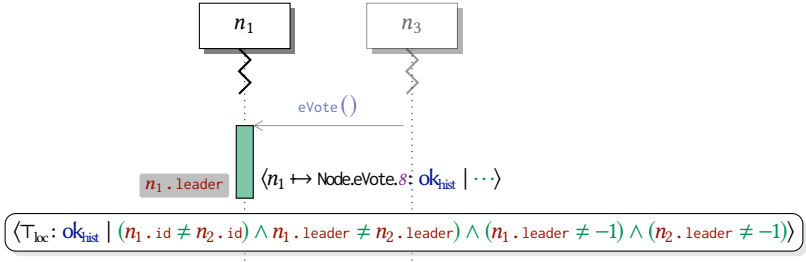
\centering\small
\begin{sequencediagram}
  \myinststyle
  \newinst{n1}{$\exnode1$}\sdbreak{n1}
  \begin{scope}[donestyle]
  \myinstdonestyle
  \newinst[1]{n3}{$\exnode3$}\sdbreak{n3}

  \storecausalthreadstyle
  \begin{messcall}{n3}{ $\exbaremsg{eVote}{}$ }{n1}
  \postlevel
  \end{messcall}
  \sdinstend{n1}{n1return}
  \sdinstend{n3}{n3return}
  \end{scope}

 \sdtriggerleft[$\exfld{1}{leader}$]{n1}{n1trigger}
  \sdstatewithin[%
    $\exabssysstate{\exloc{1}{Node}{Init}{eVote}{\ref{line:set-leader}}}{\abshistempx}{\cdots}$
  
  ]{eElectedLocation}

  \sdnodemidbelow{n1return}{n3return}{n1n3midbelow}
  \node[statestyle, xshift=1cm] at (n1n3midbelow) {%
  $\exabssysstate
  { \loctopx }
  { \abshistempx }
  {
    (\exfld{1}{id} \neq \exfld{2}{id})
    \land
    \exfld{1}{leader} \neq \exfld{2}{leader})
    \land
    (\exfld{1}{leader} \neq -1)
    \land
    (\exfld{2}{leader} \neq -1)
  }$
  };
\end{sequencediagram}
\caption{ Materializing an actor from a field dependency. Since the error state
constrains the leader field, the analysis materializes actor $\exnode{1}$ at a
program point that writes leader, namely the eVote handler at
\reftxt{location}{line:set-leader}. Activation bars represent handler instances
under consideration and are annotated on the right with the state's updated
location, and on the left with a label denoting the field or message causing the
update. Actors appear across the top of the diagram, with greyed actors not yet
materialized.  Message sends are shown as arrows. }
\label{fig:distributed-control-materialization}
\end{figure}

The backwards analysis begins from the goal state, whose constraints reference
two error-relevant fields, $\exfld{1}{leader}$ and $\exfld{2}{leader}$. To
explain how these fields came to hold their values, the analysis locates the
program points that assign \ifmtfld{leader} ---
\reftxt{location}{line:set-leader} in \fmtcode{eVote} and
\reftxt{location}{pt:eInit-pre-leader-assign} in \fmtcode{eInit}---and
hypothesizes that each of $\exnode{1}$ and $\exnode{2}$ last executed one of
them. Branching over both choices for both actors yields four states. The two
states that place an actor in \fmtcode{eInit} are immediately infeasible:
\fmtcode{eInit} sets \ifmtfld{leader} to $-1$, contradicting the goal constraint
\purecolor{$(-1 \neq -1)$}. Only the two states in which both actors are in
\fmtcode{eVote} survive, and as they are symmetric we follow the path for
$\exnode{1}$ and elide the other, since any witness must eventually account for
both assignments.

A distributed protocol is composed of an unbounded number of actors, so a system
state must, in principle, account for the local state and control location of
infinitely many actors at once --- an obstacle to any finite analysis. We
overcome it through \emph{materialization}, echoing the causal-relevance
argument above: just as only a few \emph{events} bear on a violation, only a few
\emph{actors} need be reasoned about explicitly to refute it. The analysis
therefore leaves the actor population unmanifest, as an anonymous, unbounded
``blob'' about which nothing is assumed, and \emph{materializes} an actor out of
this blob one at a time, and only when an outstanding obligation demands it.
The crucial point is that finitely many manifest actors suffice to refute a
violation even though protocol executions span unboundedly many actors. A
refutation inspects only a \emph{suffix} of an execution that is concerned with
the activity of a finite set of manifest actors while leaving the rest of
unmanifest actors unconstrained. The suffix itself need not be finite --- here
it happens to be, but in general an unbounded suffix over finitely many manifest
actors is captured by the notion of \emph{state entailment} we develop in
\autoref{sec:distributed-control-location-materialization}.

In \autoref{fig:distributed-control-materialization}, we materialize
$\exnode{1}$: we fix its program location to
$\exloc{1}{Node}{Init}{eVote}{\ref{line:set-leader}}$, committing to the
hypothesis that $\exnode{1}$ reached this point of the \fmtcode{eVote} handler.
Constraints are unchanged, hence elided with ellipsis. 
Actor $\exnode{1}$'s activation bar \inlineactbar[\storecausalthreadstyle] is
labeled with \colorbox{lightgray}{$\exfld{1}{leader}$}, highlighting cause for this materialization,
namely a write to this field.  Stepping backwards through $\fmtcode{eVote}$, the
analysis learns \purecolor{$\exfld{1}{leader} = \exfld{1}{id}$} from the
assignment \ifmtcode{leader = id}, and \purecolor{$\exfld{1}{forMe} + 1 \ge
(\exsyssize/2)$} from the guarding conditional, where \purecolor{$\exsyssize$}
denotes the number of actors of type $\fmtcode{Node}$ (i.e.,
\ifmtcode{{\lstkeywordstyle sizeof}(Node)}).  Reaching the top of the handler,
the analysis incurs a new obligation: for $\exnode{1}$ to have executed
\fmtcode{eVote}, it must have received an \histcolor{\fmtcode{eVote}} message
from some actor, which we name $\exnode{3}$.  Because the analysis has not yet
materialized $\exnode{3}$ or witnessed the corresponding send, we render both
\donetransparent{$\exnode{3}$} and its message partially transparent in the
diagram.

\subsection{Causal Message Sends}
\label{sec:actor-causality-reduction}

Having reached the top of the \fmtcode{eVote} handler, the analysis has fully
accounted for this activation of actor $\exnode{1}$ and drops the constraint on
its location (\autoref{fig:actor-causality-reduction}): the event preceding this
handler is not yet determined and $\exnode{1}$'s location is therefore
unconstrained. What remains is the obligation to explain the triggering message,
which the analysis records by extending the witness history with\footnote{
  \textbf{Notation.} $\exmsg{3}{1}{eVote}{ }$ denotes a message of type
  $\exfmtcode{eVote}$ sent from $\exnode{3}$ to $\exnode{1}$ with empty payload.
  Given a message $m$, $\abshisthyp{m}{\abshistempx}$ denotes all realizable
  message histories that can be extended with $m$ and still remain realizable
  (e.g., they still remain affine).
  The operator $\abshisthyp{}{}$ is right-associative.
} $\exmsg{3}{1}{eVote}{ }$.

\begin{figure}\centering\small
\begin{sequencediagram}
  \myinststyle
  \newinst{main}{$\mainmach$}\sdbreak{main}
  \begin{scope}[donestyle]
  \myinstdonestyle
  \newinst{n1}{$\exnode1$}\sdbreak{n1}
  \end{scope}
  \myinststyle
  \newinst{n3}{$\exnode3$}\sdbreak{n3}
  \sdinstend{n1}{n1top}
  \sdinstend{n3}{n3top}
  \postlevel

  \sdnodemidbelow{n1top}{n3top}{n1n3midtop}
  \node[statestyle, xshift=3cm] at (n1n3midtop) {%
  $\exabssysstate
  { \loctopx }
  { \abshisthyp{ (\exspecialmsg{\mainmach}{3}{eInit}{\purecolor{-}}) }{\abshisthyp{ (\exmsg{3}{1}{eVote}{ }) }{\abshistempx} }}
  {
    (\exfld{1}{id} \neq \exfld{2}{leader})
    \land
    \cdots
    \land
    (\exfld{1}{forMe} + 1 \ge \exsyssize / 2)
  }$
  };

  \msgcausalthreadstyle
  \begin{messcall}{main}{ $\exbaremsg{eInit}{\purecolor{-}}$ }{n3}
  \postlevel
  \end{messcall}
  \sdinstend{n1}{n1returnmain}
  \sdinstend{n3}{n3returnmain}
  \postlevel
  \sdtriggerleft[$\exmsg{3}{1}{eVote}{}$]{n3}{n3trigger}
  \sdstatewithin[%
    $\exabssysstate{ \exloc{3}{Node}{Init}{eInit}{\ref{pt:eInit-post-eVote-send}} }{ \abshisthyp{ (\exmsg{3}{1}{eVote}{} ) }{\abshistempx} }{ \cdots }$
  ]{eInitLocation}

  \postlevel
  \postlevel
  \postlevel

  \sdnodemidbelow{n1return}{n3return}{n1n3midbelow}
  \node[statestyle, xshift=3cm,yshift=-.5cm] at (n1n3midbelow) {%
  $\exabssysstate
  { \loctopx }
  { \abshisthyp{ (\exmsg{3}{1}{eVote}{} ) }{\abshistempx} }
  {
    (\exfld{1}{id} \neq \exfld{2}{leader})
    \land
    (\exfld{1}{id} \neq -1)
    \land
    (\exfld{2}{leader} \neq -1)
     \land
    (\exfld{1}{forMe} + 1 \ge \exsyssize / 2)
  }$
  };

  \storecausaldonethreadstyle
  \prelevel
  \prelevel
  \prelevel %
  \begin{messcall}{n3}{ $\exbaremsg{eVote}{}$}{n1}
  \postlevel
  \end{messcall}

  \begin{scope}[donestyle]
    \sdtriggerleft[$\exfld{1}{leader}$]{n1}{n1trigger}
    \sdstatewithin[%
      
  $\exabssysstate{\exloc{1}{Node}{Init}{eVote}{\ref{line:set-leader}}}{\abshistempx}{\cdots}$
  
  ]{eVoteLocation}
  \end{scope}
\end{sequencediagram}
\caption{Upwards continuation of
\autoref{fig:distributed-control-materialization} showing the materialization of
the necessary message-sends to witness the message
\histcolor{$\exmsg{3}{1}{eVote}{}$}. The analysis prioritizes
\emph{actor-causality reduction}, meaning that it first derives the sequence of
messages in the abstract message history \histcolor{$\abshisthyp{
(\exspecialmsg{\mainmach}{3}{eInit}{\purecolor{-}}) }{
\abshisthyp{(\exmsg{3}{1}{eVote}{ })}{\abshistempx} }$} for witnessing the
message \histcolor{$\exmsg{3}{1}{eVote}{}$} before considering other branches.
}
\vspace*{-0.2in}
\label{fig:actor-causality-reduction}
\end{figure}

Discharging this obligation introduces new goal constraints---on the field $\exfld{1}{forMe}$ and on the origin of the recorded message, which some actor must have sent. The analysis therefore now considers not only assignments to error-relevant fields but also the \fmtcode{send} commands that could have produced the messages the history must witness. This is where confining exploration to \emph{causally relevant} events pays off: rather than interleaving arbitrary handlers, the analysis follows the message backwards to its sender. In \autoref{fig:actor-causality-reduction}, again read from the bottom, we choose to witness this send first. The analysis materializes a fresh actor $\exnode{3}$ at $\exloc{3}{Node}{Init}{eInit}{\ref{pt:eInit-post-eVote-send}}$, the \lstinline!send n,eVote! command in $\fmtcode{eInit}$. We draw its activation bar in $\histcolor{blue}$ to signal that $\exnode{3}$ was materialized to explain a message in the history rather than a field constraint.

Under our parameterized-system model, the \fmtcode{eInit} message that started $\exnode{3}$, \histcolor{$\exspecialmsg{\mainmach}{3}{eInit}{\purecolor{-}}$}, originates from the main $\mainmach$ machine outside the system under analysis, so no further send need be witnessed along this path. Reaching the top of the handler, the analysis appends this message to the history and unconstrains $\exnode{3}$'s location. Although its send cannot be witnessed, the message must still be recorded, since any history in which $\exnode{3}$ had received another message beforehand would be unrealizable.

\subsection{Witnessing a Violation}

With $\exnode{1}$ again at an idle location, the analysis selects a further obligation to discharge. In \autoref{fig:message-segments}, we continue by witnessing the field \fmtcode{forMe}, constrained by \smash{\purecolor{$(\underline{\exfld{1}{forMe}} + 1 \ge \exsyssize/2)$}}. Two handlers assign this field, giving two ways to explain its value:
\begin{inparaenum}[(1)]
  \item\label{enum:forMe-eInit} the \ifmtcode{forMe = 0} initialization in the $\fmtcode{eInit}$ handler (materializing $\exnode{1}$ at $\exloc{1}{Node}{Init}{eInit}{\ref{pt:eInit-post-forMe-assign}}$); or
  \item\label{enum:forMe-eVote} the \ifmtcode{forMe = forMe + 1} increment in the $\fmtcode{eVote}$ handler (materializing $\exnode{1}$ at $\exloc{1}{Node}{Init}{eVote}{\ref{pt:eVote-post-forMe-assign}}$).
\end{inparaenum}
Assume the system has at least two actors, \purecolor{$\exsyssize \geq 2$}, as
the goal state with two distinct leaders is otherwise trivially unreachable.
Assuming the analysis follows option~\eqref{enum:forMe-eInit}, it substitutes
\ifmtcode{forMe = 0} and updates the constraint to \purecolor{$(1 \ge
\exsyssize/2) \land (\exsyssize \geq 2)$} at the point just before the
initialization
($\exloc{1}{Node}{Init}{eInit}{\ref{pt:eInit-pre-forMe-assign}}$). This
constraint is satisfiable, so the path survives: the analysis reaches the top of
\fmtcode{eInit}, then discharges the symmetric obligation arising from
$\exfld{2}{leader}$, which likewise yields no contradiction. With no fields or
messages left to witness, the analysis returns the accumulated history and
constraints as a \emph{potential safety violation}---exactly the bug present in
the buggy version of \autoref{fig:eVote}.

\subsection{Message Segments Summarize Unbounded Sequences of Messages}
\label{sec:message-segments}

Now consider the fixed version of the protocol in \autoref{fig:eVote}, in which
the guard $\ge$ is replaced with $>$. The analysis proceeds exactly as before up
to \autoref{fig:actor-causality-reduction}. On reaching an idle location, it
again chooses between the two assignments to \fmtcode{forMe}. Choosing
\fmtcode{eInit} (option~\eqref{enum:forMe-eInit}) now yields the contradictory
constraint \purecolor{$(1 > \exsyssize/2) \land (\exsyssize \geq 2)$}, so that
path is pruned.

The remaining option is~\eqref{enum:forMe-eVote}, witnessing another
\fmtcode{eVote} increment. This updates the constraint to
\smash{\purecolor{$(\underline{\exfld{1}{forMe}} + 1 + 1 > \exsyssize/2)$}} but
once more leaves \fmtcode{forMe} unconstrained, so witnessing it again simply
replays the same step. The backwards analysis thus falls into an unbounded cycle
in the message history of the form shown below:

\begin{example}[A cycling message history, repeatedly visiting location $\exloc{1}{Node}{Init}{eVote}{\ref{pt:eVote-post-forMe-assign}}$]\label{ex:cycling-message-history}\small
\[
\histcolor{
\abshisthyp{\cdots}{
  \abshisthyp{ (\exmsg{43}{1}{eVote}{ }) }
  {
    \abshisthyp{ (\exmsg{42}{1}{eVote}{ }) }
    {
      \abshisthyp{ (\exmsg{41}{1}{eVote}{ }) }
      {
        \abshisthyp{\cdots}{\abshistempx}
      }
    }
  }
}
}
\]
\normalsize
\end{example}

This cycle is the crux of what makes automated verification of quorum-based
protocols hard. Intuitively, the protocol accumulates votes through repeated
executions of the \fmtcode{eVote} handler: each invocation increments the
actor's \fmtcode{forMe} field until a quorum is reached and the actor declares
itself elected. The backwards analysis unrolls these executions one vote at a
time and never terminates. To converge, it must eventually \emph{widen} this
unbounded sequence of handler executions into a finite summary.

Summarizing this repetition is subtler than folding a sequential loop. Each
iteration here is a distinct message-handler invocation, and in an asynchronous
system such invocations may be arbitrarily interleaved with other handlers.
Before the repeated \fmtcode{eVote} handlers can be treated as a sequential
loop, we must therefore show that they can be reordered to run consecutively. We
establish this using Mazurkiewicz trace equivalence~\cite{mazurkiewicz}:
whenever the cycling handler is independent of every handler that might
interleave with it, an equivalent execution exists in which its invocations
occur consecutively. In our motivating example the only other relevant handler
is \fmtcode{eInit}, which cannot interleave with \fmtcode{eVote}, so the
condition holds and we may consider the trace in which the unbounded
\fmtcode{eVote}s appear back to back.

\begin{figure}\centering\small
\begin{sequencediagram}
  \myinststyle
  \newinst{main}{$\mainmach$}\sdbreak{main}
  \begin{scope}[donestyle]
  \myinstdonestyle
  \newinst{n1}{$\exnode1$}\sdbreak{n1}
  \end{scope}
  \myinststyle
  \newinst{n2}{$\exnode2$}\sdbreak{n2}
  \newinst{n3}{$\exnode3$}\sdbreak{n3}
  \myinstdonestyle
  \newinst{sys}{$ \exsummarymach$}
  \myinststyle
  \begin{scope}[donestyle]
  \sdbreak{sys} 
  \end{scope}

  \sdinstend{n1}{n1top}
  \sdinstend{n3}{n3top}
  \sdnodemidbelow{n1top}{n3top}{n1n3midtop}

  \storecausalthreadstyle
  \begin{messcall}{main}{ $\exbaremsg{eInit}{\purecolor{-}}$ }{n1}
  \postlevel
  \end{messcall}
  \sdinstend{n1}{n1returninit}
  \sdinstend{n3}{n3returninit}
  \sdnodemidbelow{n1returninit}{n3returninit}{n1n3midreturninit}

  \sdtriggerleft[$\exfld{1}{forMe}$]{n1}{n1inittrigger}
  \sdstatewithin[%

    $\exabssysstatetwo
  { \exloc{1}{Node}{Init}{eInit}{\ref{pt:eInit-post-forMe-assign}} }
  {
   \cdots
  }
  { \cdots
  }$
  ]{eVoteLocation}

  \postlevel
  \node[statestyle, xshift=3cm, yshift=-2ex] at (n1n3midreturninit) {%
  $\exabssysstatetwo
  { \loctopx }
  {
    \abshisthyp
    {(\exspecialmsg[\exsegparam{1}]{ \exsummarymach}{1}{eVote}{ })}
    { \abshisthyp{ \cdots }{\abshisthyp{ (\exmsg{3}{1}{eVote}{ }) }{ \abshisthyp{\cdots}{\abshistempx} }} }
  }
  {
    (\exfld{1}{id} \neq \exfld{2}{leader})
    \land
    \cdots
    \land
    (\exfld{1}{forMe} + \exsegparam{1} + 1 > \exsyssize/2)
    \land
    (1 + \exsegparam{1} \leq \exsyssize)
  }$
  };

  \postlevel
  \summarycausalthreadstyle
  \begin{messcall}{sys}{ $\exbaremsg[\exsegparam{1}]{eVote}{ }$ }{n1}
  \postlevel
  \end{messcall}
  \sdinstend{n1}{n1returnvote}
  \sdinstend{n3}{n3returnvote}
  \sdnodemidbelow{n1returnvote}{n3returnvote}{n1n3midreturnvote}
\sdtriggerleft[$\exfld{1}{forMe}$]{n1}{n1trigger}
  \sdstatewithin[%

    $\exabssysstate
  { \exloc{1}{Node}{Init}{eVote}{\ref{pt:eVote-post-forMe-assign}}}
  { \abshisthyp{ (\exspecialmsg{\mainmach}{3}{eInit}{\purecolor{-}}) }{\abshisthyp{ (\exmsg{3}{1}{eVote}{ }) }{ \abshisthyp{(\cdots)}{\abshistempx} }} }
  {\cdots
  }$
  ]{eVoteLocation}

  \begin{scope}[donestyle]
  \node[statestyle, xshift=3cm] at (n1n3midreturnvote) {%
  $\exabssysstate
  { \loctopx }
  { \abshisthyp{ (\exspecialmsg{\mainmach}{3}{eInit}{\purecolor{-}}) }{\abshisthyp{ (\exmsg{3}{1}{eVote}{ }) }{ \abshisthyp{(\cdots)}{\abshistempx} }} }
  {
    (\exfld{1}{id} \neq \exfld{2}{leader})
    \land
    \cdots
    \land
    (\exfld{1}{forMe} + 1 > \exsyssize/2)
  }$
  };
  \end{scope}
\end{sequencediagram}
\caption{Handling the message segment \histcolor{$\exspecialmsg[\exsegparam{1}]{ \exsummarymach}{1}{eVote}{ }$} summarizes the handling of a sequence of $\exsegparam{1}$ instances of the message \histcolor{$\exspecialmsg{-}{1}{eVote}{ }$} for arbitrary sending actors. The special sender $\exsummarymach$ denotes an abstract set of actors of type \fmtspecialvar{Node}, representing the set of senders of the messages summarized by the segment.
}
\label{fig:message-segments}
\end{figure}

This repetition motivates our novel notion of a \emph{message segment}: a single
history element that summarizes the handling of an unbounded sequence of like
messages. In our motivating example, the analysis replaces the cycle with the
segment \histcolor{$\exspecialmsg[\exsegparam{1}]{ \exsummarymach}{1}{eVote}{
}$}, summarizing all the votes $\exnode{1}$ receives from an abstract summary
actor \donetransparent{$\exsummarymach$}. This summary actor represents the set
of senders of messages summarized in the segment. Its length $\exsegparam{1}$ is
a symbolic parameter denoting the number of summarized messages, which lets the
analysis learn the closed-form constraint \purecolor{$(\exfld{1}{forMe} +
\exsegparam{1} + 1 > \exsyssize/2)$} in place of the divergent unrolling.
\autoref{fig:message-segments} shows the effect of this summarization; we draw
the segment's activation bar \inlineactbar[\summarycausalthreadstyle] with a
gradient from \purecolor{green} to \histcolor{blue}, reflecting that a segment
combines both forms of causality discussed above.

\subparagraph{Affine protocols, abstract message histories, and system-size lower bounds.}

On its own, a segment's length bounds from below the number of messages an error requires. The affine restriction turns this message count into an actor count, relating $\exsegparam{1}$ to \purecolor{$\exsyssize$}.
Because affine communication permits each actor to send $\exnode{1}$ at most one \fmtcode{eVote}, the senders of the $\exsegparam{1}$ summarized \histcolor{$\exspecialmsg{-}{1}{eVote}{ }$} messages must be $\exsegparam{1}$ distinct actors, yielding \purecolor{$\exsegparam{1} \leq \exsyssize$}. Moreover, the actor $\exnode{3}$ that sent the separately materialized \histcolor{$\exmsg{3}{1}{eVote}{ }$} must differ from all $\exsegparam{1}$ of these senders, strengthening the bound to \purecolor{$1 + \exsegparam{1} \leq \exsyssize$}.

\subsection{Verifying the Fix}

Completing the symmetric analysis for $\exnode{2}$ materializes a second
\fmtcode{eVote} segment, of length $\exsegparam{2}$, after which the affine
reasoning leaves the analysis with a constraint of the following form:

\[
\purecolordecl
    \ldots
    \land
    \underline{(\exsegparam{1} + 1 > \exsyssize / 2)
    \land
    (\exsegparam{2} + 1 > \exsyssize / 2)
    \land
    (1 + \exsegparam{1} + 1 + \exsegparam{2} \leq \exsyssize)}
    \]

    The two underlined quorum requirements---one per candidate leader---together force strictly more than \purecolor{$\exsyssize$} participating actors, whereas the affine bound \purecolor{$1 + \exsegparam{1} + 1 + \exsegparam{2} \leq \exsyssize$} admits at most \purecolor{$\exsyssize$}. The conjunction is therefore unsatisfiable \emph{for every system size}, so the analysis refutes this witness the moment it produces the second segment. As every backwards path is likewise refuted, the goal state is \emph{unreachable on all backwards paths}, establishing the safety property that two actors never disagree on the elected leader.

    It is worth noting how a single refutation covers every configuration of the protocol. Any violating system must contain the two disagreeing actors $\exnode{1}$ and $\exnode{2}$, together with the actors that voted them into leadership, which the analysis denotes $\exnode{3}$ and $\exnode{4}$. The constraints never require $\exnode{3}$ and $\exnode{4}$ to differ from $\exnode{1}$ and $\exnode{2}$, though $\exnode{3}$ and $\exnode{4}$ cannot themselves coincide; every remaining voter is subsumed by the summary segments and their symbolic lengths $\exsegparam{1}$ and $\exsegparam{2}$. The refutation is thus parametric in the number of actors, discharging all system sizes at once.

\section{Goal-Directed Actor Materialization}\label{sec:actorlogic}

In this section, we formalize our approach with an abstract machine model for actor systems~\cite{agha}. Our model is unique in exposing fine-grained distributed control-location maps as a separation logic formula (\autoref{sec:concrete-semantics}).
We then describe a backwards, goal-directed symbolic execution semantics that makes use of \emph{actor materialization} to reason about parametric actor systems (\autoref{sec:abstract-semantics}).
A key distinction between this and other automated approaches to distributed systems verification is the ability to perform local reasoning via separation logic.

\subsection{An Abstract Machine for Actor Systems}
\label{sec:concrete-semantics}

\subsubsection{Syntax.}\label{sec:dcsnh-syntax}%
We first define the syntax and operational state of our abstract machine. A distributed protocol $\topsys$ is composed of a set of actor classes each defining the internal logic of a specific machine type $\mach$ in the protocol.
A \emph{message handler} ($\hand \bnfdef \mkonhand{\evnt}{ \seq{\var} }{\proc}$) is a handler for event-type identifier $\evnt$ with formal parameters $\seq{\var}$ and a procedure body $\proc$.\footnote{
  Notation: We write $\seq{\square}$ to denote a sequence. We use
  script letters (e.g., $\mach$, $\evnt$) for identifiers of state-machine
  components and italic words (e.g., $\actor$, $\field$,
  $\hand$, $\proc$) for syntactic components. For convenience, we treat the
  protocol specification $\pi$ as a map: $\maplookup{\topsys}{\mach}$ gets the actor class $\actor$ corresponding to machine-type id $\mach$ (and similarly for $\maplookup{\actor}{\evnt}$).}
A procedure
is a set of transitions $\trans$, where each transition ($\trans \bnfdef
\mktrans{\prelbl\loc}{\instr}{\postlbl\loc} \bnfalt
\mktrans{\prelbl\loc}{\cmd}{\postlbl\loc}$) is a control-flow edge from a
pre-location $\prelbl\loc$ to post-location $\postlbl\loc$, labeled with
either an imperative instruction $\instr$ or an actor command $\cmd$.
A procedure has two special control-flow
locations $\entryloc$  and $\exitloc$ ($\loc
\bnfdef \entryloc \bnfalt \exitloc \bnfalt \cdots$).
Actor commands
$\cmd$
include sending a message to another actor with $\mksend{\var}{\evnt}{\seq{\varalt}}$,
where variable $\var$ is bound to the recipient address (i.e., a unique identifier for the recipient actor), $\evnt$ is an event-type identifier, and variables $\seq{\varalt}$ are bound to its payload.
A figure summarizing the syntax of the language can be found in \autoref{sec:syntax-appendix}.

The semantics for actor commands follows what would be expected for  asynchronous distributed systems (e.g., \cite{plang}).
A $\sendkw$ command enqueues a message to the network, which is a bag of messages that have been sent but not yet handled by the recipient (cf. \autoref{sec:transition-system-semantics}).
Unlike the standard actor calculus~\cite{agha}, we have no $\mknew{\mach}$ command to spawn a new actor of machine type $\mach$ --- we consider parametrized actor systems where only the main machine can spawn actors to set up the system.
As such, there is a special $\mainmach$ machine-type identifier for the main machine that starts the system ($\mach \bnfdef \mainmach \bnfalt \cdots$).

There are also two special event-type identifiers ($\evnt \bnfdef \initevnt
\bnfalt \idleevnt \bnfalt \cdots$): $\initevnt$ for a special event type
that can only be sent by the main machine to start the protocol execution
and $\idleevnt$ for a special idle event type with no handler. An
$\initevnt$ message can have different payloads but can only be received
once by each actor and before receiving any other messages.
The $\idleevnt$ event type is used to model actors that are idling between handling messages. Conceptually, a machine is idle when it reaches $\exitloc$ for any handler. 
The quiescent point between rounds of affine communication is essentially where all actors except the main machine are idling.

\begin{figure}\small
\begin{mathpar}
\text{distributed states} \quad \sysstate \bnfdef \mksysstate{\sysloc}{\sysstore}{\net}{\hist}

\text{control-location maps} \quad \sysloc \bnfdef \ctxnil \bnfalt \ctxcons{\sysloc}{ \mapping{\addr}{ \mkmachsttevntloc{\mach}{\stt}{\evnt}{\loc} } }

\text{distributed stores} \quad \sysstorecolored \bnfdef \ctxnil \bnfalt \ctxcons{\sysstore}{ \mapping{\addr}{\actorstore} }

\text{actor addresses} \quad \addr

\\
\text{network} \quad \netcolored \bnfdef \ctxnil \bnfalt \ctxcons{\net}{\msg}

\text{message histories} \quad \histcolored \bnfdef \histemp \bnfalt \histcons{\hist}{\msg}

\text{messages} \quad \msg \bnfdef \mkmsg{\fromlbl\addr}{\tolbl\addr}{\evnt}{ \seq{\val} }

\text{actor stores} \quad \storecolor{\actorstore} \bnfdef \mkactorstore{\fieldslbl\store}{\varslbl\store}

\text{stores} \quad \storecolor{\store} \bnfdef \ctxnil \bnfalt \ctxcons{\store}{ \mapping{\var}{\val} }
\end{mathpar}
\caption{Semantic domains for the \dCSNH{} machine state. Following classical abstract machines, we name our state for its components: distributed Control, Stores, Network, and History.}
\label{fig:semantic-domains}
\end{figure}

\subsubsection{Semantic Domains.}\label{sec:concrete-semantics-domains}
We define the program state as an abstract machine model, denoted \dCSNH{}, and give the semantic domains in \autoref{fig:semantic-domains}. The name alludes to classical abstract machines such as CESK~\cite{vanaam}, reflecting the similar state representation.\footnote{To allude to classical abstract machines, we call our model \dCSNH{} (the \underline{d}istributed \underline{c}ontrol-\underline{s}tore-\underline{n}etwork-\underline{h}istory machine).}
The distributed state $\sysstate$
is defined as a 4-tuple
$\mksysstate{\sysloc}{\sysstore}{\net}{\hist}$ for the control-location map $\sysloc$, the distributed store \storecolor{$\sysstore$}, the
network \netcolor{$\net$}, and the message history \histcolor{$\hist$}.
The \emph{control-location map} $\sysloc \bnfdef \ctxnil \bnfalt \ctxcons{\sysloc}{ \mapping{\addr}{ \mkmachsttevntloc{\mach}{\stt}{\evnt}{\loc} } }$ is a map from each actor address
$\addr$ to its control location, a 3-tuple
$\mkmachsttevntloc{\mach}{\stt}{\evnt}{\loc}$ with the actor's machine type
$\mach$, current event-type being handled $\evnt$, and the program location. The \emph{distributed store}
$\storecolor{\sysstore} \bnfdef \ctxnil \bnfalt \ctxcons{\sysstore}{
\mapping{\addr}{\actorstore} }$ maps actor addresses to their local stores tracking fields and their values. The \emph{network} $\netcolor{\net} \bnfdef \ctxnil \bnfalt
\ctxcons{\net}{\msg}$ is a set of messages $\msg$ that have been sent but
not yet handled by the receiver, while the \emph{message history}
$\histcolor{\hist} \bnfdef \histemp \bnfalt \histcons{\hist}{\msg}$ is a
trace of messages $\msg$ that have been sent. A message $\msg \bnfdef
\mkmsg{\fromlbl\addr}{\tolbl\addr}{\evnt}{ \seq{\val} }$ is a tuple with
the sender address $\fromlbl\addr$, the recipient address $\tolbl\addr$,
the event-type identifier $\evnt$ for the message handler to be invoked,
and a sequence of values $\seq{\val}$ for the message payload.

\subsubsection{Transition-System Semantics.}
\label{sec:transition-system-semantics}

\newcommand{\jhyp}[2][]{#1 \vdash #2}

\newcommand{\jaffine}[1]{#1\;\mathsf{affine}}
\newcommand{\jpaffine}[1]{#1\;\mathsf{paffine}}
\newcommand{\jhyppaffine}[1]{\jhyp{\jpaffine{#1}}}

\newcommand{\paffinefun}[1]{\operatorname{paffine}(#1)}

\newcommand{\steparrow}{\longrightarrow}
\newcommand{\jstep}[3][]{#2 \steparrow_{#1} #3}

\newcommand{\backsteparrow}{\Longleftarrow}
\newcommand{\jbackstep}[3][]{#2 \backsteparrow_{#1} #3}

\newcommand{\causalbacksteparrow}{\leftsquigarrow}
\newcommand{\jcausalbackstep}[3][]{#2 \causalbacksteparrow_{#1} #3}

\newcommand{\evalarrow}{\Downarrow}
\newcommand{\jeval}[3]{\tuple{#1,#2} \evalarrow #3}

\newcommand{\jincluded}[2]{#1 \sqsubseteq #2}

We define an operational semantics for \dCSNH{} executing a protocol $\topsys$
using small-step transitions of the form \smash{\fbox{$\jstep[\topsys]{\sysstate}{\sysstate'}$}}.
When the protocol $\topsys$ is clear from the context, we elide it and
simply write $\jstep{\sysstate}{\sysstate'}$. Following typical asynchronous actor models, the semantics handles actor-internal transitions, and the sending and receiving of messages.
 We assume the internal transition language $\instr$ to be parametric to the semantics, so long as the language cannot affect global state. Thus, an internal transition in the system updates an actor's control location, and possibly store, but no other state components.
\begin{mathpar}\small
\infer[Step-Instr]{
  \trans\colon \mktrans{\prelbl\loc}{\instr}{\postlbl\loc} \in \maplookup{\maplookup{\topsys}{\mach}}{\evnt}
  \and
  \jstep[\trans]{ \mkpstate{\prelbl\loc}{\actorstore} }{ \mkpstate{\postlbl\loc}{\actorstore'} }
}{
  \jstep[\topsys]{
    \mksysstate
      { (\ctxcons{\sysloc}{ \mapping{\addr}{ \mkmachsttevntloc{\mach}{\stt}{\evnt}{\prelbl\loc} } }) }
      { (\ctxcons{\sysstore}{\mapping{\addr}{\actorstore}}) }
      { \net }{\hist}
  }{
    \mksysstate
      { (\ctxcons{\sysloc}{ \mapping{\addr}{ \mkmachsttevntloc{\mach}{\stt}{\evnt}{\postlbl\loc} } }) }
      { (\ctxcons{\sysstore}{\mapping{\addr}{\actorstore'}}) }
      { \net }{\hist}
  }
}
\end{mathpar}

Unlike internal instructions, executing a $\sendkw$ command has visible effects on the network. If an actor $\fromlbl\addr$'s control location is the pre-location $\mkmachsttevntloc{\mach}{\stt}{\evnt}{\prelbl\loc}$ of a
$\sendkw$ command for event $\evnt'$, it steps to the post-location
$\postlbl\loc$ and adds a message $\msg\colon \left(\mkmsg{\fromlbl\addr}{
\tolbl\addr }{\evnt'}{ \seq{ \val } }\right)$ to the network buffer
\netcolor{$(\ctxcons{\net}{\msg})$}.
The recipient address $\tolbl\addr$ and payload $\seq{\val}$ arguments are retrieved from the sending actor's store $\actorstore$:
\begin{mathpar}\small
\infer[Step-Send]{
  \mktrans{\prelbl\loc}{ \mksend{\var}{\evnt'}{ \seq{\varalt} } }{\postlbl\loc}
  \in \maplookup{\maplookup{\topsys}{\mach}}{\evnt}
  \\
  \mapping{\fromlbl\addr}{\actorstore} \in \sysstore
  \and
  \msg = \left(\mkmsg{\fromlbl\addr}{ \tolbl\addr }{\evnt'}{ \seq{ \val } }\right)
  \and
  \tolbl\addr = \maplookup{\actorstore}{\var}
  \and
  \seq{ \val = \maplookup{\actorstore}{\varalt} }
}{
  \jstep{
    \mksysstate
      { (\ctxcons{\sysloc}{ \mapping{\fromlbl\addr}{ \mkmachsttevntloc{\mach}{\stt}{\evnt}{\prelbl\loc} } }) }
      { \sysstore }
      { \net }{\hist}
  }{
    \mksysstate
      { (\ctxcons{\sysloc}{ \mapping{\fromlbl\addr}{ \mkmachsttevntloc{\mach}{\stt}{\evnt}{\postlbl\loc} } }) }
      { \sysstore }
      { (\ctxcons{\net}{\msg}) }{\hist}
  }
}
\end{mathpar}
Messages leave the network when they are \emph{handled} by the receiving actor $\tolbl\addr$. In rule \TirName{Step-Handle}, we state that if an actor $\tolbl\addr$ is
not currently executing a handler (i.e., is at the exit location $\mkmachsttevntloc{\mach}{\stt}{\evnt}{\exitloc}$
of some procedure $\proc$ of event type $\evnt$), and a message $\msg$ is addressed to it
$\left(\mkmsg{\fromlbl\addr}{\tolbl\addr}{\evnt'}{ \seq{\val} }\right)$
which it can handle ($\mkonhand{\evnt'}{ \seq{\var} }{\proc} \in
\maplookup{\topsys}{\mach}$), then it handles the message by removing it
from the network $\netcolored$ and stepping to the entry location
$\entryloc$ of the handler procedure $\proc'$ of event type $\evnt'$\footnote{When we write a control location $\mkmachsttevntloc{\mach}{\stt}{\evnt}{\loc}$, we implicitly assume that it is a valid location in the distributed protocol $\topsys$ (i.e., $\loc$ is a location in the procedure $\maplookup{\maplookup{\topsys}{\mach}}{\evnt}$), though we make this well-formedness condition explicit in rule \TirName{Step-Handle} for clarity.}:
\begin{mathpar}\small
\infer[Step-Handle]{
  \mkonhand{\evnt}{ \seq{\varalt} }{\proc} \in \maplookup{\topsys}{\mach}
  \and
  \msg = \left(\mkmsg{\fromlbl\addr}{\tolbl\addr}{\evnt'}{ \seq{\val} }\right)
  \and
  \mkonhand{\evnt'}{ \seq{\var} }{\proc'} \in \maplookup{\topsys}{\mach}
  \\
  \sysstore =
      \left(\ctxcons{\sysstore_0}{ \mapping{\tolbl\addr}{ \mkactorstore{\fieldslbl\store}{\varslbl\store} } }\right)
  \and
  \sysstore' = \left(\ctxcons{\sysstore_0}{ \mapping{\tolbl\addr}{ \mkactorstore{\fieldslbl\store}{ (\seq{\mapping{\var}{\val}}) } } }\right)
  \and
  \jpaffine{ (\histcons{\hist}{\msg}) }
}{
  \jstep{
    \mksysstate
      { (\ctxcons{\sysloc}{ \mapping{\tolbl\addr}{ \mkmachsttevntloc{\mach}{\stt}{\evnt}{\exitloc} } }) }
      { \sysstore }
      { (\ctxcons{\net}{\msg}) }{\hist}
  }{
    \mksysstate
      { (\ctxcons{\sysloc}{ \mapping{\tolbl\addr}{ \mkmachsttevntloc{\mach}{\stt}{\evnt'}{\entryloc} } })}
      {\sysstore'}
      { \net }{ (\histcons{\hist}{\msg}) }
  }
}
\end{mathpar}

On stepping, we initialize the handler-local store with the values of the
message payload $\seq{\val}$, binding the formal parameters $\seq{\var}$ of
the handler to these values ($\seq{\mapping{\var}{\val}}$). On handling a
message, we also add the message $\msg$ to the message history $\histcolor{
\histcons{\hist}{\msg} }$ to keep track of the order in which messages were
handled.

We require additionally the extended history to be realizable under an affine protocol, denoted as $\jpaffine{ (\histcons{\hist}{\msg})}$:
\begin{inparaenum}
  \item it must not have messages sent by an actor before it has received the special $\initevnt$ message from the main machine $\mainmach$ and
  \item it must not contain a message that has been handled more than once by the same recipient nor two messages with the same event type $\evnt$ with different payloads from the same sender.
\end{inparaenum}
This condition is crucial as it allows the derivation of constraints on summarized message histories (see \autoref{sec:backwards-symbolic-execution-semantics}).

Under the affine model we can have transition rules for message loss but not for message duplication. This is not a significant limitation, as message deduplication can be implemented by the run-time system with message IDs. We elide the rule for message loss due to its simplicity.

\subsection{Goal-Directed Actor Materialization}\label{sec:abstract-semantics}

We now define a symbolic abstraction of the concrete
\dCSNH{} semantics using separation, linear, and ordered logics.
Crucially, this allows us to describe \emph{actor materialization}, which enables a backwards-from-error, goal-directed analysis of parametric actor systems. Each logic component is \emph{intuitionistic}: the lack of a constraint (such as $\loctop$) denotes no restriction, rather than empty.

\subsubsection{Symbolic Abstraction of Distributed Control.}\label{sec:symbolic-semantic-domains}

\begin{figure}\small
\begin{mathpar}
\text{dist. states} \quad \abssysstate
\bnfdef \mkabssysstate{\abssysloc}{\abssysstore}{\absnet}{\abshist}
\bnfalt \abssysstate_1 \lor \abssysstate_2
\bnfalt \bot

\text{pure constraints} \quad \purecolor{\puredom}
\bnfdef \absval_1 = \absval_2
\bnfalt \puredom_1 \land \puredom_2
\bnfalt \ltrue
\bnfalt \cdots

\text{values} \quad \absval \bnfdef \symvar \bnfalt \absaddr

\text{symbolic variables} \quad \symvar, \symvaralt

\text{actor addresses} \quad \absaddr

\text{control-location maps} \quad \abssysloc
\bnfdef \mapping{\absaddr}{ \mkmachsttevntloc{\mach}{\stt}{\evnt}{\loc} }
\bnfalt \abssysloc_1 \storesep \abssysloc_2
\bnfalt \loctop

\text{dist. stores} \quad \storecolor{\abssysstore}
\bnfdef \mapping{\absaddr}{\absactorstore}
\bnfalt \abssysstore_1 \storesep \abssysstore_2
\bnfalt \storetop

\text{network} \quad \netcolor{\absnet} \bnfdef
\absmsg
\bnfalt \absnet_1 \netsep \absnet_2
\bnfalt \absnetnil

\text{message histories} \quad \histcolor{\abshist}
\bnfdef \abshistemp
\bnfalt \abshisthyp{\absmsg}{\abshist}

\text{messages} \quad \absmsg \bnfdef \mkmsg{\fromlbl\absaddr}{\tolbl\absaddr}{\evnt}{ \seq{\absval} }

\text{actor stores} \quad \storecolor{\absactorstore} \bnfdef \mkactorstore{\fieldslbl\absstore}{\varslbl\absstore}
\bnfalt \absactorstore_1 \lor \absactorstore_2
\bnfalt \bot

\text{stores} \quad \storecolor{\absstore}
\bnfdef \mapping{\var}{\absval}
\bnfalt \absstore_1 \storesep \absstore_2
\bnfalt \storetop
\end{mathpar}
\caption{Symbolically abstracting a \dCSNH{} state using separation, linear, and ordered logics.}
\label{fig:abstract-semantic-domains}
\end{figure}

\autoref{fig:abstract-semantic-domains} introduces a symbolic abstraction
of the concrete \dCSNH{} state defined in \autoref{fig:semantic-domains}. We use
``hatted'' meta-variables to refer to abstract counterparts of the concrete semantic domains.
For example, $\abssysstate$ is an abstract state corresponding to a concrete state $\sysstate$. In \autoref{sec:motivation}, these ``hats'' were omitted to reduce clutter.
Abstract states
\smash{$\abssysstate \bnfdef
\mkabssysstate{\abssysloc}{\abssysstore}{\absnet}{\abshist} \bnfalt
\cdots$} differ from the concrete states by the addition of pure
constraints \purecolor{$\puredom$}. These constraints, written in propositional logic, capture the relationships
between symbolic values $\absval$ collected during analysis and may relate values sent between actors, rather than just relations among an individual actor's fields in an actor-local state.
Symbolic values $\absval \bnfdef \symvar \bnfalt \absaddr$ are simply either symbolic variables $\symvar$ or symbolic actor addresses $\absaddr$.

\subparagraph{Abstract Locations and Stores.}

Recall that a concrete location map $\sysloc$ is a finite --- but unbounded --- map from actor addresses $\addr$ to their control locations $\mkmachsttevntloc{\mach}{\stt}{\evnt}{\loc}$.
We define the abstract counterpart to these locations as a separation logic formula
\smash{$\abssysloc \bnfdef
\mapping{\absaddr}{ \mkmachsttevntloc{\mach}{\stt}{\evnt}{\loc} }
\bnfalt \abssysloc_1 \storesep \abssysloc_2
\bnfalt \loctop$}
where
$\mapping{\absaddr}{ \mkmachsttevntloc{\mach}{\stt}{\evnt}{\loc} }$ is a
materialized actor $\absaddr$ at a concrete control location
$\mkmachsttevntloc{\mach}{\stt}{\evnt}{\loc}$,
\smash{$\abssysloc_1
\storesep \abssysloc_2$} is a separating conjunction of two abstract
location maps, and
$\loctop$ is a top element representing an
summary of an arbitrary number of actors at arbitrary locations.

An abstract distributed store
\smash{\storecolor{$\abssysstore$}} is a separation logic formula with
points-to constraints of the form
\smash{$\mapping{\absaddr}{\absactorstore}$}. The formalization maintains
the invariant that \smash{$\dom(\abssysloc) = \dom(\abssysstore)$}, since
both \smash{$\abssysloc$} and \smash{\storecolor{$\abssysstore$}}
constrain the same set of materialized actor addresses $\absaddr$.
An abstract actor-local store $\absactorstore$ is the tuple
\smash{$\mkactorstore{\fieldslbl\absstore}{\varslbl\absstore}$}, where
\smash{$\fieldslbl\absstore$} and \smash{$\varslbl\absstore$} are
separation logic formulas denoting the stores for actor's fields and
procedure-local variables, respectively.

\subparagraph{Abstract Networks and Histories.}

An abstract message $\absmsg \bnfdef \mkmsg{\fromlbl\absaddr}{\tolbl\absaddr}{\evnt}{
\seq{\absval} }$  is just like a concrete message $\msg$, but with addresses and arguments made symbolic. We model abstract networks
$\netcolor{\absnet}
\bnfdef \absmsg
\bnfalt \absnet_1 \netsep \absnet_2
\bnfalt \absnetnil$
as an unordered linear resource of abstract messages $\absmsg$ produced and consumed during protocol execution. A backwards-from-error symbolic execution starts
with $\absnetnil$ as the  network buffer denoting an an
arbitrary number of unknown messages. As the execution progresses
backwards, we may materialize specific abstract messages $\absmsg$, adding them to the network  via linear conjunction  $\netsep$. The abstract network \smash{$\absnet_1 \netsep \absnet_2$} is a disjoint union of two message sets. Abstract message histories \smash{$\histcolor{\abshist}
\bnfdef \abshistemp \bnfalt \abshisthyp{\absmsg}{\abshist}$} are an ordered
list of materialized abstract messages. We use the special history $\abshistemp$ to denote the set of all realizable affine message histories.

\subsubsection{Concretization.}

Now we formally connect symbolic and concrete states. We write $\concstate{\abssysstate}{\sysstate}$ to mean that the concrete state $\sysstate$ is in the concretization of the abstract state $\abssysstate$.
We define this concretization relation by induction on the structure of $\abssysstate$, with an individual state being defined component-wise:
\begin{mathpar}\small
\begin{array}{l}
\concstate{\mkabssysstate{\abssysloc}{\abssysstore}{\absnet}{\abshist}}{\mksysstate{\sysloc}{\sysstore}{\net}{\hist}}
\qquad\text{iff}
\\
\qquad\text{%
$\conc{\abssysloc}{\sysloc}$
and
$\conc{\storecolor{\abssysstore}}{\storecolor{\sysstore}}$
and
$\conc{\netcolor{\absnet}}{\netcolor{\net}}$
and
$\conc{\histcolor{\abshist}}{\histcolor{\hist}}$
and
$\concpure{\purecolor{\puredom}}$
s.t. $\dom(\sysloc) = \dom(\storecolor{\sysstore})$
for some \purecolor{$\valuation$}
}
\end{array}
\end{mathpar}
The valuation \purecolor{$\valuation$} maps symbolic variables $\symvar$ and addresses $\absaddr$ to concrete values $\val$ and addresses $\addr$, respectively. The constraint $\dom(\sysloc) = \dom(\storecolor{\sysstore})$ captures the invariant that we have one set of actors.
Concretizing disjunctive states are as expected:
\begin{mathpar}\small
\text{$\concstate{\abssysstate_1 \lor \abssysstate_2}{\sysstate}$}
\quad \text{iff} \quad
\text{$\concstate{\abssysstate_1}{\sysstate}$ or $\concstate{\abssysstate_2}{\sysstate}$}

\text{$\concstate{\bot}{\sysstate}$}
\quad \text{never}
\end{mathpar}

\subparagraph{Concretizing Actor Locations and Stores.}
We begin by defining concretization of actor locations and stores, which both use separation logic.
\begin{mathpar}\small
\text{$\conc{ \mapping{\absaddr}{ \mkmachsttevntloc{\mach}{\stt}{\evnt}{\loc} } }{\sysloc}$}
\quad \text{iff} \quad
\text{$\mapping{ \maplookup{\purecolor{\valuation}}{\absaddr} }{ \mkmachsttevntloc{\mach}{\stt}{\evnt}{\loc} } \in \sysloc$ or ($\evnt = \idleevnt$ and $\maplookup{\purecolor{\valuation}}{\absaddr} \in \dom(\sysloc)$)}
\end{mathpar}
We use the $\idleevnt$ event in the abstraction to represent an actor at an arbitrary control location, which we will leverage further in \autoref{sec:reflogic}.
An abstract actor location \smash{$\abssysloc$} is a separation logic formula with a disjointness constraint on the actor addresses:
\begin{mathpar}\small
\begin{array}{rcl}
\text{$\conc{ \abssysloc_1 \storesep \abssysloc_2 }{ \ctxcons{\sysloc_1}{\sysloc_2} }$}
& \text{iff} &
\text{$\conc{\abssysloc_1}{\sysloc_1}$ and $\conc{\abssysloc_2}{\sysloc_2}$ s.t. $\dom(\sysloc_1) \cap \dom(\sysloc_2) = \emptyset$}
\\
\text{$\conc{ \loctop }{ \sysloc}$}
& \text{always}
\end{array}
\end{mathpar}
Distributed store concretization  $\conc{\storecolor{\abssysstore}}{\storecolor{\sysstore}}$ is analogous to the above, so we elide it for brevity. 

\subparagraph{Concretizing Networks and Histories.}
An abstract message $\absmsg$ simply concretizes to a concrete message $\msg$ by mapping symbolic values to concrete values via the valuation \purecolor{$\valuation$}:
\begin{mathpar}\small
\text{$\conc
{ \mkmsg{\fromlbl\absaddr}{\tolbl\absaddr}{\evnt}{ \seq{\absval} } }
{ \mkmsg{ \maplookup{\purecolor{\valuation}}{\fromlbl\absaddr} }{ \maplookup{\purecolor{\valuation}}{\tolbl\absaddr} }{\evnt}{ \seq{ \maplookup{\purecolor{\valuation}}{\absval} } } }
$}
\quad \text{always}
\end{mathpar}
An abstract network represents a set of pending messages, which must correspond to a subset of the concrete network:
\begin{mathpar}\small
\text{$\conc
{ \netcolor{\absmsg} }
{ \net }
$}
\quad \text{iff} \quad
\text{$\msg \in \net$ for some $\msg$ s.t. $\conc{\absmsg}{\msg}$}

\text{$\conc
{ \netcolor{\absnet_1 \netsep \absnet_2} }
{ \netcons{\net_1}{\net_2}  }
$}
\quad \text{iff} \quad
\text{$\conc{\netcolor{\absnet_1}}{\net_1}$ and $\conc{\netcolor{\absnet_2}}{\net_2}$ s.t. $\net_1 \cap \net_2 = \emptyset$}

\text{$\conc
{ \netcolor{\nettop} }
{ \net }
$}
\quad \text{always}
\end{mathpar}

An abstract message history $\abshist$ captures realizable message histories $\hist$, which must be affine:
\begin{mathpar}\small
\text{$\conc
{ \histcolor{\abshistemp} }
{ \hist }
$}
\quad \text{iff} \quad
\text{$\jaffine{\hist}$}

\text{$\conc
{ \histcolor{\abshisthyp{\absmsg}{\abshist}} }
{ \hist  }
$}
\quad \text{iff} \quad
\text{$\conc{\absmsg}{\msg}$ implies $\conc{\histcolor{\abshist}}{ \histcons{\hist}{\msg} }$}
\end{mathpar}

The hypothetical history $\abshisthyp{\absmsg}{\abshist}$ is interpreted as an implication: if the next handled message belongs to the concretization of $\absmsg$, then the extended history must satisfy $\abshist$. The abstract history $\histcolor{\abshistemp}$ represents all affine concrete histories ($\jaffine{\hist}$).

\subsubsection{Backwards Symbolic Execution Semantics.}\label{sec:backwards-symbolic-execution-semantics}

\newcommand{\mkruleabs}[1]{\textsc{Abs-Step-#1}}

We have now introduced a concrete system, as well as an abstraction which concretizes to sets of states. We now define the semantics of backwards symbolic execution of a
distributed protocol as a transition relation on abstract states $\abssysstate$. In particular, we consider the judgment form \smash{\fbox{$\jbackstep[\topsys]{\abssysstate'}{\abssysstate}$}} that says, ``Abstract state $\abssysstate$ steps backwards to abstract state $\abssysstate'$ under protocol $\topsys$.''
The rules we present here first are essentially direct abstractions of the concrete transition rules from \autoref{sec:transition-system-semantics}:
\begin{mathpar}\small
\infer[Abs-Step-Instr]{
  \trans\colon \mktrans{\prelbl\loc}{\instr}{\postlbl\loc}
  \in \maplookup{\maplookup{\topsys}{\mach}}{\evnt} \and 
  \jbackstep[\trans]{
    \mkabspstate[\puredom']{\prelbl\loc}{\absactorstore'}
  }{ 
    \mkabspstate{\postlbl\loc}{\absactorstore}
  }
}{
  \jbackstep{
    \mkabssysstate[\puredom']{\absstorecons{\abssysloc}{\mapping{\absaddr}{\mkmachsttevntloc{\mach}{\stt}{\evnt}{\prelbl\loc}}}}{(\absstorecons{\abssysstore}{\mapping{\absaddr}{\absactorstore'}})}{\absnet}{\abshist}
  }{
    \mkabssysstate{\absstorecons{\abssysloc}{\mapping{\absaddr}{\mkmachsttevntloc{\mach}{\stt}{\evnt}{\postlbl\loc}}}}{(\absstorecons{\abssysstore}{\mapping{\absaddr}{\absactorstore}})}{\absnet}{\abshist}
  }
}
\end{mathpar}
Like the concrete semantics, we assume an abstract semantics $\jbackstep[\trans]{ \mkabspstate[\puredom']{\prelbl\loc}{\absactorstore'} }{ \mkabspstate{\postlbl\loc}{\absactorstore} }$ for actor-internal transitions on instructions $\instr$.
Notice that the pure constraints \purecolor{$\puredom'$}  are propagated from the local state to the global, allowing for a relation of values between actors necessary for deriving contradictions (cf. \autoref{sec:motivation}).

Reading the semantics backwards,
the abstract version of the \infrule{Step-Send} rule removes a message from the network. Intuitively, if we witness a message \netcolor{$\absmsg$} in the network buffer
\netcolor{$\absnetcons{\absnet}{\absmsg}$}, then there must have been an
actor $\fromlbl\absaddr$ at a location
$\mkmachsttevntloc{\mach}{\stt}{\evnt}{\prelbl\loc}$ that could have sent
it (\smash{$\mktrans{\prelbl\loc}{ \mksend{\var}{\evnt'}{ \seq{\varalt} }
}{\postlbl\loc} \in \maplookup{\maplookup{\topsys}{\mach}}{\evnt}$}).
We remove $\absmsg$ from the network, move $\fromlbl{\absaddr}$ to
$\prelbl{\loc}$, and update the pure constraints \purecolor{$\puredom'$} to
relate the formal parameters and actual arguments of \sendkw:

\begin{mathpar}\small
\infer[Abs-Step-Send]{
  \mktrans{\prelbl\loc}{ \mksend{\var}{\evnt'}{ \seq{\varalt} } }{\postlbl\loc}
  \in \maplookup{\maplookup{\topsys}{\mach}}{\evnt}
  \\
  \mapping{\fromlbl\absaddr}{\absactorstore} \in \abssysstore
  \and
  \absmsg = \left(\mkmsg{\fromlbl\absaddr}{ \tolbl\absaddr }{\evnt'}
    { \seq{ \absval  } }\right)
  \and  
  \puredom' = \left(\puredom \land \tolbl\absaddr = \maplookup{\absactorstore}{\var} \land \bigwedge\seq{ \absval = \maplookup{\absactorstore}{\varalt} }\right)
}{
  \jbackstep{
    \mkabssysstate[ \puredom' ]
      { (\absstorecons{\abssysloc}{ \mapping{\fromlbl\absaddr}{ \mkmachsttevntloc{\mach}{\stt}{\evnt}{\prelbl\loc} } }) }
      { \abssysstore }
      { \absnet }{ \abshist }
  }{
    \mkabssysstate
      { (\absstorecons{\abssysloc}{ \mapping{\fromlbl\absaddr}{ \mkmachsttevntloc{\mach}{\stt}{\evnt}{\postlbl\loc} } }) }
      { \abssysstore }
      { (\absnetcons{\absnet}{\absmsg}) }{\abshist}
  }
}
\end{mathpar}

We defer presenting an abstract version of \TirName{Step-Handle} until \autoref{sec:reflogic}, as it pertains to actor-causal reduction.

\subsubsection{Actor Location Materialization.}\label{sec:distributed-control-location-materialization}

For the \TirName{Abs-Step-Send} rule to be applicable, the abstract sender $\fromlbl\absaddr$ must be materialized at the $\sendkw$'s post-location (${\mapping{\fromlbl\absaddr}{ \mkmachsttevntloc{\mach}{\stt}{\evnt}{\postlbl\loc} }}$) (or similarly, for \TirName{Abs-Step-Instr} to be applicable).
However, the abstract actor $\fromlbl\absaddr$ of interest may correspond to a concrete actor summarized in $\loctop$. Thus, with only these rules, symbolic execution can become stuck unless all relevant actors have already been materialized. While this is sufficient, it requires knowledge which may be unknown a priori.
We address this by introducing actor materialization as a mechanism to constrain abstract location maps \smash{$\abssysloc$} with additional actors at particular locations. This enables a goal-directed analysis to materialize new actors as needed to progress.

We define a reflexive and transitive entailment judgment
$\jincluded{\abssysstate}{\abssysstate'}$  which states that the abstract state
$\abssysstate$ is included in $\abssysstate'$, meaning its concretization is a subset of the latter's. Entailment allows a rule analogous to the rule of consequence, allowing a strengthening of the pre-state:
\begin{mathpar}\small
\infer[Abs-Step-Consequence]{
  \jincluded{\postlbl\abssysstate}{\postlbl\abssysstate'}
  \and
  \jbackstep{ \prelbl\abssysstate' }{ \postlbl\abssysstate' }
  \and
  \jincluded{\prelbl\abssysstate'}{\prelbl\abssysstate}
}{
  \jbackstep{ \prelbl\abssysstate }{ \postlbl\abssysstate }
}
\end{mathpar}

This strengthening allows us to define materialization by constraining an additional symbolic actor $\absaddr'$. In \TirName{Materialize-Actor}, the new actor is materialized at a control location $\mkmachsttevntloc{\mach}{\stt}{\evnt}{\loc}$. This rule has two cases: either it aliases an actor $\absaddr_i$ that is already materialized at that control location, or it is a distinct new actor ($\smash{\abssysloc} \storesep (\mapping{\absaddr'}{ \mkmachsttevntloc{\mach}{\stt}{\evnt}{\loc} })$).
\begin{mathpar}\small
\infer[Materialize-Actor]{
  \abssysstate =
    \left(\mkabssysstate
      { \abssysloc }
      { \abssysstore }
      { \absnet }
      { \abshist }\right)
  \and
  \absaddr' \notin \dom(\abssysloc)
  \and
  \text{(for all $\mapping{\absaddr_i}{ \mkmachsttevntloc{\mach}{\stt}{\evnt}{\loc} } \in \abssysloc$)}
}{
  \jincluded{
    \abssysstate
  }{
    \left(
    \left(\bigvee_{i} (\abssysstate\purecolor{{} \land \absaddr_i = \absaddr'})\right)
    \;\lor\;
    \mkabssysstate
      { (\abssysloc \storesep (\mapping{\absaddr'}{ \mkmachsttevntloc{\mach}{\stt}{\evnt}{\loc} })) }
      { (\abssysstore \storesep (\mapping{\absaddr'}{\storetop})) }
      { \absnet }
      { \abshist }\right)
  }
}
\end{mathpar}

\newcommand{\cons}[2]{#1\cdot#2}
\newcommand{\overapprox}[2]{#1\vDash#2}
\newcommand{\mainactorinit}{Main.\initstt.eInit.\entryloc}
\newcommand{\initsystate}{\sysstate_{init}}

\newcommand{\jbackstepn}[3]{#2 \backsteparrow^{#1} #3}
\newcommand{\jstepn}[3]{#2 \steparrow^{#1} #3}

\newcommand{\unreachable}[1]{\vdash_{unreach} #1}
\newcommand{\initial}[1]{\vdash_{init} #1}

\subsubsection{Local Soundness.}\label{sec:local-soundness}

\newcommand{\inv}{\hat{\Sigma}}

The abstract semantics we have so far are essentially a direct symbolic abstraction of the concrete semantics, so local soundness is direct:
\begin{lemma}[Subsumption Soundness]\label{lem:subsumption-soundness}
If $\jincluded{\abssysstate}{\abssysstate'}$ and $\concstate{\abssysstate}{\sysstate}$, then $\concstate{\abssysstate'}{\sysstate}$.
\end{lemma}  
\begin{lemma}[Local Soundness of Backwards Symbolic Execution]\label{lem:local-soundness}
If
$\jstep[\topsys]{\prelbl\sysstate}{\postlbl\sysstate}$
and
$\jbackstep[\topsys]{\prelbl\abssysstate}{\postlbl\abssysstate}$
such that
$\concstate{\postlbl\abssysstate}{\postlbl\sysstate}$,
then $\concstate{\prelbl\abssysstate}{\prelbl\sysstate}$.
\end{lemma}

These lemmas are corollaries of Theorem~\ref{thm:causal-soundness} which is
mechanized in Lean. \autoref{sec:mechanization} contains proof sketches.

\section{Actor-Causality Reduction}\label{sec:reflogic}

With materialization (\autoref{sec:actorlogic}), we can dynamically instantiate individual actors during backwards symbolic execution.
Applying it too frequently can rapidly blow up the state space, yet it is also necessary for precision, making balance between the two key to our technique's practicality.

Following the intuition in \autoref{sec:actor-causality-reduction}, we reduce the non-de\-ter\-min\-ism of our abstract semantics by leveraging our causally consistent model and following backwards only causal dependencies that are \emph{relevant} to refuting a given error goal.
Even though any message could have been the last one handled previously, we can leverage our error-goal state to restrict our choices to only messages relevant to refuting the state. For instance, if a message handler neither updates any fields in the goal nor sends any messages that are currently present in the network, it can be ignored without affecting refutability.
While this does not remove non-determinism entirely, it provides a substantive reduction to the search space, especially as the program grows in scale.

To formalize this reduction, we define a causal-reduced transition relation $\causalbacksteparrow$
that extends the transition relation $\backsteparrow$ from \autoref{sec:backwards-symbolic-execution-semantics}.
Specifically, we consider the judgment form \smash{\fbox{$\jcausalbackstep[\topsys]{\abssysstate'}{\abssysstate}$}} that says, ``$\abssysstate'$ can be reached backwards from the state $\abssysstate$ with a backwards causally-reduced step in protocol $\topsys$.''

\subparagraph{Message Handling.}

We first define a rule for message handling. This rule is close to the concrete transition semantics, save for a single detail: the pre-location is the exit location of the distinguished idle event $\idleevnt$, rather than an arbitrary one.
\begin{mathpar}\small
\infer[Abs-Step-Handle]{
  \absmsg = \left(\mkmsg{\fromlbl\absaddr}{\tolbl\absaddr}{\evnt'}{ \seq{\absval} }\right)
  \and
  \mkonhand{\evnt'}{ \seq{\var} }{\proc} \in \maplookup{\topsys}{\mach}
  \\
  \abssysstore' =
      \left(\absstorecons{\abssysstore_0}{ \mapping{\tolbl\absaddr}{ \mkactorstore{\fieldslbl\absstore}{\storetop} } } \right)
  \and
  \abssysstore =
      { \left(\absstorecons{\abssysstore_0}{ \mapping{\tolbl\absaddr}{ \mkactorstore{\fieldslbl\absstore}{ (\seq{\mapping{\var}{\absval}}) } } }\right) }
  \and
  \purecolor{\puredom'} = \purecolor{\puredom \land \paffinefun{ \abshisthyp{\absmsg}{\abshist} }}
}{
  \jcausalbackstep{
    \mkabssysstate[\puredom']
      { (\absstorecons{\abssysloc}{ \mapping{\tolbl\absaddr}{ \mkmachsttevntloc{\mach}{\stt}{\idleevnt}{\exitloc} } }) }
      { \abssysstore' }
      { (\absnetcons{\absnet}{\absmsg}) }
      { \abshisthyp{\absmsg}{\abshist} }
  }{
    \mkabssysstate
      { (\absstorecons{\abssysloc}{ \mapping{\tolbl\absaddr}{ \mkmachsttevntloc{\mach}{\stt}{\evnt'}{\entryloc} } })}
      { \abssysstore }
      { \absnet }
      { \abshist  }
  }
}
\end{mathpar}
The use of the $\idleevnt$ event prevents branching that would occur if our abstract rule could step back into any handler. Recall that this location is effectively the $\exitloc$ of \emph{any} handler, so we are deferring the choice of previous handler to another rule.
On the backwards transition, the abstract history \histcolor{$\abshist$} is updated to
\histcolor{$\abshisthyp{\absmsg}{\abshist}$} to constrain the possible message history with the handling of $\absmsg$.
Note that when a message is handled, we also update the pure constraint $\puredom$ to assert the affine property on the updated message history (\purecolor{$\paffinefun{ \abshisthyp{\absmsg}{\abshist} }$}). This ensures that as the analysis proceeds, the history considered remains valid under our affine assumption.
Now, we must define additional rules to allow the system to step back from the $\idleevnt$ event to locations that are \emph{relevant} to our goal state.

\subparagraph{Store-Constraint Causality.}

One way a handler may be relevant is by setting a field that is relevant to a constraint. We use this \emph{store-constraint causality} to step back to a message handler that assigns a value to this field $\var$ (i.e., a transition $\mktrans{\prelbl\loc}{\mkassign{\var}{\expr}}{\postlbl\loc} \in \maplookup{\maplookup{\topsys}{\mach}}{\evnt}$).
This corresponds to a form of data-relevance constraint~\cite{hopper} adapted to actor systems:
\begin{mathpar}\small
\infer[Causal-Step-Store-Constraint]{
 \mktrans{\prelbl\loc}{\mkassign{\var}{\expr}}{\postlbl\loc} \in \maplookup{\maplookup{\topsys}{\mach}}{\evnt}
 \and 
 \mapping{\var}{\symvar} \in \fieldslbl\absstore
 \and 
 \mapping{\absaddr}{\mkactorstore{\fieldslbl\absstore}{\storetop}} \in \abssysstore
}{
    \jcausalbackstep{
    \mkabssysstate{(\absstorecons{\abssysloc}{\mapping{\absaddr}{\mkmachsttevntloc{\mach}{\stt}{\evnt}{\postlbl\loc}}})}
    {\abssysstore}{\absnet}{\abshist}
    }{
     \mkabssysstate{(\absstorecons{\abssysloc}{\mapping{\absaddr}{\mkmachsttevntloc{\mach}{\stt}{\idleevnt}{\exitloc}}})}{\abssysstore}{\absnet}{\abshist}
    }
}
\end{mathpar}
Conversely, if an actor field is unused by any pure constraints $\puredom$, then it is irrelevant for refuting the goal state. 
The pre-state differs only in that the actor has now jumped back to the post-location $\mapping{\absaddr}{\mkmachsttevntloc{\mach}{\stt}{\evnt}{\postlbl\loc}}$ corresponding to a witnessing assignment $\mktrans{\prelbl\loc}{\mkassign{\var}{\expr}}{\postlbl\loc}$.

\subparagraph{Network Causality.}

When there is a message $\absmsg$ in the network component (\netcolor{$\absnetcons{\absnet}{\absmsg}$}), it must have been sent by an actor, and this sending may be relevant to refuting the error goal.
We call this \emph{network causality}:
\begin{mathpar}\small
\infer[Causal-Step-Network]{
  \mktrans{\prelbl\loc}{\mksend{\var}{\evnt'}{\seq{\varalt}}}{\postlbl\loc} \in \maplookup{\maplookup{\topsys}{\mach}}{\evnt}
  \and 
  \absmsg = \mkmsg{\fromlbl\absaddr}{\tolbl\absaddr}{\evnt'}{\seq{\symvar}}
}{
  \jcausalbackstep{
    \mkabssysstate{(\absstorecons{\abssysloc}{\mapping{\fromlbl\absaddr}{\mkmachsttevntloc{\mach}{\stt}{\evnt}{\postlbl\loc}}})}
    {\abssysstore}{ (\absnetcons{\absnet}{\absmsg}) }{\abshist}
  }{
     \mkabssysstate{(\absstorecons{\abssysloc}{\mapping{\fromlbl\absaddr}{\mkmachsttevntloc{\mach}{\stt}{\idleevnt}{\exitloc}}})}{\abssysstore}{ (\absnetcons{\absnet}{\absmsg}) }{\abshist}
  }
}
\end{mathpar}
In \TirName{Causal-Step-Network}, the control location of the materialized actor $\fromlbl\absaddr$ jumps back to the post-location $\mapping{\fromlbl\absaddr}{\mkmachsttevntloc{\mach}{\stt}{\evnt}{\postlbl\loc}}$ for a $\sendkw$ transition $\mktrans{\prelbl\loc}{\mksend{\var}{\evnt'}{\seq{\symvar}}}{\postlbl\loc}$
.

\paragraph{Causal Materialization.}

Actor materialization is necessary for precision but can lead to a rapid explosion in the state space.
The \TirName{Materialize-Actor} rule from \autoref{sec:distributed-control-location-materialization} allows arbitrary materialization, even though most choices would be irrelevant. We recognize that the goal state, which captures the relevant system slice already, can be used to guide materialization, allowing effective automation.
We now leverage our two new inference rules to determine more specific materialization conditions.

In the case of the \TirName{Causal-Step-Store-Constraint} rule, actor $\absaddr$ has a materialized field, so the invariant guarantees that it is already present in the location map $\abssysloc$, and thus additional materialization is unnecessary.
For the \TirName{Causal-Step-Network} rule, the receiving actor $\tolbl\absaddr$ must already be present from message handling (cf. \TirName{Abs-Step-Handle}), but the sender $\fromlbl\absaddr$ may not be materialized. We define a materialization rule that is constrained by a message \netcolor{$\absmsg$} in the network and only materializes a sending actor $\fromlbl\absaddr$ that could have sent this message ($\mktrans{\prelbl\loc}{ \mksend{\var}{\evnt'}{ \seq{\varalt} } }{\postlbl\loc} \in \maplookup{\maplookup{\topsys}{\mach}}{\evnt}$):
\begin{mathpar}\small
\infer[Causal-Materialize-Actor-Network]{
  \mktrans{\prelbl\loc}{ \mksend{\var}{\evnt'}{ \seq{\varalt} } }{\postlbl\loc}
  \in \maplookup{\maplookup{\topsys}{\mach}}{\evnt}
  \and 
  \absmsg = \mkmsg{\fromlbl\absaddr}{\tolbl\absaddr}{\evnt'}{\seq{\symvar}}
  \\
  \abssysstate =
    \left(\mkabssysstate
      { \abssysloc }
      { \abssysstore }
      { (\absnetcons{\absnet}{\absmsg}) }
      { \abshist }\right)
  \and
  \fromlbl\absaddr \notin \dom(\abssysloc)
  \and
  \text{(for all $\mapping{\absaddr_i}{ \mkmachsttevntloc{\mach}{\stt}{\evnt}{\postlbl\loc} } \in \abssysloc$)}
}{
  \jincluded{
    \abssysstate
  }{
    \left(
    \left(\bigvee_{i} (\abssysstate\purecolor{{} \land \absaddr_i = \fromlbl\absaddr})\right)
    \;\lor\;
    \mkabssysstate
      { (\abssysloc \storesep (\mapping{\fromlbl\absaddr}{ \mkmachsttevntloc{\mach}{\stt}{\evnt}{\postlbl\loc} })) }
      { (\abssysstore \storesep (\mapping{\fromlbl\absaddr}{\storetop})) }
      { (\absnetcons{\absnet}{\absmsg}) }
      { \abshist }\right)
  }
}
\end{mathpar}

\paragraph{Causal Soundness.}
\newcommand{\jcausalbackstepn}[3][*]{#2 \causalbacksteparrow_\pi^{#1} #3}
\newcommand{\predtrace}{\tau}
\newcommand{\predtracecons}[2][\predtrace]{#1 #2}
\newcommand{\causalinv}{\hat{\Sigma}}
\newcommand{\predlbl}[1]{\sublbl{#1}{pred}}
Our causally-reduced transition systems makes the soundness condition more complex than before.  Previously, the abstract system could precisely mimic each step of the concrete system, but now some steps are impossible, as irrelevant actors are skipped over. Intuitively, a causal step jumps back past an arbitrary number of irrelevant predecessor states to a relevant one. This resembles stuttering simulations~\cite{namjoshi}, where sequences of concrete transitions that are irrelevant to the abstract behavior are collapsed.
To describe this notion of jumping, we instrument our concrete semantics to keep track of predecessors, defining $\jstep[\topsys]{\predtrace}{\predtrace'}$ that says, ``The trace $\predtrace$ can be extended to $\predtrace'$ in protocol $\topsys$.'':
\begin{mathpar}\small
\text{ traces}\quad \predtrace
\bnfdef \cdot
\bnfalt \predtracecons{\sysstate}

\fbox{$\jstep[\topsys]{\predtrace}{\predtrace'}$}

\infer{
  \jstep[\topsys]{\sysstate}{\sysstate'}
}{
  \jstep[\topsys]{ \predtracecons{\sysstate} }{ \predtracecons[ \predtracecons{\sysstate} ]{\sysstate'} }
}
\end{mathpar}

We now extend our concretization relation to relate concrete predecessor traces and abstract states. A trace $\predtrace\sysstate$ is in the concretization of an abstract state $\abssysstate$ iff either the current concrete state $\sysstate$ is in the concretization of $\abssysstate$, \emph{or} if one of its predecessors $\predlbl\sysstate \in \predtrace$ is in the concretization of $\abssysstate$:
\begin{mathpar}\small
\text{$\concstate{ \abssysstate }{ \predtracecons{\sysstate} }$}
\quad \text{iff} \quad
\text{$\concstate{\abssysstate}{\sysstate}$ or $\concstate{\abssysstate}{\predlbl\sysstate}$ for some $\predlbl\sysstate \in \predtrace$}
\end{mathpar}

Intuitively, a predecessor trace may contain additional states, as long as some state in the trace is in the concretization of the abstract state.
Causal soundness then states that a causal step can jump back past irrelevant states to a relevant one:
\begin{theorem}[\thmcolor{Causal Soundness}]\label{thm:causal-soundness}
If
$\jstep[\topsys]{\prelbl\predtrace}{\postlbl\predtrace}$
and
$\jcausalbackstep[\topsys]{\prelbl\abssysstate}{\postlbl\abssysstate}$
such that
$\concstate{\postlbl\abssysstate}{\postlbl\predtrace}$,
then $\concstate{\prelbl\abssysstate}{\prelbl\predtrace}$.
\end{theorem}
We mechanize a proof of this theorem in Lean. The \thmcolor{purple} color-coding
in a theorem statement indicates mechanization. For a proof sketch, see
\autoref{sec:mechanization}.

\section{Message-Segment Summarization}\label{sec:parametric}

Leveraging the causal-reduced transition system defined in
\autoref{sec:reflogic}, observe that it is possible to cycle with the same
materialized actor
$\smash{\mapping{\absaddr}{\mkmachsttevntloc{\mach}{\stt}{\evnt}{\loc}} \in
\abssysloc}$ using the \TirName{Causal-Step-Store-Constraint} and
\TirName{Causal-Step-Network} rules. In this section, we address this challenge by discussing various widening strategies for ensuring termination. In particular, we define \emph{message-segment summarization} to abstract over a core class of cycles observed in asynchronous distributed protocols. We formalize the causes of non-termination via \emph{causal dependencies}, constructing \emph{causal-dependence graphs} of protocols, and prove that our techniques guarantee termination on a subset of these graphs.

Consider the $\fmtcode{eVote}$ handler from \autoref{sec:message-segments}. With the rule \TirName{Causal-Step-Store-Constraint} and no other mechanism, our analysis would cycle attempting to witness the assignment to the $\fmtcode{forMe}$ vote tally, in the handler of the $\fmtcode{eVote}$ event.
The number of \histcolor{$\exfmtcode{eVote}$} messages it receives --- and then handles via the $\TirName{Abs-Step-Handle}$ rule --- is dependent on the number of actors in the system, a fixed but unbounded parameter \purecolor{$\exsyssize$}. Thus, the network would continually grow, and each message leads to new materialization, growing the state space, and producing more interleavings to explore.

\subparagraph{Widening.}
We can enforce termination with a coarse widening~\cite{cousot}
by dropping the offending store constraint  $\mapping{\var}{\symvar} \in \fieldslbl\absstore$ ~\cite{DBLP:conf/pldi/BlackshearCS13}
or in-network message constraint
$\absmsg \in \absnet$
to weaken the abstract state $\abssysstate$.
While this naive operation breaks the cycle, it may also eliminate information needed to derive contradictions. In our running example (\autoref{sec:motivation}), this would prevent the analysis from deriving the contradiction on the $\fmtcode{forMe}$ vote tally because it drops the constraints on the $\fmtcode{forMe}$ field (namely, $\exfld{1}{forMe}$ from \autoref{sec:message-segments}).
Rather than discarding constraints, we exploit the structure of our parameterized affine actor model, defining a notion of \emph{message-segment summarization} to summarize the effect of a sequence of messages, in certain cases.
In our example, message-segment summarization handles the $\fmtcode{forMe}$ count from the sequence of \histcolor{$\exfmtcode{eVote}$} messages sent from a set of distinct actors (\exref{ex:cycling-message-history}).

To leverage this insight, we define a notion of \emph{message segments}, which we then use to extend our abstract domain:

\begin{mathpar}\small
\text{messages}\quad \absmsg \bnfdef \cdots \bnfalt \mkmsgseg{\absmach}{\tolbl\absaddr}{\evnt}{ \seq{\absval} }
\end{mathpar}

A message segment $\mkmsgseg{\absmach}{\tolbl\absaddr}{\evnt}{\seq{\absval}}$ summarizes a sequence of messages sent from a summarized machine $\absmach$ to a single symbolic actor $\tolbl\absaddr$, all with event type $\evnt$ and symbolic payload $\seq{\absval}$. The summarized machine tracks the type of the machine which sends the message, as the segment produces an analysis obligation to verify that the sending of the message being summarized is possible. In addition to these fields, each message segment carries a symbolic segment-length parameter $\seglen$ recording the number of messages it summarizes.

\subsection{Backwards Symbolic Execution with Message Segment Handlers}
\label{sec:segment-handlers}

We now define a mechanism for producing and consuming these message-segments. Intuitively, in a message history a  segment represents the execution of a sequence of messages. Consequently, the analysis must summarize updates to constraints in $\abssysstore$ and $\puredom$ caused by the corresponding sequence of handler invocations.
To this end, we define a message segment handler, which over-approximates the effects of executing a handler in a cycle. We think of this handler as a procedure $\proc$ that summarizes the effect of the transitions along the detected cycle in terms of both the handler $\seq{\var}$ parameters and a segment-length parameter $\seglen$.
Specifically, we extend syntax of handlers $\hand$ with a segment handler that has an additional formal parameter $\seglen$ for the segment length $\hand \bnfdef \cdots \bnfalt \mkonseghand{\evnt}{ \seq{\var} }{\proc}$.
We consider a cycling handler $\evnt$ to be replaced by a segment handler $\evnt$.

\subparagraph{Deriving Segment Handlers.}
We consider transforming the original handler into a segment handler,
as solving a recurrence relation is generally non-trivial.
Consider an assignment of the form $\mkassign{\var}{\var + \val}$ in a loop, iterating $\seglen$ times.
This is a \emph{simple recurrence}~\cite{cra}, and as such always has a closed form $\var(\seglen) = \var_0 + \val\cdot\seglen$.
The addition of control-flow, such as if-else branches, leads to a more complex recurrence, formalized as \emph{linear recurrence inequations} in \cite{cra}.
Thus, a sound approach to generating message-segment handlers is to utilize the CRA domain introduced by \citet{cra}. Particularly, we consider Algorithm 4.14, proven to soundly represent such loops' effects and is parametric on an iteration count, which is our segment length parameter $\seglen$.

\subparagraph{Interpreting Segment Handlers.}
To use segment handlers in our analysis, we add the \\ \TirName{Abs-Step-Segment-Handle} rule to handle message segments:
\begin{mathpar}\small
\infer[Abs-Step-Segment-Handle]{
  \absmsg = \left(\mkmsgseg{\absmach}{\tolbl\absaddr}{\evnt}{ \seq{\absval} }\right)
  \and
  \mkonseghand{\evnt}{ \seq{\var} }{\proc} \in \maplookup{\topsys}{\mach}
  \\ \vspace{-1em}
  \abssysstore' =
      \left(\absstorecons{\abssysstore_0}{ \mapping{\tolbl\absaddr}{ \mkactorstore{\fieldslbl\absstore}{\storetop} } } \right)
  \and
  \abssysstore =
      { \left(\absstorecons{\abssysstore_0}{ \mapping{\tolbl\absaddr}{ \mkactorstore{\fieldslbl\absstore}{ (\mapping{\seglen}{\symseglen}, \seq{\mapping{\var}{\absval}}) } } }\right) }
  \and
  \puredom' = \puredom \land \paffinefun{ \abshisthyp{\absmsg}{\abshist} }
}{
  \jcausalbackstep{
    \mkabssysstate[\puredom']
      { (\absstorecons{\abssysloc}{ \mapping{\tolbl\absaddr}{ \mkmachsttevntloc{\mach}{\stt}{\idleevnt}{\exitloc} } }) }
      { \abssysstore' }
      { (\absnetcons{\absnet}{\absmsg}) }
      { \abshisthyp{\absmsg}{\abshist} }
  }{
    \mkabssysstate
      { (\absstorecons{\abssysloc}{ \mapping{\tolbl\absaddr}{ \mkmachsttevntloc{\mach}{\stt}{\evnt}{\entryloc} } })}
      { \abssysstore }
      { \absnet }
      { \abshist  }
  }
}
\end{mathpar}
The \TirName{Abs-Step-Segment-Handle} differs from \TirName{Abs-Step-Handle} rule in the addition of the segment-length $\seglen$, a formal parameter of the handler.
The segment-length parameter $\seglen$ can be constrained by the $\paffinefun{ \abshisthyp{\absmsg}{\abshist} }$ constraint. It is from the affine restriction that we derive the contradiction in our motivating example in \autoref{sec:message-segments} (based on number of copies of \fmtcode{eVote}).

The \TirName{Abs-Step-Segment-Send} rule is analogous to the normal \TirName{Abs-Step-Send} rule, witnessing a sender of a message of the message segment, which can instantiate the payload under the affine protocol property:
\begin{mathpar}\small
\infer[Abs-Step-Segment-Send]{
  \mktrans{\prelbl\loc}{ \mksend{\var}{\evnt'}{ \seq{\varalt} } }{\postlbl\loc}
  \in \maplookup{\maplookup{\topsys}{\mach}}{\evnt}
  \\
  \mapping{\fromlbl\absaddr}{\absactorstore} \in \abssysstore
  \and
  \absmsg = \left(\mkmsgseg{ \absmach }{ \tolbl\absaddr }{\evnt'}
    { \seq{ \absval  } }\right)
  \and  
  \puredom' = \left(\puredom \land \tolbl\absaddr = \maplookup{\absactorstore}{\var} \land \bigwedge\seq{ \absval = \maplookup{\absactorstore}{\varalt} }\right)
}{
  \jcausalbackstep{
    \mkabssysstate[ \puredom' ]
      { (\absstorecons{\abssysloc}{ \mapping{\fromlbl\absaddr}{ \mkmachsttevntloc{\mach}{\stt}{\evnt}{\prelbl\loc} } }) }
      { \abssysstore }
      { \absnet }{ \abshist }
  }{
    \mkabssysstate
      { (\absstorecons{\abssysloc}{ \mapping{\fromlbl\absaddr}{ \mkmachsttevntloc{\mach}{\stt}{\evnt}{\postlbl\loc} } }) }
      { \abssysstore }
      { (\absnetcons{\absnet}{\absmsg}) }{\abshist}
  }
}
\end{mathpar}
With the original event handler replaced by the generalized
segment handler, applications of actor-constraint causality rules (e.g. the
\TirName{Causal-Step-Store-Constraint} rule) can be used to witness the
effect of the segment handler on the store.

\subsection{Concretization and Soundness}

With our abstract domain extended, we now define concretization for message-segments. As a segment summarizes a sequence of messages, the concrete representation is a set of messages with matching receiver and argument, whose cardinality satisfies the segment-length.

\begin{mathpar}\small
\begin{array}{@{}ll@{}}
\text{$\conc{ \histcolor{\abshisthyp{\mkmsgseg{\mach}{\tolbl\absaddr}{\evnt}{\seq{\absval}}}{\abshist}} }{ \hist }$}
& \text{iff}
\\
\multicolumn{2}{@{}l@{}}{
\quad\qquad\text{%
$\conc{ \mkmsg{\absaddr_1}{\tolbl\absaddr}{\evnt}{\seq{\absval}} }{\msg_1}$ and $\cdots$ and
$\conc{ \mkmsg{\absaddr_k}{\tolbl\absaddr}{\evnt}{\seq{\absval}} }{\msg_k}$ s.t. $\seq{\msg}$ distinct implies
$\conc{ \abshist }{ \histcons{\histcons{\hist}{\msg_1}\cdots}{\msg_k} }$
}}
\\[1ex]
\text{$\conc{ \netcolor{\mkmsgseg{\mach}{\tolbl\absaddr}{\evnt}{\seq{\absval}}} }{ \netcons{\netcons{\msg_1}{\cdots}}{\msg_k} }$}
& \text{iff}
\\
\multicolumn{2}{@{}l@{}}{
\quad\qquad\text{%
$\conc{ \mkmsg{\absaddr_1}{\tolbl\absaddr}{\evnt}{\seq{\absval}} }{\msg_1}$ and $\cdots$ and
$\conc{ \mkmsg{\absaddr_k}{\tolbl\absaddr}{\evnt}{\seq{\absval}} }{\msg_k}$ s.t. $\seq{\msg}$ distinct
}
}
\end{array}
\end{mathpar}

In the message history, the set of messages must be consecutive. This introduces an important restriction for soundness, due to asynchrony.

\subparagraph{Message Orderings and Trace Equivalence.} A message-segment summarizes a sequence of $\seglen$ identical messages being received by an actor. However, in the concrete system  such a sequence may be broken up by additional messages received by the same actor, and this history is not considered by the analysis.
This adds an additional soundness criterion for message-segments: we can only replace a message cycle on event $\evnt$ with a message segment when all other messages an actor may receive are reorderable with respect to the handler for $\evnt$.

More precisely, to summarize a handler $h$ which may be interleaved with another handler $h_1$, we must show the \emph{independence} of the two handlers. Concretely, this means executing $h$ then $h_1$ must produce a state identical to the one produced by executing $h_1$ then $h$ (differing only by message order in the history). We formalize this notion of independence using Mazurkiewicz trace equivalence~\cite{mazurkiewicz}, under which executions differing only by reorderings of independent actions are considered equivalent. Our message-segments cause the analysis to pick a canonical ordering, which covers all equivalent traces.
Our additional soundness criterion requires this equivalence class to cover the possible behaviors of the concrete system, so that the message cycle can be completely replaced by a message-segment, without losing behaviors. Formally, we define $\mathit{symmetric}(h)$ to mean that the handler $h$ is independent of all other handlers, and thus reordering the history by only moving $h$'s position produces a Mazurkiewicz-equivalent trace.
Similarly for traces, we define $\mathit{symmetric}(\predtrace,\predtrace')$ to mean two traces are Mazurkiewicz-equivalent. In \autoref{sec:symmetry-appendix}, we provide additional details, including a precise formal rule for our abstraction.

\subparagraph{Message Segment Lengths.}

The segment-length of a message-segment is bounded, both from above and below, to prevent materializing more messages than physically possible in the concrete system. $\seglen$ must be non-negative, and from the affine protocol restriction, the upper bound of $\seglen$ is the size of the system, for instance \purecolor{$\exsyssize$} in our election example.
These trivial bounds alone are unlikely to be sufficient to provide a contradiction, but additional constraints on $\seglen$ may exist in the goal-state. In the election example, we know $\seglen > \exsyssize/2$ in the correct case, due to the constraint on \fmtcode{forMe}.
Having defined a summarization technique for histories, we again restate a soundness theorem:

\begin{theorem}[Message-Segment Soundness]\label{thm:segment}
  
If
$\jstep[\topsys]{\prelbl\predtrace}{\postlbl\predtrace}$
and
$\jcausalbackstep[\topsys]{\prelbl\abssysstate}{\postlbl\abssysstate}$
such that
$\concstate{\postlbl\abssysstate}{\postlbl\predtrace}$,
then $\mathit{symmetric}(\prelbl\predtrace,\predtrace')$ and $\concstate{\prelbl\abssysstate}{\predtrace'}$ for some $\predtrace'$.

\end{theorem}

This is similar to \autoref{thm:causal-soundness}, but differs in that the preceding trace need not match, so long as the new trace $\predtrace'$ is \emph{Mazurkiewicz-equivalent}.
See \autoref{sec:mechanization} for proof. This follows from our transformed handler being sound according to \cite{cra}, our restriction  based on symmetry to ensure trace coverage, and the affine restriction producing constraints on symbolic segment lengths.

\subsection{Backwards Abstract Interpretation with Materialization}

Having established the soundness of the backwards abstract semantics, we now describe its use in verifying a protocol. Starting from an error-goal state, we apply the backwards abstract semantics repeatedly, until fixed-point is reached. If the concretization of the fixed-point invariant $\abssysstate$ does not include any initial states, then the analysis has refuted the error-goal state (and consequently, proven safe its negation, the safety condition).
The remaining challenge is to ensure the analysis terminates, and thus the fixed-point can be computed. Specifically, the analysis may not terminate with unbounded actor materialization, causing unbounded handler exploration.

In the remainder of this section, we describe the sources of cycles in the analysis, and define the notion of an over-approximate \emph{causal-dependence graph} for a protocol to capture these sources. We then discuss strategies for widening to ensure termination.

\subsubsection{Causal Cycles.}\label{sec:causal-cycles}

The causally-reduced abstract semantics $\jcausalbackstep[\topsys]{\abssysstate'}{\abssysstate}$ of \autoref{sec:reflogic} explores execution of a handler in two cases: attempting to witness a field assignment relevant to a constraint and attempting to witness a sender of a message in the network buffer.
The former is captured by the \TirName{Causal-Step-Store-Constraint} rule that attempts to witness a points-to constraint in an actor's store
\[
  \text{$\mapping{\var}{\symvar}$ with an assignment transition
  $\mktrans{\prelbl\loc}{\mkassign{\var}{\expr}}{\postlbl\loc} \in \maplookup{\maplookup{\topsys}{\mach}}{\evnt}$}
\;,
\]
while the latter is captured by the \TirName{Causal-Step-Network} rule that attempts to witness a message in the network
\[
  \text{$\mkmsg{\fromlbl\absaddr}{\tolbl\absaddr}{\evnt'}{\seq{\symvar}}$ with a send transition
  $\mktrans{\prelbl\loc}{ \mksend{\var}{\evnt'}{ \seq{\varalt} } }{\postlbl\loc} \in \maplookup{\maplookup{\topsys}{\mach}}{\evnt}$}
\;.
\]

Attempting to witness store constraints $\mapping{\var}{\symvar}$ or network constraints $\mkmsg{\fromlbl\absaddr}{\tolbl\absaddr}{\evnt'}{\seq{\symvar}}$ requires the analysis to explore the execution of a handler from a corresponding control location. This can update or add constraints, including another message for the same event, or the same variable.
While the affine restriction prevents witnessing the same message $\sendkw$ for a given symbolic actor $\fromlbl\absaddr$, it does not prevent unbounded materialization of actors in the parameterized actor model.

\subsubsection{Causal-Dependence Graphs.}

To capture these dependencies statically, we define the \emph{causal-dependence graph} $\causalgraph{\topsys}$ of a distributed protocol $\topsys$. This graph over-approximates the relationships between sends and assignments which guide the causally-reduced semantics of \autoref{sec:reflogic}. 
Its vertices are fields $\mkmachfld{\mach}{\var}$ or event types $\mkmachevnt{\mach}{\evnt}$, with directed edges indicating that one field or event may require witnessing another. Intuitively, this means that witnessing the target send or assignment produces an obligation to witness the source send or assignment. In particular, we distinguish 5 types of dependence:
\begin{description}\small
  \item[data dependence]
  $\causaltofrom{\mkmachfld{\mach}{\varalt}}{\mkmachfld{\mach}{\var}}$ if there is an assignment to variable $\var$ that is data dependent on $\varalt$ (i.e.,
  $\causaltofrom{\mkmachfld{\mach}{\varalt}}{\mkmachfld{\mach}{\var}} \in \causalgraph{\topsys}$ if
  $(\mktrans{-}{\mkassign{\var}{\expr}}{-}) \in \maplookup{\maplookup{\topsys}{\mach}}{\evnt}$ where $\varalt$ is in the free variables of expression $\expr$, for some event type $\evnt$).
  \item[control dependence]
  $\causaltofrom{\mkmachfld{\mach}{\varalt}}{\mkmachfld{\mach}{\var}}$ if there is an assignment to variable $\var$ that is control dependent on $\varalt$ (i.e.,
  $\causaltofrom{\mkmachfld{\mach}{\varalt}}{\mkmachfld{\mach}{\var}} \in \causalgraph{\topsys}$ if
  $\mktrans{\prelbl\loc}{\mkassign{\var}{-}}{-} \in \maplookup{\maplookup{\topsys}{\mach}}{\evnt}$ where a guard using $\varalt$ dominates location $\prelbl\loc$, for some event type $\evnt$).
  \item[send-control dependence] 
  $\causaltofrom{\mkmachfld{\mach'}{\varalt}}{\mkmachevnt{\mach}{\evnt}}$ if there is a send of event $\evnt$ with a control dependence on variable $\varalt$ (i.e.,
   $\causaltofrom{\mkmachfld{\mach'}{\varalt}}{\mkmachevnt{\mach}{\evnt}} \in \causalgraph{\topsys}$ if
  $\mktrans{\prelbl\loc}{ \mksend{\var}{\evnt}{-} }{-} \in \maplookup{\maplookup{\topsys}{\mach'}}{\evnt'}$ 
  where $\var$ has machine type $\mach$ and a guard using $\varalt$ dominates location $\prelbl\loc$, for some event type $\evnt'$ and machine type $\mach'$).
  \item[assign-handle dependence] 
  $\causaltofrom{\mkmachevnt{\mach}{\evnt}}{\mkmachfld{\mach}{\var}}$ if there is an assignment to variable $\var$ in a handler for event $\evnt$ (i.e., $\causaltofrom{\mkmachevnt{\mach}{\evnt}}{\mkmachfld{\mach}{\var}} \in \causalgraph{\topsys}$ if $\mktrans{-}{\mkassign{\var}{-}}{-} \in \maplookup{\maplookup{\topsys}{\mach}}{\evnt}$).
  \item[send-handle dependence] 
  $\causaltofrom{\mkmachevnt{\mach'}{\evnt'}}{\mkmachevnt{\mach}{\evnt}}$ if there is a send of an event $\evnt$ in a handler for event $\evnt'$ (i.e., $\causaltofrom{\mkmachevnt{\mach'}{\evnt'}}{\mkmachevnt{\mach}{\evnt}} \in \causalgraph{\topsys}$ if $\mktrans{-}{ \mksend{\var}{\evnt}{-} }{-} \in \maplookup{\maplookup{\topsys}{\mach'}}{\evnt'}$ where $\var$ has machine type $\mach$).
\end{description}

This causal-dependence graph thus captures which fields and events are dependent upon one another. We say it is actor-insensitive, as it does not consider distinct actors, only their types.

\begin{figure}
\begin{subfigure}[t]{.4\textwidth}
\centering\begin{tikzpicture}
  \tikzset{every path/.style={thick}}
  \tikzset{every loop/.style={min distance=10mm,in=90,out=40,looseness=10}}
  
  \node[circle, fill=Midnight, inner sep=2pt, label=right:\textcolor{Midnight}{\sfmtcode{Node.eInit}}] (e0) at (1,2) {};
  \node[circle, fill=Midnight, inner sep=2pt, label=left:\textcolor{Midnight}{\sfmtcode{Node.eVote}}] (e1) at (0,3) {};
  
  \node[draw=Cayenne, rectangle, fill=Cayenne, minimum size=4pt, inner sep=0pt, label=right: \textcolor{Cayenne}{\sfmtcode{Node.forMe}}] (f0) at (1,3) {};
  \node[draw=Cayenne, rectangle, fill=Cayenne, minimum size=4pt, inner sep=0pt, label=left:\textcolor{Cayenne}{\sfmtcode{Node.id}}] (f1) at (1,1) {};
  \node[draw=Cayenne, rectangle, fill=Cayenne, minimum size=4pt, inner sep=0pt, label=left:\textcolor{Cayenne}{\sfmtcode{Node.leader}}] (f2) at (0,2) {};

  \draw[<-] (e0) -- (e1);
  \draw[<-] (e1) -- (f2);
  \draw[<-] (f0) -- (f2);
  \draw[<-] (e1) -- (f0);
  \draw[<-] (f1) -- (f2);
  
  \draw[<-, dashed, dash pattern=on 2pt off 1pt] (f0) edge[loop above] ();
\end{tikzpicture}

\caption{}
\label{fig:osv}
\end{subfigure}
\quad
\begin{subfigure}[t]{.5\textwidth}
\centering\begin{tikzpicture}
  \tikzset{every path/.style={thick}}
  \tikzset{every loop/.style={min distance=10mm,in=0,out=60,looseness=10}}
  \node[circle, fill=Midnight, inner sep=2pt, label=left:\textcolor{Midnight}{\sfmtcode{Participant.eInit}}] (e0) at (0,4) {};
  \node[circle, fill=Midnight, inner sep=2pt, label=right:\textcolor{Midnight}{\sfmtcode{Coordinator.eInit}}] (e1) at (1.5,4) {};
  \node[circle, fill=Midnight, inner sep=2pt, label=left:\textcolor{Midnight}{\sfmtcode{Participant.eProp}}] (e2) at (0,3) {};
  \node[circle, fill=Midnight,inner sep=2pt, label=right:\textcolor{Midnight}{\sfmtcode{Coordinator.eResp}}] (e3) at (1.5,3) {};
  \node[circle, fill=Midnight, inner sep=2pt, label=right:\textcolor{Midnight}{\sfmtcode{Participant.eFinal}}] (e4) at (1.5,1) {};

  \node[draw=Cayenne, rectangle, fill=Cayenne, minimum size=4pt, inner sep=0pt, label=left:\textcolor{Cayenne}{\sfmtcode{Participant.potential}}] (f0) at (0,2) {};
  \node[draw=Cayenne, rectangle, fill=Cayenne, minimum size=4pt, inner sep=0pt, label=left:\textcolor{Cayenne}{\sfmtcode{Participant.internal}}] (f1) at (0,1) {};
  \node[draw=Cayenne, rectangle, fill=Cayenne, minimum size=4pt, inner sep=0pt, label={[yshift=-5pt,xshift=4pt]right:\textcolor{Cayenne}{\sfmtcode{Coordinator.received}}}] (f2) at (1.5,2) {};
  
  \draw[<-, bend right=35] (e3) to (e4);
  \draw[<-] (f0) -- (f1);
  \draw[<-] (e1) -- (e2);
  \draw[<-] (f2) -- (e4);
  \draw[<-] (e3) -- (f2);
  \draw[<-] (e4) -- (f1);
  \draw[<-] (e2) -- (e3);
  \draw[<-] (f1) -- (e3);
  \draw[<-] (e2) -- (f0);
  \draw[<-, dashed, dash pattern=on 2pt off 1pt, loop right] (f2) to ();
\end{tikzpicture}

\caption{}
\label{fig:tpcv}

\end{subfigure}

\caption{Causal dependence graphs abstract the causally-reduced abstract semantics $\jcausalbackstep[\topsys]{\abssysstate'}{\abssysstate}$ of \autoref{sec:reflogic} in an actor-insensitive manner. Field vertices $\mkmachfld{\mach}{\var}$ are drawn as squares \storecolor{$\blacksquare$}, and event vertices $\mkmachevnt{\mach}{\evnt}$ are drawn as circles \histcolor{$\bullet$}. For presentation, we elide the assign-handle dependence edge $\causaltofrom{\mkmachfld{\mach}{\initevnt}}{\mkmachfld{\mach}{\var}}$ for initializing an actor's field $\var$.
We draw dashed edges $\causaltofromdash{\mkmachevnt{\mach}{\var}}{\mkmachevnt{\mach}{\var}}$ for data-dependence self-loops that we can summarize using recurrence analysis and message-segment summarization.
(a) The one-shot leader election protocol from \autoref{fig:eVote} in \autoref{sec:motivation}.
(b) A two-phase commit protocol.
On a receiving a proposed value in a \histcolor{$\exfmtcode{eProp}$}, a $\ifmtcode{Participant}$ sends a positive response \histcolor{$\exfmtcode{eResp}$} to the \ifmtcode{Coordinator} if its \storecolor{$\exfmtcode{internal}$} field is not set, and when the \ifmtcode{Coordinator} receives this response, it increments a counter \storecolor{$\exfmtcode{received}$}. When the counter \storecolor{$\exfmtcode{received}$} surpasses a threshold, it sends \histcolor{$\exfmtcode{eFinal}$} messages that causes each \ifmtcode{Participant} to commit the proposed value stored in \storecolor{$\exfmtcode{potential}$} to \storecolor{$\exfmtcode{internal}$}.
}
\label{fig:causal-graphs}
\end{figure}

In \autoref{fig:causal-graphs}, we show two examples of causal-dependence graphs for (\subref{fig:osv}) the one-shot leader election protocol from \autoref{sec:motivation} and (\subref{fig:tpcv}) for a two-phase commit protocol. Any cycle in the graph corresponds to a potential source of non-termination in the abstract semantics where we need some strategy to ensure termination.

\subsubsection{Cycle Summarization.}

We consider the one-shot leader election's causal-dependence graph in \autoref{fig:osv}, which has only the self-loop on $\mkmachfld{\ifmtcode{Node}}{\ifmtcode{forMe}}$ as a cycle. We show this self-loop as a dashed edge $\causaltofromdash{\mkmachfld{\ifmtcode{Node}}{\ifmtcode{forMe}}}{\mkmachfld{\ifmtcode{Node}}{\ifmtcode{forMe}}}$ to indicate that it can be handled by the message-segment summarization described in \autoref{sec:segment-handlers}.
In particular, by rewriting the relevant handler to explicitly use a closed-form solution, the resulting program has no cycle involving the field, and thus we can reduce this graph to an acyclic one.

While not all cycles can be summarized this way, the presence of a cycle in the causal-dependence graph does not necessarily cause non-termination. Backwards exploration still terminates when repeated traversal of a cycle produces no new abstract states, yielding an inductive invariant. To detect convergence, we revisit the entailment relation described in \autoref{sec:actorlogic}, to check whether states produced by cycle exploration are subsumed by previously encountered states.

We give an operational definition of the entailment relation $\jincluded{\abssysstate}{\abssysstate'}$ by matching the separation, linear, and ordered logic constraints of $\abssysstate'$ against those of $\abssysstate$ component-wise under a unification map for existentially quantified symbolic values $\absval$. Operationally, this amounts to a separation-logic subtraction algorithm~\cite{DBLP:conf/aplas/BerdineCO05}. Since location maps $\abssysloc$ are represented as separation-logic constraints rather than simple locations, they are handled by the same procedure. Consequently, entailment between abstract states
\smash{$\jincluded{ \abssysstate\colon \mkabssysstate{\abssysloc}{\abssysstore}{\absnet}{\abshist} }{ \abssysstate'\colon \mkabssysstate{\abssysloc'}{\abssysstore'}{\absnet'}{\abshist'} }$} only if the location maps satisfy entailment \smash{$\jincluded{\abssysloc}{\abssysloc'}$}.

In \autoref{fig:tpcv}, we show the causal-dependence graph for a two-phase commit protocol. The $\causaltofromdash{ \mkmachfld{\ifmtcode{Coordinator}}{\ifmtcode{received}} }{ \mkmachfld{\ifmtcode{Coordinator}}{\ifmtcode{received}} }$ self-loop is resolved by message-segment summarization, as in the one-shot leader election example.
But there is another cycle from $\mkmachfld{\ifmtcode{Participant}}{\ifmtcode{internal}}$ to $\mkmachevnt{\ifmtcode{Participant}}{\ifmtcode{eFinal}}$ to $\mkmachfld{\ifmtcode{Coordinator}}{\ifmtcode{received}}$ to $\mkmachevnt{\ifmtcode{Coordinator}}{\ifmtcode{eResp}}$ back to $\mkmachfld{\ifmtcode{Participant}}{\ifmtcode{internal}}$. In this example, this cycle is resolved by entailment in the fixed-point iteration.

\subsubsection{Termination With Causal Dependencies.}

Suppose cycles in the causal-dependence graph of the protocol are eliminated by message-segments, entailment, or additional mechanisms. We are left with acyclic causal-dependence graphs (more generally, the acyclic condensation of the strongly-connected components).
Even with cycles removed, termination of the fixed-point iteration of the abstract causal-reduced semantics $\jcausalbackstep[\topsys]{\abssysstate'}{\abssysstate}$ is still non-trivial. The analysis is allowed to materialize an unbounded number of actors, which can lead to non-termination. In the remainder of this section, we formalize a sufficient termination condition of the analysis of a protocol $\topsys$ leveraging an acyclic causal-dependence graph abstraction $\causalgraph{\topsys}$ and provide a ranking-function proof of this condition.

We write $\jinitbot{\abssysstate}$ for the judgment that abstract state $\abssysstate$ includes a concrete initial distributed-system state or is contradictory. Furthermore, let us write $\jcausalbackstepn[n]{\abssysstate'}{\abssysstate}$ for the $n$-step iteration of the abstract causal-reduced semantics on non-disjunctive abstract states.
\vspace{-.4cm}
\begin{theorem}[Termination of Fixed-Point Iteration with Unbounded Actor Materialization under Acyclic Causal Dependencies]\label{thm:termination}
Assume the causal-dependence graph $\causalgraph{\topsys}$ of a protocol $\topsys$ is acyclic. Then, for any non-disjunctive abstract state $\abssysstate$, there exists a natural-number bound $n$ such that $\jcausalbackstepn[n]{\abssysstate'}{\abssysstate}$ leads to some non-disjunctive abstract state $\abssysstate'$ where
$\jinitbot{\abssysstate'}$.
\end{theorem}
\vspace{-.4cm}
See Appendix \ref{sec:termination} for proof details. Here, we give some intuition.
Let $D$ be the maximum depth of the causal-dependence graph $\causalgraph{\topsys}$, and $H$ be the maximum number of steps needed to reach an $\entryloc$ location from an $\exitloc$ location in the underlying semantics of the protocol $\topsys$. We define $W$ to be the max-indegree of any location within a handler, to conservatively account for branching control-flow.
We can define the amount of work remaining to be done as $i+j+k$, where we let $i, k,j$ be the number of non-initial actors, unconstrained fields, and non-initial messages, respectively. We shall show the amount of work to be done grows by a finite amount and ultimately grows by 0. We end up with a sequence  $s(n) = 2^n(D!/(D-n)!)(i+j+k)HW^n$. By \reftxt{Lemma}{thm:worklist} in Appendix \ref{sec:mechanization}, this sequence converges to 0. $\qed$

\section{Empirical Evaluation}\label{sec:evaluation}

In this section, we describe our verification tool \tool{} and evaluate the efficacy of our techniques for automatically verifying the safety of distributed protocols with affine communication.
In particular, we consider the following research questions (RQs):
\begin{description}
  \item[RQ1] Are our techniques effective at verifying safety of parametric actor systems with non-trivial affine communication?
  \item[RQ2] Is actor-causal reduction necessary for reducing the explored state space?
  \item[RQ3] Is message-segment summarization necessary to verify distributed protocols?
  \item[RQ4] Are entailment checks necessary for reducing the explored state space?
\end{description}

\subparagraph{Implementation.}

We implement our goal-directed verification engine in a tool called \tool{}
written in OCaml. The input to \tool{} is the P source code of a parametric
actor program and a negated safety property. \tool{} verifies safety properties of these protocols by performing backwards abstract interpretation from a state over-approximating the negated property, and refuting the existence of a path that leads to a violation of safety. The output of \tool{} is either an error trace demonstrating a possible path to a violation or a statement that no such path exists.

\subsubsection*{\textbf{RQ1}: Are our techniques effective at verifying safety of parametric actor systems with non-trivial affine communication?}

\begin{table}[b]
\caption{Distributed protocols defined as parametric actor systems with affine communication verified by \tool. Note that \tool{} verifies a consensus or mutual execution safety property of these protocols without user-provided invariants.}
\centering\smaller
\begin{tabular}{ | c | c | c | c | c | c | c |  c | c |}
\hline
\textbf{Benchmark} & \textbf{SLOC} & \textbf{Handlers} & \makecell{\textbf{Unique}\\ \textbf{Msgs}} & \makecell{\textbf{States} \\ \textbf{Explored}} &
\makecell{\textbf{Actors} \\ \textbf{Mat.}} &
\makecell{\textbf{History} \\ \textbf{Length}} & \makecell{\textbf{Summary} \\ \textbf{Messages}} & \makecell{\textbf{Time}\\ \textbf{(h:m:s)}}  \\
\hline
One-Shot Election & 56 & 3 & 3 & 669 & 6 & 9 & 2 &  0:02:40 \\
Bakery & 99 & 9 & 4 & 16 & 4 & 4 & 1 & 0:00:08 \\
Dekker's Algorithm & 68 & 5 & 5 & 1030 & 7 & 11 & 0 & 0:00:43 \\
Drinking Philosophers & 58 & 3 & 2 & 110 & 4 & 4 & 0 & 0:00:22 \\
Two-Phase Commit & 55 & 5 & 5 & 5888  & 8  & 12 & 1 & 1:08:46 \\
Barrier Lock & 49 & 5 & 4  & 35 & 4 & 5 & 1 & 0:00:04 \\
Atomic Broadcast & 80 & 7 & 6 & 9275 & 8 & 13 & 2 & 1:12:46 \\
\hline
\end{tabular}
\label{tab:bench}  
\end{table}

In \reftxt{Table}{tab:bench}, we list the non-trivial distributed protocols we verify using \tool{}.
One-Shot Election refers to the leader election protocol described in \autoref{sec:motivation}, Bakery refers to the distributed implementation of Lamport's Bakery algorithm for mutual exclusion~\cite{bakery-cacm74},
Two-Phase Commit is a standard transaction protocol where a set of processes agree to commit a value (or decide not to), and
Barrier Lock is a distributed lock server. Drinking Philosophers is a variant of the classical dining philosophers algorithm with asynchronous communication, and Dekker's algorithm is another classical mutual exclusion algorithm.

For each benchmark, \tool{} starts only with the error-goal state violating either consensus for One-Shot,Two-Phase Commit, and Atomic Broadcast, or mutual execution for the remainder. No additional user-provided specifications or invariants are needed to verify these properties, and the tool is able to refute the existence of a path to a violation of safety in each case. In particular, this is the first time we are aware of that Two-Phase Commit, as well as One-Shot (sometimes referred to as Toy Consensus, as in Ivy~\cite{ivy}), have been verified without user-provided invariants.

The SLOC column shows the number of lines of P source code, the Handlers
column lists the number of unique message handlers (e.g., \lstinline!on eVote do $\ldots$!),
and the Unique Msgs column lists the maximum number of unique events (e.g., \lstinline!eVote!) that occur in an execution trace of the protocol. Handlers and Unique Msgs
determine the number of actors to be materialized and branches to be explored in each case.
Bakery, for example, requires exploring executions containing up to 4 (unique) messages before the analysis can refute safety violation,
whereas Two-Phase Commit looks at all 5 unique message types in its code. 
The experiments were run on a MacBook Pro with a 2.4GHz i9 processor, and 64 GB of RAM, using 12 cores.

\begin{wraptable}{r}{.5\linewidth}\smaller
\caption{Number of queries checked over the course of analysis for each benchmark. The number of entailment queries grows quickly with the state space, whereas feasibilty grows more linearly.}
\label{tab:queries}
\begin{tabular}{ | c | c | c |}
\hline
\textbf{Benchmark} & \makecell{\textbf{Entailment} \\ \textbf{Queries}} & \makecell{\textbf{Feasibility} \\ \textbf{Queries}}\\ 
\hline
Bakery & 7 & 25 \\
\hline
Two-Phase Commit & 339864 & 8910 \\
\hline 
One-Shot & 4209 & 815 \\
\hline
Barrier Lock & 13 & 50 \\
\hline
Dekker's Algorithm & 7806 & 1214 \\
\hline
Philosophers & 522 & 227 \\
\hline
Atomic Broadcast & 629940 & 12429 \\
\hline
\end{tabular}
\end{wraptable}%
Table~\ref{tab:queries} collects the total number of entailment and feasibility queries checked over the experiment's runtime, to see where verification time is spent.
As the number of states grows, so does the number that may entail new states, and thus the number of entailment checks.
Two-Phase Commit and Broadcast have far more entailment queries than the rest, contributing to the extended runtime, while the faster experiments make fewer queries. Deeper analysis of the history leads to this state space growth, and thus the increased time spent checking entailment.

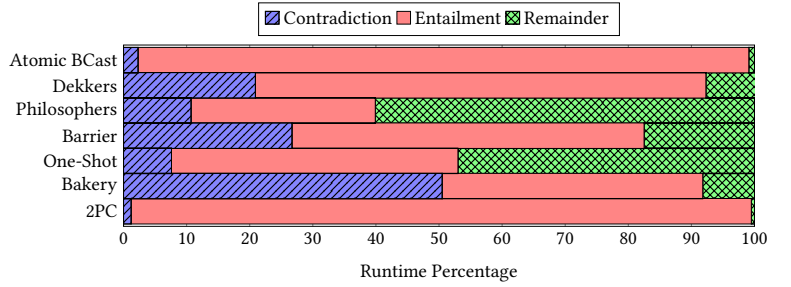
\begin{wrapfigure}{r}{0.7\linewidth}
\centering\smaller
\scalebox{0.8}{\begin{tikzpicture}
\begin{axis}[
    xbar stacked,
    bar width=12pt,
    symbolic y coords={2PC,Bakery,One-Shot,Barrier, Philosophers, Dekkers, Atomic BCast},
    ytick=data,
    xmin=0, xmax=100,
    xlabel={Runtime Percentage},
    height=3cm,
    width=0.75\textwidth,
    scale only axis,
    tick label style={font=\small},
    label style={font=\small},
    legend style={
        at={(0.5,1.05)},
        anchor=south,
        legend columns=3,
        font=\small
    }
]

\addplot[
    fill=blue!80!white!60,
    postaction={pattern=north east lines}
] coordinates {(1.2,2PC) (50.5,Bakery) (7.6,One-Shot) (26.7,Barrier) (10.7,Philosophers) (20.9,Dekkers) (2.3,Atomic BCast)};

\addplot[
    fill=red!80!white!60
] coordinates {(98.3,2PC) (41.3,Bakery) (45.4,One-Shot) (55.8,Barrier) (29.2,Philosophers) (71.4,Dekkers) (96.8,Atomic BCast)};

\addplot[
    fill=green!80!white!60,
    postaction={pattern=crosshatch}
] coordinates {(0.5,2PC) (8.2,Bakery) (47,One-Shot) (17.5,Barrier) (60.1,Philosophers) (8.7,Dekkers) (0.9,Atomic BCast)};

\legend{Contradiction, Entailment, Remainder}

\end{axis}
\end{tikzpicture}}
\caption{Runtime analysis of time spent resolving SMT queries on contradiction and entailment, as opposed to the remainder of the execution.}
\label{fig:smtdata}
\end{wrapfigure}
To understand the impact of SMT queries on runtime, in \autoref{fig:smtdata}, we collect the percentage of time spent resolving SMT queries, divided into entailment and contradiction, versus the remainder of the analysis, which includes both exploration and parallelism coordination.
At runtime, we resolve these queries via Z3. These queries are heavily parallelized, yet remain a significant bottleneck to the runtime. While queries are convenient to generate, they may be inefficient to solve compared with more specialized domain solvers.
We observe that the SMT queries, especially entailment, make up a substantial amount of the runtime, which grows to dominate as the state space considered grows.

\vspace{-.25cm}
\subsubsection*{\textbf{RQ2}: Is actor-causal reduction necessary for reducing the explored state space?}

We ran an ablation study disabling the causal reductions of \autoref{sec:reflogic}. We capped the analysis at exploring the same number of states as in our primary run. We collect the same metrics on explored states, to compare how deep into execution the analysis explores without reductions. Without
causality reductions, we expect the analysis to branch extensively,
manifesting many states with shallow message histories.
In \autoref{tab:ablation}, we collect the results of this study.
In each benchmark, the queue remaining is non-trivial with respect to the states explored in the original benchmark.
The remaining queue is a conservative metric, as each queued state could cause many more to be added, so the comparison is not just a matter of considering the extra states.
Additionally, for all benchmarks but the Barrier Lock, the history length is much smaller than the final maximum length needed to refute the safety violation. This means each run remained closer to the error state, instead of getting to initial or contradiction, due to more aggressive branching. These results confirm causality reduction's
crucial role in reducing the state space and enabling the analysis to verify the safety property.

\begin{wraptable}{r}{0.5\linewidth}
\caption{Actor-causal reduction ablation: we collect the number of remaining states to explore, and history length and actor count.}
\centering
\smaller
\begin{tabular}{ | c | c | c | c |}
\hline
\textbf{Benchmark} & \makecell{\textbf{Remaining} \\ \textbf{Queue}} & \makecell{\textbf{History} \\ \textbf{Length}} &\textbf{Actors}\\ 
\hline
Bakery & 15 & 2 & 2 \\
\hline
Two-Phase Commit & 1258 & 6  & 2 \\
\hline
Atomic Broadcast & 1451 & 9 & 2 \\
\hline 
One-Shot & 125 & 2 & 2 \\
\hline
Barrier Lock & 3 & 2 & 3 \\
\hline
Drinking Philosophers & 22 & 1 & 2 \\
\hline
Dekker's Algorithm & 78 & 5 & 3 \\
\hline
\end{tabular}
\label{tab:ablation}  
\end{wraptable}

\vspace{-.25cm}
\subsubsection*{\textbf{RQ3}: Is message-segment summarization necessary to verify distributed protocols?}

\tool{} produces message-segments when a data-dependent cycle occurs. If a segment is produced, the cycle is not unrolled. Without another form of widening the cycle is not removed, preventing termination.
In our benchmarks, most protocols synthesize message-segments while executing backwards. Dekker's algorithm, as well as the philosophers, do not need these as they do not keep counters or similar to summarize, while algorithms such as Bakery and Two-Phase Commit require summarizing counters relevant to the safety property.

\vspace{-.25cm}
\subsubsection*{\textbf{RQ4}: Are entailment checks necessary for reducing the explored state space?}

To gain insight into entailment's effect on the analysis, we ran an ablation study with entailment disabled. Instead, the tool relies on pruning infeasible and duplicate states to reach a fixed point. We cap this run at $25000$ states, recognizing that it may still terminate, though looking at more states.
Without this check, most benchmarks fail to terminate within $25000$ states. The Barrier Lock is an exception, reaching $87$ states, more than double that of the non-ablative run. This single-use barrier is simple compared to other mutual exclusion components, and doesn't exhibit complex cycling, making entailment is unnecessary. The remaining algorithms lead to cycling during analysis which cannot be pruned. For instance, Two-Phase Commit cycles on an unbounded chain of proposed values, which entailment cuts.
Despite this, the analysis runs much quicker, as entailment is expensive check to the rest of the analysis. Most algorithms fail to verify within the limit as they explore longer paths, leading to more actor materialization and more branching. This produces a much larger state space, which could be mitigated by more aggressive widening operations.

\subparagraph{Limitations.}

While we verify consensus protocols without user-specified invariants or hints, including Two-Phase Commit for the first time, scalability remains a concern for complex protocols. Single-Round Paxos and Raft fall under the affine restriction, and so are theoretically in scope, but resist automated verification in practice: after several hours on Single-Round Paxos, our tool showed no sign of completing. This suggests asynchrony and entailment cost remain non-trivial to overcome; reducing the state space while staying fully automated is a natural future direction.

Additionally, the safety properties we consider are not arbitrary formulas. Each benchmark's mutual exclusion or consensus property is expressible in prenex-universal form; negating it produces the form $\exists\bar{x}.\phi$ where $\phi$ is unquantifiedm, and skolemization turns the existential values into actors and fields for the goal-state. Another future direction is extending this formula space.

\section{Related Work}
\label{sec:related}
\subparagraph{\textbf{Formal Verification of Distributed Systems}} A range of automation levels have been considered for distributed systems verification, from foundational proof assistants to push-button verification.

Frameworks embedded in foundational proof assistants offer little automation, but allow extraction of verified executable code. Examples include Verdi~\cite{verdi}, Disel~\cite{disel}, Actris~\cite{actris}, Aneris~\cite{aneris}, and Grove~\cite{grove}. Notably, Verdi has been used to verify the linearizability guarantee of Raft protocol~\cite{raft-atc14} and extract its executable implementation. This required \~50k lines of proof in Coq~\cite{verdi-cpp16}, making the effort prohibitive at scale. Disel,
Actris, Aneris, and Grove use separation logic to streamline proofs, without automation. These frameworks have produced verified implementations which perform comparably to industry-strength
implementations~\cite{verdi,grove}. 

Partial automation has been explored to reduce proof burden.
IronFleet~\cite{ironfleet} is an end-to-end verification methodology implemented
in Dafny~\cite{dafny}, while Igloo~\cite{igloo} is a framework for
IronFleet-style systems in Isablelle/HOL, allowing SMT-aided
automation~\cite{sledgehammer-icar10}. Automation here is fragile due to first
order logic's undecidability. A complementary approach --- implemented in
Ivy~\cite{ivy} --- restricts the language to ensure decidable verification
conditions, but requires programmers to write explicit inductive invariants.
Building on Ivy, I4~\cite{i4}, DistAI~\cite{distai-osdi21}, and
DuoAI~\cite{duoai-osdi22} infer inductive invariants via heuristics and
data-driven techniques. We take a similar stance of full automation, but use
abstract interpretation to compute an explicit state invariant backwards, rather
than a templated search.

\subparagraph{\noindent\textbf{Symbolic Model Checking and Partial Order
Reduction}} Symbolic model checking is the closest line of work to our approach,
and has been used to reason about actor systems ~\cite{psym,quicksilver,symtlc,
apalache-oopsla19,q9-oopsla18, samc-osdi14}. These approaches are fully
automated, but require the system size to be artificially bounded, precluding
parametric reasoning. Pretend Synchrony~\cite{pretend} uses partial-order
reduction to reduce asynchronous message-passing communication to synchronous
read-writes to shared state. The resultant verification problem is more
tractable, but still requires explicit inductive invariants, which we avoid. 

\subparagraph{\textbf{Backwards Parametric Model Checking}}
Backwards reachability is well studied~\cite{ghilardi,monocegar,cubicle} for
synchronous~\cite{delzanno2003constraint} and shared-memory mutual exclusion
protocols, where it reduces the search space for safety
verification~\cite{paroshwqo} - a benefit our approach shares.
Cubicle~\cite{cubicle} specifies protocols as synchronous global transition
systems, simplifying verification. Non-determinism is mitigated by an auxiliary
invariant synthesis procedure.  In contrast, we consider actors with purely
local state and asynchronous communication, which enables direct modeling of
message-passing distributed protocols, at the cost of a larger state space. To
control the resulting non-determinism, we exploit causality information from the
program model (causal consistency). Complementarily, invariant synthesis
techniques such as those used in Cubicle could be incorporated to further
improve scalability.

\subparagraph{\textbf{Goal-Directed Analysis with Message Histories}}
Goal-directed/ Backwards-from-error abstract interpretation has been explored for
event-driven systems~\cite{hopper,historia}. In these systems, like in actors,
control-flow may jump between many components, and not all jumps
are possible. In~\cite{historia}, message histories are used to capture control-flow feasible transitions in Android. While
control flow stays between the framework and the application in Android, it
spans an unbounded number of actors in our case, rather than finite components. One of our contributions
(relative to~\cite{historia}) is generalizing message histories to
address this scenario via message-segments. 
\subparagraph{\textbf{Recurrence Analysis}}
Recurrences have been used to find invariants in sequential programs~\cite{cra}: loop bodies are modeled as recurrence relations, whose solutions approximate the loop's behavior. We extend this reasoning to repetition in message histories rather than syntactic loops, using message-segments to summarize handlers, analogous to loop bodies. For simplicity we consider a restricted class of recurrences compared to sequential programs~\cite{nlcra}.

\section{Conclusion}

We have presented a novel symbolic abstraction for goal-directed verification of parametric actor systems. Leveraging precise local reasoning and actor materialization, our abstraction lets the analysis  focus on a system slice \emph{relevant} to refuting safety violations. We introduce actor-causality reduction to restrict exploration and message-segments to \emph{summarize} recurring unbounded communication, demonstrating their use to verify safety for protocols with affine communication.
This extends existing backwards-from-error
techniques~\cite{historia,hopper,cubicle} to automatically verify protocols previously only verified manually, namely Two-Phase Commit and One-Shot Consensus. We implement these techniques in \tool{}, evaluating it on a suite of
standard distributed protocols.

\bibliography{ref} 
\bibliographystyle{ACM-Reference-Format}

\clearpage
\appendix
\section{Actor Syntax}\label{sec:syntax-appendix}

The syntax of the actor language used throughout the paper can be found summarized in \autoref{fig:syntax}.

\begin{figure}\small
\begin{mathpar}
\text{distributed protocol}
\quad \topsys \bnfdef \mktopsys{\seq{\msgkind}}{\seq{\actor}}

\text{actor classes} \quad \actor \bnfdef \mkactor{\mach}{ \seq{\field} }{ \seq{\hand} }

\text{fields} \quad \field \bnfdef \mkfield{\var}{\val}

\text{variables} \quad \var, \varalt

\text{values} \quad \val

\text{message handlers} \quad \hand
  \bnfdef \mkonhand{\evnt}{ \seq{\var} }{\proc}
  
\text{procedures} \quad \proc \bnfdef \ctxnil \bnfalt \ctxcons{\proc}{\trans}

\text{transitions} \quad \trans
\bnfdef \mktrans{\prelbl\loc}{\instr}{\postlbl\loc}
\bnfalt \mktrans{\prelbl\loc}{\cmd}{\postlbl\loc}

\text{control-flow locations} \quad \loc \bnfdef \entryloc \bnfalt \exitloc \bnfalt \cdots

\text{instructions} \quad \instr \bnfdef \mkassign{\var}{\expr} \bnfalt \cdots
\qquad
\text{expressions} \quad \expr \bnfdef \var \bnfalt \cdots
\\

\text{actor commands} \quad \cmd
\bnfdef %
\mksend{\var}{\evnt}{\seq{\varalt}}
\\

\text{machine-type ids} \quad \mach \bnfdef \mainmach \bnfalt \cdots

\text{event-type ids} \quad \evnt \bnfdef \initevnt \bnfalt \idleevnt \bnfalt \cdots
\end{mathpar}
\caption{Abstract syntax of \dCSNH{}, a core language for reasoning about actor systems.}
\label{fig:syntax}
\end{figure}

A distributed protocol $\topsys$ is composed of a set of actor classes
$\seq{\actor}$ each defining the internal logic of a specific type of a
machine in the protocol. The Two-Phase Commit~\cite{twopc-grey06} protocol,
for example, is composed of two types of machines --- \emph{Participant} and \emph{Coordinator}. An \emph{actor class} ($\actor
\bnfdef \mkactor{\mach}{ \seq{\field} }{ \seq{\hand} }$) is composed of a
machine-type identifier $\mach$, a set of actor-local fields
$\seq{\field}$, and a set of handlers $\seq{\hand}$. An \emph{actor field}
($\field \bnfdef \mkfield{\var}{\val}$) is a mutable cell referenced with a
variable identifier $\var$ and initialized with a value $\val$. 
A \emph{message handler} ($\hand \bnfdef \mkonhand{\evnt}{ \seq{\var} }{\proc}$) is a handler for an event-type identifier $\evnt$ with formal parameters $\seq{\var}$ and a procedure body $\proc$.\footnote{
  Notation: We write $\seq{\square}$ to denote a sequence. We use
  script letters (e.g., $\mach$, $\evnt$) for identifiers of state-machine
  components and italic words (e.g., $\actor$, $\field$,
  $\hand$, $\proc$) for syntactic components. For convenience, we treat the
  protocol specification $\pi$ as a map: $\maplookup{\topsys}{\mach}$ gets the actor class $\actor$ corresponding to machine-type id $\mach$ (and similarly for $\maplookup{\actor}{\evnt}$) .}
A procedure ($\proc\ \bnfdef \ctxnil \bnfalt \ctxcons{\proc}{\trans}$) is a
set of transitions $\trans$, where each transition $\trans \bnfdef
\mktrans{\prelbl\loc}{\instr}{\postlbl\loc} \bnfalt
\mktrans{\prelbl\loc}{\cmd}{\postlbl\loc}$ is a control-flow edge from a
pre-location $\prelbl\loc$ to post-location $\postlbl\loc$, labeled with
either an imperative instruction $\instr$ or an actor command $\cmd$. Our
formalization is parametric on the instruction $\instr$ and expression
$\expr$ languages that allow writing to and reading from an actor's local
state (which we call its \emph{store}). A procedure has two special control-flow
locations $\entryloc$ for the entry location, which has no predecessors,
and $\exitloc$ for the exit location, which has no successors ($\loc
\bnfdef \entryloc \bnfalt \exitloc \bnfalt \cdots$).

Actor commands
$\cmd$
include sending a message to another actor with $\mksend{\var}{\evnt}{\seq{\varalt}}$,
where variable $\var$ is bound to the recipient address (i.e., a unique identifier for the recipient actor), $\evnt$ is an event-type identifier, and variables $\seq{\varalt}$ are bound to its payload.
We use the term \emph{message} for the package of the sender and recipient addresses with a event-type identifier $\evnt$ and a payload.
To focus on the core challenge of reasoning about asynchronous actor communication, we consider $\sendkw$ as the only actor command $\cmd$.
For simplicity, we consider actor commands in their A-normal form (i.e.,
their sub-expressions are flattened to variable reads). 

Unlike the standard actor calculus~\cite{agha}, we have no $\mknew{\mach}$ command to spawn a new actor of machine type $\mach$ --- we consider parametrized actor systems where only the main machine can spawn actors to set up the system.
As such, there is a special $\mainmach$ machine-type identifier for the main machine that starts the system ($\mach \bnfdef \mainmach \bnfalt \cdots$).
There are also two special event-type identifiers ($\evnt \bnfdef \initevnt
\bnfalt \idleevnt \bnfalt \cdots$): $\initevnt$ for a special event type
that can only be sent by the main machine to start the protocol execution
and $\idleevnt$ for a neutral, idle event type that has no handler. An
$\initevnt$ message can have different payloads but can only be received
once by each actor and before receiving any other messages.

\section{Proofs and Mechanization}\label{sec:mechanization}
Here we provide details on all theorems and proofs, as well as the mechanization of our semantics.

\subsection{Mechanization in Lean}{\label{sec:mech-details}

We have defined our concrete and abstract semantics in Lean, in particular the abstract semantics of sections 3 and 4. We have chosen not to mechanize a proof of the semantics of section 5, as this would require a mechanization of the semantics of \cite{cra}, which we deem outside the scope of this work.

The primary proof we have mechanized is that of \autoref{thm:causal-soundness}. This theorem is at the heart of our contributions as well as the prototype implementation, and we have proven it correct. In doing so, we have assumed lemmas on the concretization relation, which we believe to be reasonable.

In particular, we assume concretization to be monotonic, with the ability to extend a valuation $\valuation$ without breaking the concretization. This allows us to then show that messages and transitions in the concrete system which are not relevant to our goal-state can be ignored while preserving the relation.

From the proof of \autoref{thm:causal-soundness}, a proof of Lemma \ref{lem:local-soundness} can be derived, as it is a subset of the semantics mechanized.

We also mechanize Lemma \ref{thm:worklist}, as the final component of our termination proof. Combined with the termination proof, this lemma shows our algorithm will terminate for a particular, meaningful class of protocols.

In the remainder of this appendix, we restate each proof, note whether it is covered in the Lean mechanization, and discuss the proofs in detail of those not covered in Lean.

\subsection{Mechanized Theorems and Lemma Sketches}

\begin{lemma}[Subsumption Soundness]\label{lem:subsumption-soundness-appendix}
If $\jincluded{\abssysstate}{\abssysstate'}$ and $\concstate{\abssysstate}{\sysstate}$, then $\concstate{\abssysstate'}{\sysstate}$.
\end{lemma}

This lemma is stated as Lemma \ref{lem:subsumption-soundness} in \autoref{sec:actorlogic}. As this is used in materialization, it is not separately mechanized in Lean, but is handled in Lean in the cases materialization is used.

In \autoref{sec:parametric}, we give an operational semantics of entailment. Given a unification map $U$, we define entailment piecewise, using the subtraction algorithm~\cite{DBLP:conf/aplas/BerdineCO05} for the two separation logic components, linear logic entailment, and ordered logic entailment for the networks and histories respectively. These are all performed under formulas renamed by $U$. Finally, implication is checked for the pure constraints $\puredom$ and $\puredom'$.

The proof obligation is that given these individually sound entailment checks per component, their composition under $U$ is a sound entailment.
This follows because the components are otherwise separate, so no sub-entailment check needs to refer to any other component, and because UU
U is applied uniformly across all of them: a given symbolic value resolves to the same concrete referent in every component in which it appears. Componentwise soundness therefore composes into soundness of the whole entailment.

\begin{lemma}[Local Soundness of Backwards Symbolic Execution]\label{lem:local-soundness-appendix}
If
$\jstep[\topsys]{\prelbl\sysstate}{\postlbl\sysstate}$
and
$\jbackstep[\topsys]{\prelbl\abssysstate}{\postlbl\abssysstate}$
such that
$\concstate{\postlbl\abssysstate}{\postlbl\sysstate}$,
then $\concstate{\prelbl\abssysstate}{\prelbl\sysstate}$.
\end{lemma}

We state this lemma in \autoref{sec:actorlogic}, as Lemma \ref{lem:local-soundness}. This follows from a proof of \autoref{thm:causal-soundness}, as the causally-reduced semantics are a subset of these semantics.

We still give the intuition here, as it is quite straightforward. Using \infrule{Abs-Step-Handle}, \infrule{Abs-Step-Send}, and \infrule{Abs-Step-Instr}, if a concrete actor which is not summarized takes a step, the abstraction can precisely match it. If the actor is summarized, we apply \infrule{Materialize-Actor}, strengthening the abstract post-state to include the actor at the post-location of the concrete step. We can then  apply one of the above 3 rules, as the strengthened state contains a matching actor at the post location, in order to match the concrete behavior. The semantics of \autoref{sec:actorlogic} thus precisely simulate the concrete system.

\begin{theorem}[\thmcolor{Causal Soundness}]\label{thm:causal-soundness-appendix}
If
$\jstep[\topsys]{\prelbl\predtrace}{\postlbl\predtrace}$
and
$\jcausalbackstep[\topsys]{\prelbl\abssysstate}{\postlbl\abssysstate}$
such that
$\concstate{\postlbl\abssysstate}{\postlbl\predtrace}$,
then $\concstate{\prelbl\abssysstate}{\prelbl\predtrace}$.
\end{theorem}
We state this theorem in \autoref{sec:reflogic}, as \autoref{thm:causal-soundness}. We mechanize this proof in Lean, leveraging induction on the predecessor traces, then splitting cases for each concrete step possible. Intuitively, a causal step $\jcausalbackstep[\topsys]{\prelbl\abssysstate}{\postlbl\abssysstate}$
 may skip backward past a run of concrete transitions that are irrelevant to the goal state — the stuttering-simulation intuition discussed in \autoref{sec:reflogic}. The induction shows that whichever concrete rule fired, the extended trace $\postlbl\predtrace$ contains, among its predecessors, a state matching the abstract pre-state $\prelbl\abssysstate$: irrelevant transitions only add states to $\prelbl\predtrace$ without modifying actors and fields relevant to $\postlbl\abssysstate$.
This proof assumes monotonicity of concretization under valuation extension — that a valuation 
$\valuation$ can always be extended to a larger domain without breaking the relation (see \autoref{sec:mech-details} for discussion). This assumption lets the proof discharge irrelevant concrete steps without needing to track how they affect the valuation.

In the remainder of this appendix we provide full proof treatments of remaining theorems.

\subsection{Soundness of Message-Segments}\label{sec:soundness-app}

Here we provide a more detailed proof of soundness for the semantics of message-segments.

\begin{lemma}[Segment-Length Natural]\label{thm:seglenpos}
  Given $\mkmsgseg{\mach}{\tolbl\absaddr}{\evnt}{\symvar}$, $\seglen > 0$
\end{lemma}

Proof: By Construction of $\mkmsgseg{\mach}{\tolbl\absaddr}{\evnt}{\symvar}$.

  \begin{mathpar}
    \infer{\predtrace' = \predtracecons{\sysstate}}{
      last(\predtrace') = \sysstate
    }

    \infer{\predtrace' = \sysstate}{
       head(\predtrace') = \sysstate
    }

    \infer{\predtrace' = \predtracecons{\sysstate}}{
      head(\predtrace') = head(\predtrace)
      }
  \end{mathpar}

  \begin{lemma}[Trace-Equivalence]\label{thm:trace-eq}
    Let $\predtrace$ be a concrete trace of states, ending in the completion of handler $h$ for actor $\addr$ (reaching $\exitloc$), and $h$ is symmetric w.r.t. other handlers. Then, $\exists \predtrace'$, which is Mazurkiewicz-equivalent, such that 
 $head(\predtrace)=head(\predtrace')$ and $last(\predtrace) = last(\predtrace')$, where happens-before relations are preserved, and all invocations of $h$ received by $\addr$ are consecutive.
  \end{lemma}

  By the actor model assumption~\cite{agha}, the steps of $\predtrace$ can
be reordered to produce $\predtrace''$ in which all steps of any single
actor within a handler are sequential, preserving happens-before relations
and hence agreeing with $\predtrace$ on head and tail.

By $\mathit{symmetric}(h)$ (from hypothesis), the invocations of $h$ in $\predtrace''$
are independent of all other handlers, so $\predtrace''$ can be further
reordered to produce $\predtrace'$ grouping all invocations of $h$
consecutively, again preserving head and tail. $\qed$

\begin{theorem}[Soundness]

If
$\jstep[\topsys]{\prelbl\predtrace}{\postlbl\predtrace}$
and
$\jcausalbackstep[\topsys]{\prelbl\abssysstate}{\postlbl\abssysstate}$
such that
$\concstate{\postlbl\abssysstate}{\postlbl\predtrace}$,
then $\exists\predtrace'$ such that $symmetric(\prelbl\predtrace,\predtrace')$ and $\concstate{\prelbl\abssysstate}{\predtrace'}$.

\end{theorem}

This theorem is stated as \autoref{thm:segment} in \autoref{sec:parametric}.

We will proceed with proof by case analysis on $\jstep{\prelbl\predtrace}{\postlbl\predtrace}$ and
$\jcausalbackstep[\topsys]{\prelbl\abssysstate}{\postlbl\abssysstate}$
\\

\begin{itemize}
\item \textbf{Case \TirName{Step-Send}}: Let the command be $\mksend{\evnt}{\symvaralt}{\symvar}$. We split on if the command is in a segment handler for $\evnt$, or a standard handler.
  \begin{itemize}
  \item \textbf{Case Standard}: this step is modeled by \TirName{Causal-Step-Send}, following \autoref{thm:causal-soundness}.
    \item \textbf{Case Segment}: we have that $\evnt$ has a corresponding message-segment handler, which we denote $h$, with segment length $\seglen$. We have an abstract state $\postlbl\abssysstate$ which models $\postlbl\predtrace$, by hypothesis. We want to apply \TirName{Causal-Step-Segment-Send}, and must show all premises to hold.

Suppose there is no actor matching the abstract post state. By applying\\ \TirName{Causal-Materialize-Actor-Network}, we can produce a state $\abssysstate''$ which contains such an actor, so we will assume $\postlbl\abssysstate$ to have already done so.

Next, we must show $\postlbl\abssysstate$ to have a message-segment in its $\absnet$. If it does not, we can avoid materialization, and $\postlbl\abssysstate$ models $\prelbl\predtrace$ already, as no constraints change.

We thus can apply  \TirName{Causal-Step-Segment-Send}, and must show that the resultant state, $\prelbl\abssysstate$, models $\concstate{\prelbl\abssysstate}{\predtrace'}$, a trace in the same equivalence class as $\prelbl\predtrace$.

As in \TirName{Causal-Step-Send}, the arguments of the sending site are linked to that of the receiver directly, so $\prelbl\abssysstate$ models a state where those arguments are sent, as in $\prelbl\predtrace$. Unlike $\prelbl\predtrace$, the abstract state  $\prelbl\abssysstate$ models a state in which $k$ copies of the message are received in a row, without interleavings, while $\prelbl\predtrace$ may have arbitrary interleavings. By construction, a segment handler $h$ is only considered an option when it is independent/symmetric of all other operations, and so within the same trace equivalence class, there must exist some trace $\predtrace'$ which only differs by having each copy of the message received in a row.
  \end{itemize}

\item  \textbf{Case \TirName{Step-Handle}}:  Let the handler in question be for $\evnt$. We split on if the command is in a segment handler for $\evnt$, or a standard handler, just as in the previous case.

  \begin{itemize}
    \item \textbf{Case Standard}:  this step is modeled by \TirName{Causal-Step-Handle}, following \autoref{thm:causal-soundness}.
    \item \textbf{Case Segment}: we have that $\evnt$ has a corresponding message-segment handler, which we denote $h$, with segment length $\seglen$, by case analysis.
      We have an abstract state $\postlbl\abssysstate$ which models $\postlbl\predtrace$, by hypothesis. We want to apply \TirName{Causal-Step-Segment-Handle}, and must show all premises to hold, and that the pre-state models an equivalent trace.

Firstly, we consider the addition of a message-segment $\mkmsgseg{\mach}{\tolbl\absaddr}{\evnt}{\symvar}$. This corresponds to a sequence of concrete messages to $\tolbl\absaddr$ with argument $\symvar$ and event $\evnt$. We must show this covers all behaviors of the concrete system.

By hypothesis, we are in a message segment handler $h$. We only construct this handler if $symmetric(h)$ holds, meaning executions of $h$ can be reordered with respect to other event handlers.

Consider $\predtrace$: By Lemma \ref{thm:trace-eq}, we can produce a trace $\predtrace''$ which is equivalent but groups executions of individual handlers under the actor model assumption. By $symmetry(h)$, we can turn $\predtrace''$ into some $\predtrace'''$ such that $\predtrace'''$ differs by sequencing executions of $h$.

$\predtrace'''$ must contain a state $\sysstate''$ such that $\sysstate''$ is the first invocation of $h$: It may be the last state of $\predtrace'''$, which we know to invoke $h$ by hypothesis, or some further state. We shall construct a state $\abssysstate''$ which is modeled by this predecessor  $\sysstate''$. 

Suppose the concrete trace executes $h$ some number of times $\seglen'$. Thus,  $\sysstate''$ contains $\seglen'$ copies of the message in its network $\net$. If we can show $\seglen$, our segment length, to be satisfied by $\seglen'$, we have concretization shown for both the history and network.

We constrain $\seglen$ only using the affine protocol restriction. In particular, we require that the count of messages with $\evnt$ be at most the system size. We also require $\seglen$ to be positive.

The latter is true as alternatively, we have a contradiction in that the message could not have been sent, but by hypothesis we witness the handle step.

For the former, suppose this restriction prevents $\seglen$ from being modeled by $\seglen'$. We can then show a contradiction with the affine protocol. This can only occur when $\seglen'$ is larger than some upper bound on $\seglen$'s value. In our abstract history, this means more copies of $\evnt$ are witnessed than are possible, by affine-ness, and so the state is invalid.

Lastly, our affine restriction in $\puredom'$ must be considered. For this case, if $\seglen$ copies cannot be added to the network and satisfy affine-ness, then the concrete protocol must also not satisfy affine-ness, and the program is invalid by our initial restrictions. Otherwise, this state satisfies concretization of $\puredom'$.

With all propositions holding, we may take the step. Location and store concretization hold trivially, as the pre-state $\prelbl\abssysstate$ returns the actor to idle. Network, history, and constraint concretization follow as discussed above.
  \end{itemize}

\item \textbf{Case \TirName{Step-Instr}}: Let the command be $\mkassign{\symvar}{\expr}$. We split on if the command is in a segment handler for $\evnt$, or a standard handler. 
  \begin{itemize}
  \item \textbf{Case Standard}: We can apply \TirName{Causal-Step-Instr} in \autoref{thm:causal-soundness}
  \item \textbf{Case Segment}: we have that $\evnt$ has a corresponding message-segment handler, which we denote $h$, with segment length $\seglen$. Recall that we construct our message segment handlers according to Algorithm 4.14 in \cite{cra}. By Theorem 4.15 in \cite{cra}, we have that the body of the segment handler soundly overapproximates the behavior of the looped handler, and thus, we can construct $\abssysstate''$ by apply \TirName{Causal-Step-Instr}, such that it is modeled by $\predtrace$.
   \end{itemize}
\end{itemize} 

\subsection{Termination Proof}\label{sec:termination}

Now we will discuss our termination proof, which in principle relies on defining a ranking function to describe the growth of the worklist used in our algorithm to compute a fixpoint. We duplicate the causal graph and causal dependence definitions first, for ease of reference.

Consider an actor program $\topsys$. We define $\causalgraph{\topsys}$ to be a static, actor-insensitive abstraction describing the causal dependencies of a distributed protocol $\topsys$.
More precisely, $\causalgraph{\topsys}$ is a directed graph where the vertices are fields $\mkmachfld{\mach}{\var}$ or event types $\mkmachevnt{\mach}{\evnt}$ with the following kinds of edges:
\begin{description}\small
  \item[data dependence]
  $\causaltofrom{\mkmachfld{\mach}{\varalt}}{\mkmachfld{\mach}{\var}}$ if there is an assignment to variable $\var$ that is data dependent on $\varalt$ (i.e.,
  $\causaltofrom{\mkmachfld{\mach}{\varalt}}{\mkmachfld{\mach}{\var}} \in \causalgraph{\topsys}$ if
  $(\mktrans{-}{\mkassign{\var}{\expr}}{-}) \in \maplookup{\maplookup{\topsys}{\mach}}{\evnt}$ where $\varalt$ is in the free variables of expression $\expr$, for some event type $\evnt$).
  \item[control dependence]
  $\causaltofrom{\mkmachfld{\mach}{\varalt}}{\mkmachfld{\mach}{\var}}$ if there is an assignment to variable $\var$ that is control dependent on $\varalt$ (i.e.,
  $\causaltofrom{\mkmachfld{\mach}{\varalt}}{\mkmachfld{\mach}{\var}} \in \causalgraph{\topsys}$ if
  $\mktrans{\prelbl\loc}{\mkassign{\var}{-}}{-} \in \maplookup{\maplookup{\topsys}{\mach}}{\evnt}$ where a guard using $\varalt$ dominates location $\prelbl\loc$, for some event type $\evnt$).
  \item[send-control dependence] 
  $\causaltofrom{\mkmachevnt{\mach'}{\varalt}}{\mkmachevnt{\mach}{\evnt}}$ if there is a send of event $\evnt$ with a control dependence on variable $\varalt$ (i.e.,
  $\causaltofrom{\mkmachevnt{\mach}{\evnt}}{\mkmachevnt{\mach'}{\varalt}} \in \causalgraph{\topsys}$ if
  $\mktrans{\prelbl\loc}{ \mksend{\var}{\evnt}{-} }{-} \in \maplookup{\maplookup{\topsys}{\mach}}{\evnt'}$ 
  where $\var$ has machine type $\mach$ and a guard using $\varalt$ dominates location $\prelbl\loc$, for some event type $\evnt'$).
  \item[assign-handle dependence] 
  $\causaltofrom{\mkmachfld{\mach}{\evnt}}{\mkmachevnt{\mach}{\var}}$ if there is an assignment to variable $\var$ in a handler for event $\evnt$ (i.e., $\causaltofrom{\mkmachfld{\mach}{\evnt}}{\mkmachevnt{\mach}{\var}} \in \causalgraph{\topsys}$ if $\mktrans{-}{\mkassign{\var}{-}}{-} \in \maplookup{\maplookup{\topsys}{\mach}}{\evnt}$).
  \item[send-handle dependence] 
  $\causaltofrom{\mkmachevnt{\mach'}{\evnt'}}{\mkmachevnt{\mach}{\evnt}}$ if there is a send of an event $\evnt$ in a handler for event $\evnt'$ (i.e., $\causaltofrom{\mkmachevnt{\mach'}{\evnt'}}{\mkmachevnt{\mach}{\evnt}} \in \causalgraph{\topsys}$ if $\mktrans{-}{ \mksend{\var}{\evnt}{-} }{-} \in \maplookup{\maplookup{\topsys}{\mach'}}{\evnt'}$ where $\var$ has machine type $\mach$).
\end{description}
We say that a causal-dependence graph is an actor-insensitive abstraction because it does not distinguish between different actors of the same machine type --- as we see in the above definition --- but otherwise follows the causal dependencies described in \autoref{sec:causal-cycles}.

\begin{figure}
\begin{align*}
         \initstate{\abssysloc} \iff & \forall \absaddr \in dom(\abssysloc). \exists \mach.
\abssysloc(\absaddr) = \mkmachsttevntloc{\mach}{\initstt}{\idleevnt}{\exitloc}\\
         \netcolor{\absnet}\initstate{\purecolor{\puredom}} \iff & \forall \symvar \in dom(\purecolor{\puredom}) . \exists \absmsg \in \netcolor{\absnet} .
        \absmsg = \mkmsg{\fromlbl\absaddr}{\tolbl\absaddr}{\evnt}{\symvar}\\
        \initstate{\netcolor{\absnet}} \iff & \forall \absmsg \in \netcolor{\absnet} .
        \absmsg = \mkmsg{\fromlbl\absaddr}{\tolbl\absaddr}{\initevnt}{\symvar}\\ \initstate{\mkabssysstate{\abssysloc}{\abssysstore}{\absnet}{\abshist}} \iff & \initstate{\abssysloc} \text{ and } \netcolor{\absnet}\initstate{\purecolor{\puredom}} \text{ and } \initstate{\netcolor{\absnet}}
\end{align*}

\caption{Initial States: To be initial, a state must only have actors at the idle location, messages for the initial event in the network, and all symbolic variables must be arguments of these messages.}
\label{fig:init}
\vspace{-.5cm}
\end{figure}

we define an initial state in our abstraction, via the judgment $\initstate{\abssysstate}$, in \autoref{fig:init}. A state is initial if each actor in $\abssysloc$ is at the idle location, if every symbolic variable in $\puredom$ is an argument of a message in the network, and every message in $\absnet$ is for $\initevnt$.

With all these components, we can now prove that the analysis terminates, via a conservative over-approximation of a ranking function for the analysis. To construct this argument, we need a couple more auxiliary definitions.

For each message handler $h$, we assume that a maximum number of steps $H$ is needed by the analysis to reach the entry location from the exit location. As each of our given rules require only a single step for actor commands, this assumption is on our underlying imperative semantics, that the over-approximation given is itself terminating. We define $D$ to be the depth of $\causalgraph{\topsys}$. Lastly, we let $W$ denote the maximum in-degree of any vertex in $\causalgraph{\topsys}$.

Termination of the analysis is dependent on the structure of both graphs.  We require  $\causalgraph{\topsys}$ to be acyclic: otherwise, leveraging \texttt{Abs-StepSend} and \texttt{Abs-Step-Materialize-Net} the analysis can materialize infinitely many new actors to consider. Additionally,  when witnessing an assignment via\texttt{Abs-Step-Instr}, a relevant field may become relevant infinitely often, preventing termination.

\begin{theorem}[Termination]\label{thm:termination-appendix}
Assume the causal-dependence graph $\causalgraph{\topsys}$ of a protocol $\topsys$ is acyclic. Then, for any singleton abstract distributed state $\abssysstate$, there exists a natural-number bound $n$ such that $\jcausalbackstepn[n]{\abssysstate'}{\abssysstate}$ leads to some singleton abstract distributed state $\abssysstate'$ where
$\jinitbot{\abssysstate'}$.
\end{theorem}

We state this theorem in \autoref{sec:parametric}, as \autoref{thm:termination}.

\textbf{Proof}: Consider an abstract state $\abssysstate$. If $\jinitbot{\abssysstate}$, let $n=0, \abssysstate'=\abssysstate$.\\
Else, let $\abssysstate = \mkabssysstate{\abssysloc}{\abssysstore}{\absnet}{\abshist}$.
We have $\neg(\initstate{\abssysloc} \land \absnet\initstate{\puredom} \land \initstate{\absnet} )$ by inversion.

We must show that in a finite number of steps, a state $\abssysstate'$ can be reached which is initial, or is known to be unreachable.

A state is unreachable if its components are contradictory, or it has no predecessors and is not init. A state is contradictory if the buffer contains a message which cannot be sent, or a pure constraint over fields is not satisfiable by any valuation $\theta$. 

Let $i, k,j$ be the number of non-initial actors, unconstrained fields, and non-initial messages, respectively. We define the sum of remaining work for $\abssysstate$ to be $i+j+k$, as this many elements will be added to the worklist next (branching for each $W$ times).

To resolve a non-initial actor (bring it to the initial location), at most $1$ messages will be added to network, and at most $D$ fields become relevant. To resolve a non-initial message, at most $D$ new actors are materialized, with $D$ fields becoming relevant. For a field, at most $D$ messages, and $D$ actors are added.

Materializing these field constraints or messages may cause contradiction, or be impossible. If so, we have proven that the state is $\bot$, and are done.

Otherwise, we continue executing our semantics.
Thus, we go from $i+j+k$ to $(Dj+Dk) + (Di+Dk) + (Dj+Di)$, or $2D(i+j+k)$ (one copy of $i$ becomes $Di$, to over-approximate and yield a more conservative bound. We may repeat this process with our new values of $i,j,k$, noting that instead of $D$, we use $D-1$, as everything newly added must be one step closer to the root of $\causalgraph{\topsys}$. Thus, the next term is $4D(D-1)(i+j+k)$. We note that at each ``step'', we are actually taking $H$ steps, to resolve entire handlers, and may branch $W$ times. We can define this to be a sequence $s(n) = 2^n(D!/(D-n)!)(i+j+k)HW^n$. As $H(i+j+k)$ is a constant for our given state, we elide it. The sequence $s(n)$ tells us how much work is added to our worklist by the analysis at the $n-th$ iteration, and so it suffices to show $s(n)$ eventually becomes 0 for all $n$ after some point, to show termination. We call this the \texttt{Worklist-Lemma}, and by applying this lemma, we guarantee termination for the left side of our disjunction. $\qed$

\begin{lemma}[\thmcolor{Worklist-Lemma}]\label{thm:worklist}
$\forall D \in \mathbb{N}, \exists n \in \mathbb{N} . \forall m. m > n \implies  s(m) = 0$.
\end{lemma}

Though mechanized, we state Worklist-Lemma here since it depends on definitions introduced above, in particular the sequence $s(n)$.

The proof follows naturally from the sequence $s(m) = 2^m(D!/(D-m)!)(i+j+k)HW^m$. The coefficient $(D!/(D-m)!)$ is a falling factorial, and when $D<m$, this term reduces to $0$, making the sequence $0$.

\section{Symmetry and Trace Equivalence}\label{sec:symmetry-appendix}

The use of message-segments requires additional restrictions on trace equivalence, discussed in \autoref{sec:parametric}. Here, we provide a formalization of the intuition discussed, and further details on relaxations.

Mazurkiewicz-equivalence\cite{mazurkiewicz} is a property of two traces $\predtrace,\predtrace'$ which states the two are equal modulo the ordering of independent events. Two events are said to be independent if their relative ordering is irrelevant to the state (i.e. they commute). Formally, given a state $\sysstate$, with history $\hist$, two messages $\msg_1,\msg_2$ are independent if the state reached by handling $\msg_1$ then $\msg_2$ is equivalent to the one reached by handling the inverse order. The two final states differ \emph{only} in their message history, which records the order.

We can extend this notion of equivalence, defining what we refer to as handler symmetry. Formally, a handler $h$ is said to be symmetric if the effects of this handler are independent of all other handlers for the actor type $\mach$. A key distinction to make is that it is not required to be independent of the distinguished initial event $\initevnt$, as it cannot be reordered under our parametric system assumption.

Below, we formalize this rule in our abstraction. Despite its verbosity, it encodes precisely the aforementioned intuition.

\begin{mathpar}
\infer[Symmetry]{
        \forall \evnt' \neq\evnt \and \abssysstate = \mkabssysstate{\abssysloc}{\abssysstore}{\absnet}{\abshist} \and
\abssysstate' = \mkabssysstate[\puredom']{\abssysloc}{\abssysstore'}{\absmsg\netsep\absnet}{\abshisthyp{\absmsg}{\abshist}} \and
        \jcausalbackstepn{\abssysstate'}{ }{\abssysstate} \\
\abssysstate'' = \mkabssysstate[\puredom'']{\abssysloc}{\abssysstore''}{\absmsg'\netsep\absnet}{\abshisthyp{\absmsg'}{\abshist}} \and 
\jcausalbackstepn{\abssysstate''}{ }{\abssysstate} \and
\absmsg = \mkmsg{\absaddr_1}{\absaddr}{\evnt}{\symvar} \and
\absmsg' = \mkmsg{\absaddr_2}{\absaddr}{\evnt'}{\symvaralt} \\
\jcausalbackstepn{\mkabssysstate[\puredom''']{\abssysloc}{\abssysstore'''}{\absmsg\netsep\absmsg'\netsep\absnet}{\abshisthyp{\absmsg}{\abshisthyp{\absmsg'}{\abshist}}}}{ }{\abssysstate''} \\
\jcausalbackstepn{\mkabssysstate[\puredom''']{\abssysloc}{\abssysstore'''}{\absmsg'\netsep\absmsg\netsep\absnet}{\abshisthyp{\absmsg'}{\abshisthyp{\absmsg}{\abshist}}}}{ }{\abssysstate'} \and h = \mkonhand{\evnt}{\seq{\var}}{\proc}
}{
        symmetric(h)
}
\end{mathpar}

While this notion of symmetry is sufficient for our usage, it is quite a strong restriction, and not neccessary. As actors do not share memory, if two handlers $h_1$ and $h_2$ are not always independent, but they are being applied to distinct actors $\absaddr_1$ and $\absaddr_2$, the specific invocations are independent. This weakened condition is more broadly applicable, but must be checked during analysis, and cannot be done as a pre-analysis step.

\end{document}